\documentclass[letterpaper,12pt,titlepage,oneside,]{book}

\newcommand{\href}[1]{#1} 

\usepackage{ifthen}
\newboolean{PrintVersion}
\setboolean{PrintVersion}{false} 
\usepackage{ifdraft}

\usepackage{amsmath,amssymb,amstext} 
\usepackage[pdftex]{graphicx} 

\makeatletter
\def\input@path{{figures/}}
\makeatother
\graphicspath{{figures/}}

\usepackage[pdftex,pagebackref=true]{hyperref} 
\hypersetup{
    plainpages=false,       
    unicode=true,          
    pdftoolbar=true,        
    pdfmenubar=true,        
    pdffitwindow=false,     
    pdfstartview={FitH},    
    pdfnewwindow=true,      
    colorlinks=true,        
    linkcolor=blue,         
    citecolor=green,        
    filecolor=magenta,      
    urlcolor=cyan           
}
\ifthenelse{\boolean{PrintVersion}}{   
\hypersetup{	
    citecolor=black,%
    filecolor=black,%
    linkcolor=black,%
    urlcolor=black}
}{} 

\hypersetup{final}

\usepackage[nohints]{minitoc}
\dominitoc

\usepackage{amsfonts}
\usepackage{amsmath}
\usepackage{latexsym}
\usepackage{stmaryrd}
\usepackage{mathtools}
\usepackage{dsfont}
\usepackage{makeidx}
\usepackage{bibentry}
\usepackage{enumitem}
\usepackage{relsize}
\usepackage{makeidx}
\usepackage[totoc]{idxlayout}

\usepackage{amsthm}
\newtheorem{theorem}{Theorem}[chapter]
\newtheorem{lemma}[theorem]{Lemma}
\newtheorem{claim}[theorem]{Claim}
\newtheorem{prop}[theorem]{Proposition}

\theoremstyle{definition}

\newtheorem{example}[theorem]{Example}

\newtheorem{question}[theorem]{Question}

\newcommand{\tinyspace}{\mspace{1mu}}
\newcommand{\microspace}{\mspace{0.5mu}}

\newcommand{\tr}{\operatorname{Tr}}

\newcommand{\rank}{\operatorname{rank}}
\renewcommand{\int}{\operatorname{int}}
\renewcommand{\vec}{\operatorname{vec}}

\newcommand{\abs}[1]{\left\lvert #1 \right\rvert}

\newcommand{\Bigabs}[1]{\Bigl\lvert #1 \Bigr\rvert}

\newcommand{\ip}[2]{\langle #1 , #2\rangle}
\newcommand{\bigip}[2]{\bigl\langle #1, #2 \bigr\rangle}
\newcommand{\Bigip}[2]{\Bigl\langle #1, #2 \Bigr\rangle}
\newcommand{\biggip}[2]{\biggl\langle #1, #2 \biggr\rangle}

\newcommand{\norm}[1]{\lVert\tinyspace #1 \tinyspace\rVert}
\newcommand{\bignorm}[1]{\bigl\lVert\tinyspace #1 \tinyspace\bigr\rVert}

\newcommand{\biggnorm}[1]{\biggl\lVert\tinyspace #1 \tinyspace\biggr\rVert}

\newcommand{\ket}[1]{
  \lvert\microspace #1 \microspace \rangle}

\newcommand{\bra}[1]{
  \langle\microspace #1 \microspace \rvert}

\def\I{\mathds{1}}
\def\BB84{\mathsf{BB84}}
\def\CHSH{\mathsf{CHSH}}

\newcommand{\setft}[1]{\mathrm{#1}}
\newcommand{\Density}{\setft{D}}
\newcommand{\Pos}{\setft{Pos}}
\newcommand{\Channel}{\setft{C}}
\newcommand{\Proj}{\setft{Proj}}

\newcommand{\Unitary}{\setft{U}}
\newcommand{\Herm}{\setft{Herm}}
\newcommand{\Lin}{\setft{L}}
\newcommand{\Trans}{\setft{T}}
\newcommand{\Sep}{\setft{Sep}}

\newcommand{\SepD}{\setft{SepD}}

\def\complex{\mathbb{C}}
\def\real{\mathbb{R}}
\def\natural{\mathbb{N}}
\def\integer{\mathbb{Z}}

\def\ns{\textnormal{ns}}
\def\im{\textnormal{im}}

\newenvironment{mylist}[1]{\begin{list}{}{
	\setlength{\leftmargin}{#1}
	\setlength{\rightmargin}{0mm}
	\setlength{\labelsep}{2mm}
	\setlength{\labelwidth}{8mm}
	\setlength{\itemsep}{0mm}}}
	{\end{list}}

\def\X{\mathcal{X}}
\def\Y{\mathcal{Y}}
\def\Z{\mathcal{Z}}
\def\H{\mathcal{H}}
\def\W{\mathcal{W}}
\def\A{\mathcal{A}}
\def\B{\mathcal{B}}
\def\V{\mathcal{V}}
\def\U{\mathcal{U}}
\def\C{\mathcal{C}}

\def\R{\mathcal{R}}
\def\Q{\mathcal{Q}}
\def\P{\mathcal{P}}
\def\S{\mathcal{S}}

\def\NS{\mathcal{NS}}
\def\L{\mathcal{L}}

\newcommand{\reg}[1]{\mathsf{#1}}

\def \GammaA{\Gamma_{\reg{A}}}
\def \GammaB{\Gamma_{\reg{B}}}
\def \SigmaA{\Sigma_{\reg{A}}}
\def \SigmaB{\Sigma_{\reg{B}}}

\usepackage{xcolor,colortbl}
\definecolor{Gray}{gray}{0.90}

\newcommand{\eqnref}[1]{\hyperref[#1]{{(\ref*{#1})}}}
\newcommand{\thmref}[1]{\hyperref[#1]{{Theorem~\ref*{#1}}}}
\newcommand{\lemref}[1]{\hyperref[#1]{{Lemma~\ref*{#1}}}}
\newcommand{\corref}[1]{\hyperref[#1]{{Corollary~\ref*{#1}}}}
\newcommand{\defref}[1]{\hyperref[#1]{{Definition~\ref*{#1}}}}
\newcommand{\secref}[1]{\hyperref[#1]{{Section~\ref*{#1}}}}
\newcommand{\chapref}[1]{\hyperref[#1]{{Chapter~\ref*{#1}}}}
\newcommand{\figref}[1]{\hyperref[#1]{{Figure~\ref*{#1}}}}
\newcommand{\tabref}[1]{\hyperref[#1]{{Table~\ref*{#1}}}}
\newcommand{\remref}[1]{\hyperref[#1]{{Remark~\ref*{#1}}}}
\newcommand{\appref}[1]{\hyperref[#1]{{Appendix~\ref*{#1}}}}
\newcommand{\claimref}[1]{\hyperref[#1]{{Claim~\ref*{#1}}}}
\newcommand{\propref}[1]{\hyperref[#1]{{Proposition~\ref*{#1}}}}
\newcommand{\exampleref}[1]{\hyperref[#1]{{Example~\ref*{#1}}}}
\newcommand{\conjref}[1]{\hyperref[#1]{{Conjecture~\ref*{#1}}}}

\makeatletter
\providecommand*{\cupdot}{%
  \mathbin{%
    \mathpalette\@cupdot{}%
  }%
}
\newcommand*{\@cupdot}[2]{%
  \ooalign{%
    $\m@th#1\cup$\cr
    \hidewidth$\m@th#1\cdot$\hidewidth
  }%
}
\makeatother

\makeindex

\let\origdoublepage\cleardoublepage
\newcommand{\clearemptydoublepage}{%
  \clearpage{\pagestyle{empty}\origdoublepage}}
\let\cleardoublepage\clearemptydoublepage

\begin{document}

%
\pagestyle{empty}
\pagenumbering{roman}

\begin{titlepage}
        \begin{center}
        \vspace*{1.0cm}

        \Huge
        {\bf Extended nonlocal games}

        \vspace*{1.0cm}

        \normalsize
        by \\

        \vspace*{1.0cm}

        \Large
        Vincent Russo \\

        \vspace*{3.0cm}

        \normalsize
        A thesis \\
        presented to the University of Waterloo \\ 
        in fulfillment of the \\
        thesis requirement for the degree of \\
        Doctor of Philosophy \\
        in \\
        Computer Science \\

        \vspace*{2.0cm}

        Waterloo, Ontario, Canada, 2017 \\
		\copyright\ Vincent Russo 2017 
        \vspace*{1.0cm}

        \end{center}
\end{titlepage}

\noindent\textbf{Examining Committee Membership}

\noindent The following served on the Examining Committee for this thesis. The decision of the Examining Committee is by majority vote. 

\noindent\begin{tabular}{@{}p{6.6cm}@{\hspace{5cm}}p{6cm}@{}}
\bfseries External Examiner & STEPHANIE WEHNER \par Professor \\[4ex]
\bfseries Supervisors & JOHN WATROUS \par Professor \par MICHELE MOSCA \par Professor \\[4ex]
\bfseries Internal Members & RICHARD CLEVE \par Professor \par DEBBIE LEUNG \par Professor \\[4ex]
\bfseries Internal-external Member & VERN PAULSEN \par Professor \\[4ex]
\bfseries Defence Chair & SUE HORTON \par Professor 
\end{tabular}


\pagestyle{plain}
\setcounter{page}{2}

\cleardoublepage 

  \noindent
I hereby declare that I am the sole author of this thesis. This is a true copy of the thesis, including any required final revisions, as accepted by my examiners.

  \bigskip
  
  \noindent
I understand that my thesis may be made electronically available to the public.

\cleardoublepage


\begin{center}\textbf{Abstract}\end{center}

The notions of \emph{entanglement} and \emph{nonlocality} are among the most striking ingredients found in quantum information theory. One tool to better understand these notions is the model of \emph{nonlocal games}; a mathematical framework that abstractly models a physical system. The simplest instance of a nonlocal game involves two players, Alice and Bob, who are not allowed to communicate with each other once the game has started and who play cooperatively against an adversary referred to as the referee. 


The focus of this thesis is a class of games called \emph{extended nonlocal games}, of which nonlocal games are a subset. In an extended nonlocal game, the players initially share a tripartite state \emph{with the referee}. In such games, the winning conditions for Alice and Bob may depend on outcomes of measurements made by the referee, on its part of the shared quantum state, in addition to Alice and Bob's answers to the questions sent by the referee. 

We build up the framework for extended nonlocal games and study their properties and how they relate to nonlocal games. In doing so, we study the types of \emph{strategies} that Alice and Bob may adopt in such a game. For instance, we refer to strategies where Alice and Bob use quantum resources as \emph{standard quantum strategies} and strategies where there is an absence of entanglement as an \emph{unentangled strategy}. These formulations of strategies are purposefully reminiscent of the respective quantum and classical strategies that Alice and Bob use in a nonlocal game, and we also consider other types of strategies with a similar correspondence for the class of extended nonlocal games. 

We consider the \emph{value} of an extended nonlocal game when Alice and Bob apply a particular strategy, again in a similar manner to the class of nonlocal games. Unlike computing the unentangled value where tractable algorithms exist, directly computing the standard quantum value of an extended nonlocal game is an intractable problem. We introduce a technique that allows one to place upper bounds on the standard quantum value of an extended nonlocal game. Our technique is a generalization of what we refer to as the \emph{QC hierarchy} which was studied independently in works by Doherty, Liang, Toner, and Wehner as well as by Navascu\'{e}s, Pironio, and Ac\'{i}n. This technique yields an upper bound approximation for the quantum value of a nonlocal game.

We also consider the question of whether or not the dimensionality of the state that Alice and Bob share as part of their standard quantum strategy makes any difference in how well they can play the game. That is, does there exist an extended nonlocal game where Alice and Bob can win with a higher probability if they share a state where the dimension is infinite? We answer this question in the affirmative and provide a specific example of an extended nonlocal game that exhibits this behavior.   

We study a type of extended nonlocal game referred to as a \emph{monogamy-of-entanglement game}, introduced by Tomamichel, Fehr, Kaniewski, and Wehner, and present a number of new results for this class of game. Specifically, we consider how the standard quantum value and unentangled value of these games relate to each other. We find that for certain classes of monogamy-of-entanglement games, Alice and Bob stand to gain no benefit in using a standard quantum strategy over an unentangled strategy, that is, they perform just as well without making use of entanglement in their strategy. However, we show that there does exist a monogamy-of-entanglement game in which Alice and Bob do perform strictly better if they make use of a standard quantum strategy. We also analyze the \emph{parallel repetition} of monogamy-of-entanglement games; the study of how a game performs when there are multiple instances of the game played independently. We find that certain classes of monogamy-of-entanglement games obey \emph{strong parallel repetition}. In contrast, when Alice and Bob use a non-signaling strategy in a monogamy-of-entanglement game, we find that strong parallel repetition is not obeyed.

\cleardoublepage
\newpage


\begin{center}\textbf{Acknowledgements}\end{center}

I am greatly indebted to my advisors John Watrous and Michele Mosca for their guidance throughout the course of my studies. John is an incredible supervisor and one that I have been exceptionally lucky to have had the pleasure of working with. John's attention to detail and sense of humour has greatly impacted my own approach to science and life in general. I am truly humbled by John's command of mathematical rigour and writing clarity. I am also incredibly grateful to Mike, for including me in the quantum circuits group and enabling me to take part in internships.      

Gratitude is also due to professors Richard Cleve, Debbie Leung, Vern Paulsen, and Stephanie Wehner for taking their valuable time to serve on my defence committee. I thank them greatly for their input. 

Throughout my studies, I've been lucky enough to work with some outstanding students, post-doctoral researchers, and professors. Many thanks are due to Nathaniel Johnston, Rajat Mittal, Matthew Pusey, Jamie Sikora, William Slofstra, Thomas Vidick, and others who were always willing to discuss interesting ideas and explain abstract concepts to me. 

My time at IQC and in Waterloo has been full of great people and experiences. I wish to thank Sascha Agne, Srinivasan Arunachalam, Alessandro Cosentino, Arnaud Carignan-Dugas, Maria Kieferova, Robin Kothari, Anirudh Krishna, Vinayak Pathak, Dan Puzzuoli, Yuval Sanders, Basil Singer, Marco Shum, Zak Webb, and the students and faculty of IQC for making my time here incredibly enjoyable. I sincerely hope that we keep in touch as our journeys continue. Thanks are also due to the excellent administrative support of IQC and DC for always being able to quickly resolve any technical issues in running software for experiments.

To my ``urban planning'' circle of friends, thank you for making me an honorary member of the group, and accepting me even though I'm in computer science! You have been my strongest support network and my best friends in Waterloo. I will sorely miss our many political conversations and random adventures. To my musically inclined friends Jaden Hellmann, Alexandar Smith, Will Towns, and Cody Veal, our jam sessions and miscellaneous discussions were a welcome creative distraction from research. 

To my friends back home in Michigan, I thank you for understanding my absence. Thanks to Kenny G., Sara Gilhooly, Alex, Nick, and Rachel Marowsky, Mike Sanderson, Ryan Seiler, Joe Sousa, Ryan Trainor, and Matt Wolford. Whenever I've been back to visit, I was always warmly received. 

I also extend gratitude toward the hospitality received during my internships. I thank BBN Raytheon, specifically Richard Lazarus, Andrei Lapets, and Marcus da Silva for allowing me to contribute to many interesting projects. 

I cannot express enough gratitude toward my family for their encouragement, love, and support. My brothers Joey and Matthew, and my sister Theresa and her husband Colin. I also thank Beth Russo and Lauren Kisic, my Aunt Cathy, Uncle Dan, Uncle John, Uncle Steve, and Aunt Marie for a lifetime of support. A sincere and heartfelt thanks are due to my parents James and Marjorie Russo, who have always nurtured and encouraged my interests, whatever they may have been, and for their seemingly endless wells of love and support. 

Lastly, I thank Paulina Rodriguez. Your support, love, and encouragement throughout this journey cannot be overstated. This document is as much a part of me as it is a part of you. I love you. 

For all of those who I did not mention by name, please accept my sincerest apologies, and know that this document is a testament to your encouragement and support. Thank you all.

\cleardoublepage
\newpage


\begin{center}\textbf{Dedication}\end{center}

This thesis is dedicated to those who have shaped my life, but no longer walk with me through it. My grandparents Anthony and Jean Russo, Ben Benton, my Aunt Alice, and my second moms, Denise Marowsky and Debbie Gilhooly.  

\cleardoublepage
\newpage

\renewcommand\contentsname{Table of Contents}
\tableofcontents
\cleardoublepage
\phantomsection

\addcontentsline{toc}{chapter}{List of Tables}
\listoftables
\cleardoublepage
\phantomsection		

\addcontentsline{toc}{chapter}{List of Figures}
\listoffigures
\cleardoublepage
\phantomsection		
\pagenumbering{arabic}


\setcounter{mtc}{2}

\chapter{Introduction}
\label{chap:introduction}

The model of two player games has served an important role in developing our understanding of theoretical computer science and quantum information. In such a game, we consider the players, referred to as \emph{Alice} and \emph{Bob}, who are not allowed to communicate to each other once the game begins, and who play cooperatively against a party referred to as the \emph{referee}. The game begins when the referee asks questions to Alice and Bob to which they must respond. When Alice and Bob send back the responses to the referee, the referee evaluates the questions and answers against a criterion that is publicly known to the referee, Alice, and Bob that determines what constitutes a winning or losing outcome. 

A primary challenge that arises when studying these games is to determine the maximum probability with which Alice and Bob are able to achieve a winning outcome. This probability is highly dependent on the type of \emph{strategy} that Alice and Bob use in the game. Before the game begins, Alice and Bob are free to communicate with each other and decide on the type of strategy they will use. 

A \emph{classical strategy} is one in which Alice and Bob decide on a deterministic mapping of outputs for every possible combination of inputs they will receive in the game. The corresponding maximum probability achieved when Alice and Bob employ a classical strategy is referred to as the \emph{classical value} of the game. 

Another type of strategy called a \emph{quantum strategy} is one in which Alice and Bob are allowed to use nonlocal resources. This type of strategy may involve Alice and Bob sharing an arbitrary entangled state prior to the start of the game along with sets of measurements that they may apply to their portions of the state after they each receive questions from the referee. The corresponding maximum probability achieved when Alice and Bob use a quantum strategy is referred to as the \emph{quantum value} of the game. 

For certain games, the probability that Alice and Bob obtain a winning outcome is higher if they use a quantum strategy as opposed to a classical one. This striking separation is one primary motivation to study nonlocal games, as it provides examples of tasks that benefit from the manipulation of quantum information. Indeed, the model of nonlocal games have been widely studied, especially in recent years~\cite{Cleve2004, Brassard2005, Cleve2008, Doherty2008,Kempe2010,Kempe2010a,Kempe2011,Junge2011a,Buhrman2013,
Regev2013,Dinur2013,Vidick2013,Cleve2014}. 

The ability to calculate the quantum value for an arbitrary nonlocal game is a highly non-trivial task. Indeed, the quantum value is only known in special cases for certain nonlocal games. For an arbitrary nonlocal game, there exist approaches that place upper and lower bounds on the quantum value. One such approach (that we refer to as the QC hierarchy as done in~\cite{Coudron2015} and was introduced in~\cite{Doherty2008, Navascues2007}), is implemented as a hierarchy of optimization problems, referred to as \emph{semidefinite programs}, which are optimization problems where the constraints are semidefinite. Convergence is guaranteed from the QC hierarchy, yet it may be intractable to compute. The lower bound approach is also calculated using the technique of semidefinite programming~\cite{Liang2007}. While this method is more efficient to carry out, it does not guarantee convergence to the quantum value (although in certain cases, it is attained). 

In a nonlocal game, the referee is only responsible for sending questions, receiving answers, and evaluating whether the selection of questions and respective answers yields a winning or losing outcome. In this thesis, we consider a generalization of the nonlocal game model where the referee is provided with part of a quantum system prepared by Alice and Bob, and in addition, also has sets of measurements that he may apply to his portion of the quantum system to determine the outcome of the game. This type of game is referred to as an \emph{extended nonlocal game}. Extended nonlocal games constitute a wider class of games of which nonlocal games are a subset. For instance, an extended nonlocal game where the dimension of the quantum system held by the referee is one-dimensional is precisely a nonlocal game. \emph{Monogamy-of-entanglement games} are a special type of extended nonlocal game introduced in~\cite{Tomamichel2013} that has been studied with respect to the problem of position-based cryptography. 


\section{Summary of the results}
In addition to introducing the model of extended nonlocal games, we prove the following results:

\begin{itemize}

	\item We prove that there exists a class of extended nonlocal game for which no finite-dimensional quantum strategy can be optimal. This result further implies the existence of a tripartite steering inequality for which an infinite-dimensional quantum state is required in order to achieve maximal violation. 

	\item We generalize the QC hierarchy, a technique for providing upper bounds on nonlocal games, to the case of extended nonlocal games. We also present a method based on the see-saw algorithm of Liang and Doherty~\cite{Liang2007} that provides lower bounds on the class of extended nonlocal games. 	


	\item We present a number of results about the class of \emph{monogamy-of-entanglement games}, which are a specific type of extended nonlocal game. Specifically, we show that: 
	
		\begin{itemize}
			\item Monogamy-of-entanglement games obey strong parallel repetition when the size of the question set has 2 elements and the size of the answer set is arbitrary, and the sets of measurements used by the referee are projective. 

			\item Monogamy-of-entanglement games do not obey strong parallel repetition when the players use non-signaling strategies. 

			\item We present a class of monogamy-of-entanglement games where the size of the question set has 2 elements and the size of the answer set is arbitrary where Alice and Bob can always achieve the quantum value of such a game by using a strategy that does not require them to store quantum information. 

			\item There exists a monogamy-of-entanglement game in which the size of the question set has 4 elements and the answer set has 3 elements, for which Alice and Bob must store quantum information to play optimally. 
		\end{itemize}

\end{itemize}

\section{Overview}
We assume familiarity with the basic notions of quantum computation and quantum information as can be found in~\cite{Nielsen2001}. It may also be helpful to have a familiarity with the terminology and mathematics in the first two chapters of~\cite{Watrous2015}, although we shall also attempt a self-contained presentation of the necessary tools needed to understand the content herein. Throughout this thesis, we also make frequent use of the mathematical tool of semidefinite programming. Supplementary resources for the interested reader can be found in lecture 7 of~\cite{Watrous2004} as well as~\cite{Boyd2004}. 

In Chapter~\ref{chap:preliminaries}, we review the basics of quantum information, nonlocal games, and relevant notation that will be used in the remainder of this thesis. 


In Chapter~\ref{chap:extended_nonlocal_games}, we introduce the model of extended nonlocal games that is built upon the model of nonlocal games. 

In Chapter~\ref{chap:infinite_entanglement}, we present an analysis of certain properties of the extended nonlocal game model and give an example of an extended nonlocal game for which no finite-dimensional quantum strategy can be optimal. 

In Chapter~\ref{chap:extended_npa_hierarchy}, we present a method that provides upper and lower bounds on the value of an extended nonlocal game. 

In Chapter~\ref{chap:monogamy_games}, we study the class of extended nonlocal games referred to as monogamy-of-entanglement games and prove a number of properties that these games exhibit. 


Finally, in Chapter~\ref{chap:conclusions}, we present conclusions and pose open questions that may be of interest for future research. Supplementary software used in this thesis is also provided in Appendix~\ref{chap:AppendixA}, as well as on the software repositories hosted here~\cite{Russo2015a} and here~\cite{Russo2016a}. 

\noindent The following is a list of existing work directly related to the content in this document:

\begin{itemize}

	\item[$\bullet$]
	V. Russo and J. Watrous.
	\textbf{Extended nonlocal games from quantum-classical games}. 2016,
	\cite{Russo2016}.

	\item[$\bullet$]
	N. Johnston, R. Mittal, V. Russo, and J. Watrous.
	\textbf{Extended nonlocal games and monogamy-of-entanglement games}. \textit{Proc.~R.~Soc.~A 472:20160003}, 2016,
	\cite{Johnston2015a}. 
			
\end{itemize}
The following is a list of existing work completed during my Ph.D., but not directly related to my thesis work:
\begin{itemize}

	\item[$\bullet$] 
	S. Bandyopadhyay, A. Cosentino, N. Johnston, V. Russo, J. Watrous, and N. Yu. 
	\textbf{Limitations on separable measurements by convex optimization}. 
	\textit{IEEE Transactions on Information Theory}, 2015,
	\cite{Bandyopadhyay2015}.
	
	\item[$\bullet$]
	S. Arunachalam, N. Johnston, and V. Russo.
	\textbf{Is absolute separability determined by the partial transpose?}. \textit{Quantum Information \& Computation}, 2015,
	\cite{Arunachalam2015}.

	\item[$\bullet$]
	D. Gosset, V. Kliuchinikov, M. Mosca, and V. Russo.
	\textbf{An algorithm for the T-count}. \textit{Quantum Information \& Computation}, 2014,
	\cite{Gosset2014}.
	
	\item[$\bullet$]
	A. Cosentino and V. Russo.
	\textbf{Small sets of locally indistinguishable orthogonal maximally entangled states}. 
	\textit{Quantum Information \& Computation},  2014,
	\cite{Cosentino2014}.

	\item[$\bullet$]
	S. Arunachalam, A. Molina, and V. Russo.
	\textbf{Quantum hedging in two-round prover-verifier interactions}. \textit{arXiv:1310.7954}, 2013,
	\cite{Arunachalam2013}.
	
\end{itemize}

\chapter{Preliminaries}
\label{chap:preliminaries}

In this chapter, we present an overview of the relevant subject matter of quantum information theory that will be used for the remainder of this thesis. We further establish basic terminology and notation. We shall make gratuitous use of the notation conventions for quantum information theory from~\cite{Watrous2015}. The reader is assumed to be familiar with the basic underpinnings of quantum information theory, as may be found, for instance, in the following references~\cite{Nielsen2001,Kaye2007,Wilde2013}.

We also introduce the subject of convex optimization, which as we shall see, acts as a Swiss army knife for many problems of interest in quantum information, and indeed many that we will encounter in this thesis. For further information on convex optimization, the reader is referred to~\cite{Boyd2004}. 

We shall then introduce the nonlocal game formalism. This model provides an excellent venue to abstractly study one of the most crucial features of quantum information: entanglement. We shall formally define the nonlocal game model and present relevant background work, making our treatment of the subject as self-contained as possible. 

\minitoc

\section{Basic notation, terminology, and background} 

\subsection{Alphabets, symbols, and strings}

We use capital Greek letters $\Sigma, \Gamma, \Delta$, etc.~to denote finite and nonempty sets that we refer to as \index{alphabets}{\emph{alphabets}}. We shall often use lower case characters such as $x,y,a,b$, etc.~to denote elements of alphabets called \index{symbols}{\emph{symbols}}. For an alphabet $\Sigma$, a \index{string}{\emph{string}} over $\Sigma$ is a finite sequence of symbols from $\Sigma$. The \index{length (string)}{\emph{length}} of a string is the number of symbols in the sequence. We will typically use lower case characters $s$ and $t$ to refer to strings. For every string $s$, we denote the length of $s$ as $\abs{s}$. We define the \index{empty string}{\emph{empty string}}, denoted by $\varepsilon$, to represent the string where $\abs{\varepsilon} = 0$, or in other words, the string that has length $0$. For some nonnegative integer $n \geq 0$, we say that $\Sigma^{\leq n}$ denotes all strings of length at most $n$ and we say that $\Sigma^n$ denotes all strings of length $n$ over the alphabet $\Sigma$. Note that for any alphabet $\Sigma$, one has that $\Sigma^0 = \{\varepsilon\}$. We denote the set of all strings over an alphabet $\Sigma$ as $\Sigma^{\ast}$, that is
\begin{align}
	\Sigma^{\ast} = \Sigma^0 \cup \Sigma^1 \cup \cdots .
\end{align}
For strings $s$ and $t$, we represent the \index{concatenation (string)}{\emph{concatenation}} of $s$ and $t$ as $st$, which is the string composed of $s$ followed by $t$. The \index{reversal (string)}{\emph{reversal}} of a string $s$ is denoted as $s^{\mathsmaller{R}}$.

\subsection{Vectors, operators, and mappings} 

\subsubsection*{Vectors}

We shall use $\real, \complex, \natural,$ and $\integer$ to denote the sets of real numbers, complex numbers, natural numbers (including $0$), and integers respectively. We use $\integer_n$ to denote the integers modulo $n$ as denoted by
\begin{align}
	\integer_n = \{ 0, 1, \ldots, n-1 \}.
\end{align}
For some alphabet $\Sigma$, we define a \index{complex Euclidean space}{\emph{complex Euclidean space}} as the set $\complex^{\Sigma}$, which refers to the space of all complex vectors indexed by $\Sigma$. These complex Euclidean spaces will be denoted as scripted capital letters, $\A, \B, \X, \Y$, $\Z$, etc. We use lower case characters $u,v,w,z$ to represent elements in a complex Euclidean space. 

For some alphabet $\Sigma$ and any vectors $u,v \in \complex^{\Sigma}$, the inner product is defined as 
\begin{align}
	\ip{u}{v} = \sum_{a \in \Sigma} \overline{u(a)}v(a),
\end{align}
where $u(a)$ and $v(a)$ refer to the entry of vectors $u$ and $v$ indexed by $a$ for every $u,v \in \complex^{\Sigma}$. We say that two vectors $u,v \in \complex^{\Sigma}$ are \index{orthogonal}{\emph{orthogonal}} if and only if $\ip{u}{v} = 0$. We say that a set of vectors $\{u_a : a \in \Gamma \} \subset \complex^{\Sigma}$ form an \index{orthogonal set}{\emph{orthogonal set}} if $\ip{u_a}{u_b} = 0$ for all $a,b \in \Gamma$ such that $a \not= b$. 

The \index{Euclidean norm}{\emph{Euclidean norm}} of a vector $u \in \complex^{\Sigma}$ is given by 
\begin{align}
	\norm{u} = \sqrt{\ip{u}{u}}.
\end{align}
A vector $u$ is called a \index{unit vector}{\emph{unit vector}} if $\norm{u} = 1$. The \index{unit sphere}{\emph{unit sphere}}, $\S(\X)$, for a complex Euclidean space, $\X$, is the collection of all unit vectors:
\begin{align}
	\S(\X) = \{u \in \X : \norm{u} = 1\}.
\end{align}
We say that two vectors $u,v \in \complex^{\Sigma}$ are \index{orthonormal}{\emph{orthonormal}} if in addition to $u$ and $v$ being orthogonal, they are also unit vectors. We say that a set of vectors $\{u_a : a \in \Gamma \} \subset \complex^{\Sigma}$ form an \index{orthonormal set}{\emph{orthonormal set}} if $u_a$ and $u_b$ are orthonormal for all $a,b \in \Gamma$ with $a \not= b$. We refer to an \index{orthonormal basis}{\emph{orthonormal basis}} as an orthonormal set $\{u_a : a \in \Gamma \} \subset \complex^{\Sigma}$, such that $\abs{\Gamma} = \abs{\Sigma}$. The \index{standard basis}{\emph{standard basis}} of $\complex^{\Sigma}$ is the orthonormal basis given by $\{ e_a : a \in \Sigma \}$, where 
\[
e_{a}(b) =
  \begin{cases} 
      \hfill 1 \hfill & \text{ if $a=b$}, \\
      \hfill 0 \hfill & \text{ if $a \not= b$}, \\
  \end{cases}
\]
for all $a,b \in \Sigma$. We say that two orthonormal bases 
\begin{align}
	\B_0 = \{u_a : a \in \Sigma \} \subset \complex^{\Sigma} \quad \textnormal{and} \quad \B_1 = \{ v_a : a \in \Sigma \} \subset \complex^{\Sigma}
\end{align}
are \index{mutually unbiased}{\emph{mutually unbiased}} if and only if $\abs{\ip{u_a}{v_b}} = 1/\sqrt{\Sigma}$ for all $a,b \in \Sigma$. For $n \in \natural$, a set of orthonormal bases $\{ \B_0, \ldots, \B_{n-1} \}$ are \index{mutually unbiased bases}{\emph{mutually unbiased bases}} if and only if every basis is mutually unbiased with every other basis in the set, i.e. $\B_x$ is mutually unbiased with $\B_{x^{\prime}}$ for all $x \not= x^{\prime}$ with $x,x^{\prime} \in \Sigma$. 

\subsubsection*{Operators}

We use $\Lin(\X,\Y)$ to denote the set of all linear operators from the space $\X$ to $\Y$. When convenient, we use the shorthand $\Lin(\X)$ to denote $\Lin(\X,\X)$. We shall denote linear operators as capital letters $A,B,C$, etc. Linear operators and matrices have a natural correspondence, that is, for every operator $A \in \Lin(\X,\Y)$ where $\X = \complex^{\Sigma}$ and $\Y = \complex^{\Gamma}$, one may associate the matrix $M : \Gamma \times \Sigma \rightarrow \complex$ defined as 
\begin{align}
	M(a,b) = \ip{e_a}{A e_b}
\end{align}
for all $a \in \Gamma$ and $b \in \Sigma$. For an operator, $A$, when referring to the corresponding matrix, we will overload the symbol $A$ instead of using $M$ as above. For complex Euclidean spaces $\X = \complex^{\Sigma}$ and $\Y = \complex^{\Y}$, we define the \index{standard basis (operators)}{\emph{standard basis of a space of operators}} by the collection $\{E_{a,b} : a \in \Gamma, \ b \in \Sigma \}$ that forms a basis of $\Lin(\X,\Y)$. The operator $E_{a,b}$ is defined as 
\[
 E_{a,b}(c,d) =
  \begin{cases} 
      \hfill 1 \hfill & \text{ if $(c,d) = (a,b)$}, \\
      \hfill 0 \hfill & \text{ otherwise}, \\
  \end{cases}
\]
for all $c \in \Gamma$ and $d \in \Sigma$. The \index{identity operator}{\emph{identity operator}}, $\I \in \Lin(\X)$, is the operator that obeys $\I u = u$ for all $u \in \X$. In terms of its matrix representation, the identity operator has ones along the diagonal, and zeros everywhere else. The identity operator acting on space $\X$ may be written as $\I_{\X}$ or as $\I$ if it is clear what space the operator is acting on from the context. 

For any operator $A \in \Lin(\X,\Y)$ with $\X = \complex^{\Sigma}$ and $\Y = \complex^{\Gamma}$, the \index{conjugate}{\emph{conjugate}} of $A$ is denoted as $\overline{A} \in \Lin(\X,\Y)$ where the matrix representation of $\overline{A}$ has entries that are complex conjugates of the entries in the matrix representation of $A$, that is
\begin{align}
	\overline{A}(a,b) = \overline{A(a,b)},
\end{align}
for all $a \in \Gamma$ and $b \in \Sigma$. The \index{transpose}{\emph{transpose}} of $A \in \Lin(\X,\Y)$, denoted $A^{\Trans} \in \Lin(\Y,\X)$, is the operator whose matrix representation is defined by 
\begin{align}
	A^{\Trans}(b,a) = A(a,b),
\end{align}
for all $a \in \Gamma$ and $b \in \Sigma$. For any operator $A \in \Lin(\X,\Y)$, there exists a unique operator $A^* \in \Lin(\Y,\X)$ that is referred to as the \index{adjoint}{\emph{adjoint}}, where $A^*$ satisfies the equation
\begin{align}
	\ip{v}{Au} = \ip{A^*v}{u},
\end{align}
for all $u \in \X$ and $v \in \Y$. In the matrix representation, $A^*$ is the \index{conjugate transpose}{\emph{conjugate transpose}} of $A$, that is
\begin{align}
	A^* = \left(\overline{A}\right)^{\Trans} = \overline{\left( A^{\Trans}\right)}.
\end{align}
The \index{trace}{\emph{trace}} of an operator $A \in \Lin(\X)$ is the sum of its diagonal elements, that is
\begin{align} \label{eq:trace}
	\tr(A) = \sum_{a \in \Sigma} A(a,a).
\end{align}
For operators $A,B \in \Lin(\X,\Y)$ we denote the \index{Hilbert-Schmidt inner product}{\emph{Hilbert-Schmidt inner product}} as 
\begin{align}
	\ip{A}{B} = \tr(A^*B). 
\end{align}
For any operators $A,B \in \Lin(\X)$ we define the \index{Lie bracket}{\emph{Lie bracket}} $\left[A,B\right]$ as 
\begin{align}
	\left[A,B\right] = AB - BA. 
\end{align} 
We say that operators $A$ and $B$ \index{commute}{\emph{commute}} if and only if $\left[A,B\right]=0$. 

For any space $\X$, we define the following types of operators acting on the space $\X$: 
\begin{itemize}
	\item {\it Hermitian operators.} An operator $H \in \Lin(\X)$ is \index{Hermitian}{\emph{Hermitian}} if $H = H^*$. We use $\Herm(\X)$ to denote the set of all Hermitian operators. 

	\item {\it Positive semidefinite operators.} An operator $P \in \Lin(\X)$ is \index{positive semidefinite}{\emph{positive semidefinite}} if and only if it holds that $P = X^*X$ for some operator $X \in \Lin(\X)$. We use $\Pos(\X)$ to denote the set of all positive semidefinite operators. 
	
	\item {\it Density operators.} An operator $\rho \in \Lin(\X)$ is a \index{density operator}{\emph{density operator}} if $\rho \in \Pos(\X)$ and $\tr(\rho) = 1$. We use $\Density(\X)$ to denote the set of all density operators. 
	
	\item {\it Projection operators.} An operator $\Pi \in \Pos(\X)$ is a \index{projection operator}{\emph{projection operator}} if $\Pi^2 = \Pi$. We use $\Proj(\X)$ to denote the set of all projection operators. 
	
	\item {\it Unitary operators.} An operator $U \in \Lin(\X)$ is a \index{unitary operator}{\emph{unitary operator}} if $U$ is a linear isometry from $\X$ to $\Y$, where a linear isometry is an operator $U \in \Lin(\X,\Y)$ such that $U^*U = \I_{\X}$. 
\end{itemize}
For any space $\X$, the aforementioned operators obey the following relationships
\begin{align}
	\Density(\X) \subset \Pos(\X) \subset \Herm(\X) \subset \Lin(\X) \quad \textnormal{and} \quad \Proj(\X) \subset \Pos(\X),
\end{align}
as well as 
\begin{align}
	\Unitary(\X) \subset \Lin(\X).
\end{align}

\subsubsection*{Norms}

For any complex Euclidean spaces $\X$ and $\Y$ and any operator $A \in \Lin(\X,\Y)$, we define a \index{norm}{\emph{norm}} of $A$, denoted as $\norm{A}$, as a function which satisfies the following conditions:
\begin{enumerate}
	\item $\norm{A} \geq 0$ for all $A \in \Lin(\X,\Y)$,
	\item $\norm{A} = 0$ if and only if $A = 0$ for all $A \in \Lin(\X,\Y)$,
	\item $\norm{\alpha A} = \abs{\alpha}\norm{A}$ for all $\alpha \in \complex$ and for all $A \in \Lin(\X,\Y)$,
	\item $\norm{A + B} \leq \norm{A} + \norm{B}$ for all $A,B \in \Lin(\X,\Y)$. 
\end{enumerate}
 For any operator $A \in \Lin(\X,\Y)$ and any real number $p \geq 1$, one may define the \index{Schatten $p$-norms}{\emph{Schatten $p$-norms}} as
\begin{align}
	\norm{A}_p = \left( \tr \left( \left( A^* A \right)^{\frac{p}{2}} \right) \right)^{\frac{1}{p}}.
\end{align}
In particular, we focus on the Schatten $p$-norms for $p = 1$ and $p = \infty$ which are given the special names of the trace norm and the spectral norm, respectively.
\begin{itemize}
	\item \index{trace norm}{\emph{Trace norm}}. The \emph{trace norm} of an operator $A \in \Lin(\X,\Y)$ is defined by
\begin{align}
	\norm{A}_1 = \left( \tr\left( (A^* A)^{\frac{1}{2}} \right) \right)^{\frac{1}{1}} = \tr(\sqrt{A^*A}),
\end{align}
where $\sqrt{A}$ is the unique positive semidefinite operator called the \index{square root (of operator)}{\emph{square root}} of $A$ that has the property $\left(\sqrt{A} \right)^2 = A$. 
	\item \index{spectral norm}{\emph{Spectral norm}}. The \emph{spectral norm} of an operator $A \in \Lin(\X,\Y)$ is defined by
	\begin{align}
		\norm{A}_{\infty} = \max \left \{ \norm{Au} : u \in \X, \norm{u} = 1 \right \}.
	\end{align}
	When referring to the spectral norm, we often will drop the $\infty$ subscript from $\norm{\cdot}_{\infty}$ to just $\norm{\cdot}$. 
\end{itemize}

\subsubsection*{The tensor product}

For a set of $n$ complex Euclidean spaces, $\X_1 = \complex^{\Sigma_1}, \ldots, \X_n = \complex^{\Sigma_n}$, the \index{tensor product}{\emph{tensor product}} of these spaces is given by
\begin{align}
	\X_1 \otimes \cdots \otimes \X_n = \complex^{\Sigma_1 \times \cdots \times \Sigma_n}.
\end{align}
One may consider the tensor product acting on vectors $u_1 \in \X_1, \ldots, u_n \in \X_n$ denoted as 
\begin{align}
	u_1 \otimes \cdots \otimes u_n \in \X_1 \otimes \cdots \otimes \X_n,
\end{align}
which refers to the vector 
\begin{align}
	\left( u_1 \otimes \cdots \otimes u_n \right) \left( a_1, \ldots, a_n \right) = u_1(a_1) \cdots u_n(a_n).
\end{align}
One may also consider the tensor product acting on operators. For complex Euclidean spaces $\X_1 = \complex^{\Sigma_1}, \ldots, \X_n = \complex^{\Sigma_n}$ and $\Y_1 = \complex^{\Gamma_1}, \ldots, \Y_n = \complex^{\Gamma_n}$, for alphabets $\Sigma_1, \ldots, \Sigma_n$ and $\Gamma_1, \ldots, \Gamma_n$, define a set of operators 
\begin{align}
	A_1 \in \Lin(\X_1,\Y_1), \ldots, A_n \in \Lin(\X_n,\Y_n).
\end{align} 
We then define the tensor product acting on operators $A_1, \ldots, A_n$ as 
\begin{align} \label{eq:A-operators}
	A_1 \otimes \cdots \otimes A_n \in \Lin(\X_1 \otimes \cdots \otimes \X_n, \Y_1 \otimes \cdots \otimes \Y_n),
\end{align}
where the tensor product of $A_1, \ldots, A_n$ is the unique operator that satisfies 
\begin{align}
	\left(A_1 \otimes \cdots \otimes A_n \right) \left( u_1 \otimes \cdots \otimes u_n \right) = \left(A_1 u_1 \right) \otimes \cdots \otimes \left( A_n u_n \right),
\end{align}
for all $u_1 \in \X_1, \ldots, u_n \in \X_n$. 


For any complex Euclidean space $\X$, we may also use the shorthand $\X^{\otimes n}$ to denote the $n$-fold tensor product of $\X$ with itself, that is
\begin{align}
	\X^{\otimes n} = \underbrace{\X \otimes \cdots \otimes \X}_{n\textnormal{-times}}.
\end{align}

\subsubsection*{Mappings}

We denote linear mappings acting on operators as $\Phi : \Lin(\X) \rightarrow \Lin(\Y)$. We use $\Trans(\X,\Y)$ to denote the set of all such mappings. Each $\Phi \in \Trans(\X,\Y)$ has a unique adjoint mapping $\Phi^* \in \Trans(\Y,\X)$ defined as 
\begin{align}
	\ip{Y}{\Phi(X)} = \ip{\Phi^*(Y)}{X},
\end{align}
for all $X \in \Lin(\X)$ and $Y \in \Lin(\Y)$. For instance, for an operator $X \in \Lin(\X)$ where $\X = \complex^{\Sigma}$, the trace function from equation~\eqref{eq:trace} may be described as a mapping of the following form
\begin{align}
	\tr: \Lin(\X) \rightarrow \complex. 
\end{align}
For operators $X \in \Lin(\X)$ and $Y \in \Lin(\Y)$, the \index{partial trace}{\emph{partial trace}} is a map defined as $\tr_{\Y} \in \Trans(\X \otimes \Y, \X)$
\begin{align}
 \tr_{\Y} = \I_{\X} \otimes \tr.
\end{align}
For a space $\X$, the \index{identity map}{\emph{identity map}}, $\I_{\Lin(\X)} \in \Trans(\X)$, is given as 
\begin{align}
	\I_{\Lin(\X)}(X) = X
\end{align}
for all $X \in \Lin(\X)$.

We shall make use of a correspondence between $\Lin(\Y,\X)$ and $\X \otimes \Y$ for spaces $\X = \complex^{\Sigma}$ and $\Y = \complex^{\Gamma}$. This serves as a correspondence between operators and vectors, and is denoted by the ``$\vec$'' linear mapping
\begin{align}
	\vec: \Lin(\Y,\X) \rightarrow \X \otimes \Y
\end{align} 
defined by 
\begin{align}
	\vec(E_{a,b}) = e_a \otimes e_b
\end{align}
for all $a \in \Sigma$ and $b \in \Gamma$. Using the matrix representation of $A \in \Lin(\Y,\X)$, the $\vec$ mapping can be thought of as stacking the rows of $A$ to form a single vector. For example, for the matrix
\begin{align}
	A = \begin{pmatrix} a_{1,1} & \cdots & a_{1,n} \\ \vdots & \ddots & \vdots \\ a_{n,1} & \cdots & a_{n,n} \end{pmatrix} \in \Lin(\Y,\X),
\end{align}
the $\vec$ mapping has the following effect
\begin{align}
	\vec \left( A \right) = \left( a_{1,1}, \ldots, a_{1,n}, \ldots, a_{n,1}, \ldots, a_{n,n} \right)^{\Trans} \in \X \otimes \Y.
\end{align}

For arbitrary spaces $\X$ and $\Y$, we consider the following useful sets of linear mappings:
\begin{itemize}
	\item \index{completely positive (map)}{\emph{Completely positive}}. A mapping $\Phi \in \Trans(\X,\Y)$ is \emph{completely positive} if
\begin{align}
	\Phi \otimes \I_{\Z} (X) \in \Pos(\Y \otimes \Z),
\end{align}
for each complex Euclidean space $\Z$ and for any $X \in \Pos(\X \otimes \Z)$.
	\item \index{trace preserving (map)}{\emph{Trace preserving}}. A mapping $\Phi \in \Trans(\X,\Y)$ is \emph{trace preserving} if
\begin{align}
	\tr(\Phi(X)) = \tr(X)
\end{align}
for all $X \in \Lin(\X)$. 
	\item \index{Hermiticity-preserving (map)}{\emph{Hermiticity preserving}}. A mapping $\Phi \in \Trans(\X,\Y)$ is \emph{Hermiticity preserving} if
\begin{align}
	\Phi(H) \in \Herm(\Y)
\end{align}
for every Hermitian operator $H \in \Herm(\X)$. 
\end{itemize}

%

\subsection{Operator decompositions and vector decompositions}

The following operator and vector decompositions are fundamental to many proofs that appear in quantum information, and indeed also appear as essential steps in the proofs in this thesis. 

The \index{singular value theorem}{\emph{singular value theorem}} states that for any nonzero operator $A \in \Lin(\X,\Y)$ with $r = \rank(A)$, that there exists positive real numbers $s_1, \ldots, s_r \in \real$ and orthonormal sets $\{x_1, \ldots, x_r\} \subset \X$ and $\{y_1, \ldots, y_r\} \subset \Y$ such that 
\begin{align}
	A = \sum_{i=1}^r s_i y_i x_i^*.
\end{align}
Such a decomposition is referred to as a \index{singular value decomposition}{\emph{singular value decomposition}}. We refer to the real numbers $s_1, \ldots, s_r$ as the \index{singular values}{\emph{singular values}} of $A$ and the sets $y_1, \ldots, y_r$ and $x_1, \ldots, x_r$ are usually called the \index{left singular vectors}{\emph{left singular vectors}} and \index{right singular vectors}{\emph{right singular vectors}} of $A$, respectively. 

The \index{spectral theorem}{\emph{spectral theorem}} states that an operator $A \in \Lin(\X)$ with $r = \rank(A)$ is Hermitian if and only if there exists real numbers $\lambda_1, \ldots, \lambda_r \in \real$, and an orthonormal set $\{x_1, \ldots, x_r \} \subset \X$ such that 
\begin{align}
	A = \sum_{i=1}^r \lambda_i x_i x_i^*.
\end{align}
Such a decomposition is called a \index{spectral decomposition}{\emph{spectral decomposition}}. We refer to the numbers $\lambda_1, \ldots, \lambda_r$ as the \index{eigenvalues}{\emph{eigenvalues}} of $A$ and the vectors $x_1, \ldots, x_r$ as the \index{eigenvectors}{\emph{eigenvectors}} of $A$.   
 
The \index{Schmidt decomposition}{\emph{Schmidt decomposition}} of an arbitrary nonzero vector $u \in \X \otimes \Y$ consists of a positive integer $r \geq 1$ and orthonormal sets $\{x_1, \ldots, x_r\} \subset \X$ and $\{y_1, \ldots, y_r \} \subset \Y$ such that $u$ may be expressed as 
\begin{align}
	u = \sum_{i=1}^r s_i x_i \otimes y_i.
\end{align}

\subsection{Convexity and semidefinite programming}

\subsubsection*{Convexity}

We shall denote finite-dimensional real or complex vector spaces as either $\V$ or $\W$. In this section, the space $\V$ will typically denote either $\real^n$ or $\complex^n$, for some finite $n > 1$, and $\W$ shall be a subset of $\V$. We say that a set $\W \subseteq \V$ is \index{convex}{\emph{convex}} if for all $u,v \in \W$ and all $\lambda \in [0,1]$ it is true that 
\begin{align}
	\lambda u + (1-\lambda)v \in \W.
\end{align}
Otherwise, we say that $\W$ is \index{non-convex}{\emph{non-convex}} or \emph{not convex}. We say that a set $\W \subseteq \V$ is \index{open (set)}{\emph{open}} if and only if for all elements $w \in \W$ there exists a real number $\epsilon > 0$ such that 
\begin{align}
	\{ v \in \V : \norm{w-v} < \epsilon \} \subseteq \W.
\end{align}
We say that a set $\W \subseteq \V$ is \index{closed (set)}{\emph{closed}} if and only if it is the complement of an open set. For $\W \subseteq \V$ we refer to a \index{sequence}{\emph{sequence}} of vectors in $\W$ as a function 
\begin{align}
	s: \natural \rightarrow \W
\end{align} 
where a sequence is denoted as $s(n) = u_n$ with $u_n \in \W$ for all $n \in \natural$. A \index{subsequence}{\emph{subsequence}} is a sequence that is obtainable from some sequence by removing elements without altering the order of the elements that remain. For $\W \subseteq \V$, we say that a sequence $s(n) \in \W$ is a \index{convergent sequence}{\emph{convergent sequence}} or converges to $v \in \V$ if for any real number $\epsilon > 0$ there exists $N \in \natural$ such that
\begin{align}
	\norm{s(n) - v} < \epsilon
\end{align}
for all $n > N$. We say that a set is \index{compact}{\emph{compact}} if and only if every sequence in $\W$ has a convergent subsequence.

We define a \index{probability vector}{\emph{probability vector}} $p \in \real^{\Sigma}$ for some alphabet $\Sigma$ if it satisfies the property 
\begin{align}
	p(a) \geq 0
\end{align} 
for all $a \in \Sigma$ as well as 
\begin{align}
	\sum_{a \in \Sigma} p(a) = 1. 
\end{align}
We use $p \in \P(\Sigma)$ to denote the set of all such probability vectors. We define a \index{convex combination}{\emph{convex combination}} of vectors in $\W$ as 
\begin{align}
	\sum_{a \in \Sigma} p(a) u_a,
\end{align}
where $\Sigma$ is some alphabet, $p \in \P(\Sigma)$ is a probability vector, and 
\begin{align}
	\{ u_a : a \in \Sigma \} \subseteq \W,
\end{align}
is a collection of vectors in $\W$. 

\subsubsection*{Hilbert spaces}
In this thesis, we will be primarily concerned with finite-dimensional complex Euclidean spaces, however, we will encounter a few results that will require the use of a possibly infinite-dimensional space. We, therefore, introduce the notion of a \index{Hilbert space}{\emph{Hilbert space}}, which generalizes finite-dimensional complex Euclidean spaces to spaces with any finite or infinite number of dimensions. Specifically, we will restrict our attention to \index{separable Hilbert space}{\emph{separable Hilbert spaces}}, that is a Hilbert space that has a countable orthonormal basis. In this thesis, whenever we refer to a Hilbert space, it is assumed that we are referring to a separable Hilbert space. We will always refer to such Hilbert spaces as $\H$ to distinguish them from finite-dimensional complex Euclidean spaces. Much of the discussion thus far on finite-dimensional complex Euclidean spaces may be ported over to infinite-dimensional Hilbert spaces, but we make note of a few key differences between them. 

Let $\{e_n : n \in \natural \}$ be a countable orthonormal basis of a Hilbert space $\H$. Then we can write each element $v \in \H$ as 
\begin{align}
	v = \sum_{n=1}^{\infty} \ip{v}{e_n}e_n.
\end{align}
We may also consider operators acting on a (possibly) infinite-dimensional Hilbert space. Given Hilbert spaces $\H_1$ and $\H_2$, one writes $\B(\H_1,\H_2)$ to refer to the collection of all \index{bounded operators}{\emph{bounded operators}} of the form
\begin{align}
	A : \H_1 \rightarrow \H_2,
\end{align}
such that 
\begin{align}
	\norm{Av} \leq c\norm{v}
\end{align}
for all $v \in \H_1$ and for some constant $c > 0$. We use the shorthand $\B(\H)$ to refer to the collection of $\B(\H,\H)$ bounded operators. Every bounded operator $A \in \B(\H)$ has a unique adjoint operator $A^* \in \B(\H)$ satisfying 
\begin{align}
	\ip{u}{Av} = \ip{A^*u}{v},
\end{align}
for all $u,v \in \H$, behaving in a similar fashion to adjoints on finite-dimensional complex Euclidean spaces. A positive semidefinite operator $P \in \B(\H)$ is defined in an analogous way to positive semidefinite operators over finite-dimensional spaces, namely that 
\begin{align}
	P = X^*X
\end{align}
for some operator $X \in \B(\H)$. Given an orthonormal basis $\{e_n : n \in \natural\} \subset \H$, we say that $A \in \B(\H)$ is a \index{trace class}{\emph{trace class}} operator if and only if 
\begin{align}
	\sum_{n \in \natural} \bigip{ \abs{A} e_n}{e_n} < \infty,
\end{align} 
where $\abs{A} = \sqrt{A^*A} \in \B(\H)$. For $A \in \B(\H)$, define 
\begin{align}
	\norm{A}_1 = \sum_{n \in \natural} \bigip{\abs{A}e_n}{e_n}.
\end{align}
We may therefore say that a bounded operator $A \in \B(\H)$ is also trace class if $\norm{A}_1 < \infty$. A density operator $\rho \in \B(\H)$ is both a bounded operator and a trace class operator.  

Let $\Y$ be a Banach space and let $\X = \Y^*$. Then we say that a sequence \index{weak-* convergence}{\emph{converges weak-*}} to a vector $f \in \X$ if 
\begin{align}
	\lim_{n \rightarrow \infty} f_n(v) = f(v),
\end{align}
for all $v \in \Y$. A consequence of the so-called \index{Banach-Alaoglu theorem}{\emph{Banach-Alaoglu theorem}}~\cite{Rudin1991} that we will use in Chapter~\ref{chap:extended_npa_hierarchy} is that every bounded sequence has a weak-* convergent subsequence provided $\Y$ is separable.

\subsubsection*{Semidefinite programming}

Let $\X$ and $\Y$ be complex Euclidean spaces, $A\in\Herm(\X)$ and $B\in\Herm(\Y)$
be Hermitian operators, and $\Phi \in \Trans(\X,\Y)$ be a Hermiticity preserving mapping.
A \index{semidefinite program}{\emph{semidefinite program}} (SDP) is defined by the triple $(A,B,\Phi)$ and is identified with
the following pair of optimization problems.
\begin{center}
  \begin{minipage}{2in}
    \centerline{\underline{Primal problem}}\vspace{-7mm}
    \begin{align*}
      \text{maximize:}\quad & \ip{A}{X}\\
      \text{subject to:}\quad & \Phi(X) = B,\\
      & X \in \Pos(\X).
    \end{align*}
  \end{minipage}
  \hspace*{1.5cm}
  \begin{minipage}{2.4in}
    \centerline{\underline{Dual problem}}\vspace{-7mm}
    \begin{align*}
      \text{minimize:}\quad & \ip{B}{Y}\\
      \text{subject to:}\quad & \Phi^{\ast}(Y) \geq A,\\
      & Y\in\Herm(\Y).
    \end{align*}
  \end{minipage}
\end{center}
An equivalent formulation of the above primal and dual problems is the so called ``standard form'' which is written as
\begin{center} 
	\begin{equation} \label{sdp:standard-form}
  \begin{minipage}{2.6in}
    \centerline{\underline{Primal problem}}\vspace{-7mm}
    \begin{align*}
      \text{maximize:}\quad & \ip{A}{X} \\
      \text{subject to:}\quad & \ip{B_1}{X} = \gamma_1, \\
      & \vdots \\
      & \ip{B_m}{X} = \gamma_m, \\
      & X\in\Pos(\X).
    \end{align*}
  \end{minipage}
  \hspace*{13mm}
  \begin{minipage}{2.6in}
    \centerline{\underline{Dual problem}}\vspace{-7mm}
    \begin{align}
      \text{minimize:}\quad & \sum_{j = 1}^m \gamma_j y_j \nonumber \\
      \text{subject to:}\quad & \sum_{j=1}^m y_j B_j \geq A, \nonumber \\
      & y_1, \ldots, y_m \in \real. \nonumber
    \end{align}
  \end{minipage}
  \end{equation}
\end{center}

In this case, $B_1, \ldots, B_m \in \Herm(\X)$ replace the $\Phi$ operators and $\gamma_1, \ldots, \gamma_m \in \real$ replace the $B$ operators. A proof of the equivalence between the two SDP formulations may be found in~\cite{Watrous2004}. One may prefer to use either form depending on the specifics of the problem and convenience of representation. 

\section{Quantum information theory}

\subsection{Quantum states, operations, and measurements}

We shall refer to the class of density operators interchangeably as \emph{quantum states}. For some state $\rho \in \Density(\X)$, we refer to $\rho$ as a \index{pure state}{\emph{pure state}} if $\rho$ additionally satisfies the constraint that $\rank(\rho) = 1$. Equivalently, the state $\rho$ is pure if there exists some vector $u \in \X$ such that $\rho = uu^*$. Otherwise, if $\rho$ is not pure, then we refer to $\rho$ as a \index{mixed state}{\emph{mixed state}}. From the spectral theorem, it follows that every quantum state may be written as a convex combination of pure states.  

For some state $\rho \in \Density(\X)$, one may consider a \index{register}{\emph{register}}, denoted as $\reg{X}$, as a computational abstraction in which the actions on the state $\rho$ are carried out. For spaces $\X, \Y,$ and $\Z$, we shall denote the corresponding registers as $\reg{X}$, $\reg{Y}$, and $\reg{Z}$, respectively. For a register $\reg{X}$, we use $\abs{\reg{X}}$ to denote the size of the register $\reg{X}$, where the size is indicative of the dimension of $\X$. We refer to registers of the binary values, $\{0,1\}$, as \index{qubits}{\emph{qubits}}. 

For some register $\reg{X}$, we may consider \index{measurement}{\emph{measurements}} on this register as being described by a set of positive semidefinite operators $\{P_a : a \in \Gamma \} \subset \Pos(\X)$ indexed by the alphabet $\Gamma$ of measurement outcomes satisfying the constraint that 
\begin{align}
	\sum_{a \in \Gamma} P_a = \I_{\X}.
\end{align}
Performing a measurement on $\reg{X}$ in state $\rho$, the outcome $a \in \Gamma$ results with probability $\ip{P_a}{\rho}$. We call a measurement $\{\Pi_a : a \in \Gamma \}$ a \index{projective measurement}{\emph{projective measurement}} if and only if all of the measurement operators are projection operators, i.e. $\Pi_a \in \Proj(\X)$ for all $a \in \Gamma$. For a projective measurement $\{\Pi_a : a \in \Gamma \} \subset \Proj(\X)$ and associated real number outcomes $\{ \lambda_a : a \in \Gamma \}$ the \index{observable}{\emph{observable}} corresponding to this measurement is 
\begin{align}
	A = \sum_{a \in \Gamma} \lambda_a \Pi_a.
\end{align}

%

We define a \index{quantum channel}{\emph{quantum channel}} as a linear mapping $\Phi \in \Trans(\X,\Y)$ that is completely positive and trace preserving. The set of all channels is denoted by $\Channel(\X,\Y)$.

For some complex Euclidean space $\X$, any state $\rho \in \Density(\X)$ may be \index{purified}{\emph{purified}}, that is, we are guaranteed that there exists a complex Euclidean space, $\Y$, with $\dim(\Y) = \rank(\rho)$, and a unit vector $u \in \X \otimes \Y$ such that 
\begin{align} \label{eq:purification}
	\rho = \tr_{\Y}(uu^*).
\end{align}
We refer to the state $uu^*$ as a \index{purification}{\emph{purification}} of $\rho$. A proof that a purification can be performed for any state can be seen by writing $\rho$ in terms of its spectral decomposition for some basis $\{ x_1, \ldots, x_r \} \subset \X$ and set of nonnegative real numbers $s_1, \ldots, s_r \in \real$ such that
\begin{align}
	\rho = \sum_{i=1}^r s_i x_i x_i^*.
\end{align} 
Define a state $u \in \X \otimes \Y$, which can be written in terms of its Schmidt decomposition as 
\begin{align}
	u = \sum_{i=1}^r \sqrt{s_i}x_i \otimes y_i,
\end{align}
where $\{ y_1, \ldots, y_r \}$ is orthonormal. Equation~\eqref{eq:purification} then follows from a routine calculation
\begin{equation}
	\begin{aligned}
		\tr_{\Y}(uu^*) &= \tr_{\Y} \left( \left( \sum_{i=1}^r \sqrt{s_i} x_i \otimes y_i \right) \left( \sum_{j=1}^r \sqrt{s_j} x_j \otimes y_j \right)^* \right) \\
		&= \tr_{\Y} \left( \sum_{i,j}^r \sqrt{s_i s_j} x_i x_j^* \otimes y_i y_j^* \right) \\
		&= \sum_{i,j}^r \delta_{i,j} \sqrt{s_i s_j} x_i x_j^* = \rho,
	\end{aligned}
\end{equation}
where we use $\delta_{i,j}$ to denote the \index{Kronecker delta function}{\emph{Kronecker delta function}} defined as 
\[
\delta_{i,j} =
  \begin{cases} 
      \hfill 0 \hfill & \text{ if $i \not= j$}, \\
      \hfill 1 \hfill & \text{ if $i = j$}. 
  \end{cases}
\]

\subsection{Entanglement and separability}

For complex Euclidean spaces $\X = \complex^{\Sigma}$ and $\Y = \complex^{\Gamma}$, we say that a pure state $u \in \X \otimes \Y$ is \index{separable}{\emph{separable}}, or equivalently that $u$ is a \index{product state}{\emph{product state}}, if it can be written as 
\begin{align} \label{eq:bipartite-separable}
	u = v \otimes w,
\end{align}
for some $v \in \X$ and $w \in \Y$. Otherwise, we say that $u$ is \index{entangled state}{\emph{entangled}}. Equation~\eqref{eq:bipartite-separable} is over two systems, $\X$ and $\Y$. We refer to such a system as a \index{bipartite system}{\emph{bipartite system}}. However the notion of separability extends to multipartite systems. For an integer $n > 1$ and complex Euclidean spaces $\X_1 = \complex^{\Sigma_1}, \ldots, \X_n = \complex^{\Sigma_n}$, we say that a pure state $u \in \X_1 \otimes \cdots \otimes \X_n$ is separable if it can be written as 
\begin{align}
	u = v_1 \otimes \cdots \otimes v_n
\end{align}
for some $v_1 \in \X_1, \ldots, v_n \in \X_n$. Otherwise, $u$ is entangled. A pure state $u \in \X \otimes \Y$ with $\X = \complex^{\Sigma}$ and $\Y = \complex^{\Gamma}$ such that $\abs{\Sigma} \geq \abs{\Gamma}$ is \index{maximally entangled}{\emph{maximally entangled}} if 
\begin{align}
	\tr_{\X}(uu^*) = \frac{\I_{\Y}}{\abs{\Gamma}}.
\end{align}
For some positive integer $m$, the canonical bipartite maximally entangled state is written as 
\begin{align}
	u = \frac{1}{\sqrt{m}} \sum_{c \in \integer_m} e_c \otimes e_c.
\end{align}


The notions of entanglement and separability also apply to operators. For $n > 1$ and complex Euclidean spaces $\X_1 = \complex^{\Sigma_1}, \ldots, \X_n = \complex^{\Sigma_n}$, the operator $R \in \Pos(\X_1 \otimes \cdots \otimes \X_n)$ is \emph{separable} if there exists $n$ collections of positive semidefinite operators 
\begin{align}
	\{P_{a,1} : a \in \Sigma_1 \} \subset \Pos(\X_1), \ldots, \{ P_{a,n} : a \in \Sigma_n \} \subset \Pos(\X_n),
\end{align}
such that 
\begin{align} \label{eq:separability-criteria}
	R = \sum_{a \in \Sigma} P_{a,1} \otimes \cdots \otimes P_{a,n}.
\end{align}
For complex Euclidean spaces $\X$ and $\Y$, we refer to the bipartite system described by operators $P \in \Pos(\X \otimes \Y)$ satisfying the condition in equation~\eqref{eq:separability-criteria} as being contained in the set $\Sep(\X_1 : \ldots : \X_n)$. We refer to such elements in this set as \index{separable operators}{\emph{separable operators}}. If the $P$ operators are also density matrices, that is if 
\begin{align}
	P \in \Sep(\X : \Y) \cap \Density(\X \otimes \Y),
\end{align}
then we say that $P \in \SepD(\X : \Y)$. We refer to such elements in this set as \index{separable density operator}{\emph{separable density operators}}. In contrast to being separable, if instead we have that $P \not \in \Sep(\X : \Y)$, then we refer to $P$ as an \index{entangled operator}{\emph{entangled operator}}. 

The following state
\begin{align} \label{eq:entangled-bell-state}
\tau = \frac{1}{2} \left(E_{0,0} \otimes E_{0,0} + E_{0,1} \otimes E_{0,1} + E_{1,0} \otimes E_{1,0} + E_{1,1} \otimes E_{1,1} \right),
\end{align}
is an example of an entangled operator, $\tau \not \in \Sep(\X : \Y)$, since $\tau$ cannot be written as a convex combination of tensor products. The entangled operator from equation~\eqref{eq:entangled-bell-state} is also maximally entangled, and is one state that is composed from an important class of states referred to as the \index{Bell state}{\emph{Bell states}},
\begin{equation} \label{eq:Bell-states-vecs}
	\begin{aligned}
		u_{0} = \frac{1}{\sqrt{2}} \left( e_0 \otimes e_0 + e_1 \otimes e_1 \right), &\quad u_{1} = \frac{1}{\sqrt{2}} \left( e_0 \otimes e_1 + e_1 \otimes e_0 \right), \\
		u_{2} = \frac{1}{\sqrt{2}} \left( e_0 \otimes e_1 - e_1 \otimes e_0 \right), &\quad u_{3} = \frac{1}{\sqrt{2}} \left( e_0 \otimes e_0 - e_1 \otimes e_1 \right),
	\end{aligned}
\end{equation}
where the state from equation~\eqref{eq:entangled-bell-state} is given by $\tau = u_{0} u_{0}^*$. 

An important class of unitary operators are the so called \index{Pauli operators}{\emph{Pauli operators}} defined by the matrices
\begin{align} \label{eq:pauli-ops}
	\I = \begin{pmatrix} 1 & 0 \\ 0 & 1 \end{pmatrix}, \quad X = \begin{pmatrix} 0 & 1 \\ 1 & 0 \end{pmatrix}, \quad Y = \begin{pmatrix} 0 & -i \\ i & 0 \end{pmatrix}, \quad Z = \begin{pmatrix} 1 & 0 \\ 0 & -1 \end{pmatrix}, 
\end{align} 
where $\I, X, Y, Z \in \Unitary(\complex^2)$.
For any positive integer $m$, the generalizations for the Pauli-$X$ and Pauli-$Z$ operators are defined as 
\begin{align}
	X_m = \sum_{c \in \integer_m} e_{c+1} e_c^* \quad \textnormal{and} \quad Z_m = \sum_{c \in \integer_m} \gamma_m(c) e_c e_c^*,
\end{align}
where 
\begin{align}
	\gamma_m(c) = \exp(2 \pi i c / m).
\end{align}
From this, we define the \index{generalized Pauli operators}{\emph{generalized Pauli operators}} in $U(\complex^m)$ as the set
\begin{align} \label{eq:generalized-pauli-ops}
	\left \{ W_{k_1,k_2}^{(m)} : k_1, k_2 \in \integer_m \right \}.
\end{align}
where $W_{k_1, k_2}^{(m)} = X_m^{k_1} Z_m^{k_2}$. For instance, for $m = 2$, writing the generalized Pauli operators as
\begin{align}
	\I = W_{0,0}^{(2)}, \quad X = W_{1,0}^{(2)}, \quad Y = i W_{1,1}^{(2)}, \quad Z = W_{0,1}^{(2)},
\end{align}
recovers the standard Pauli operators from equation~\eqref{eq:pauli-ops}. One may also consider a generalization of the Bell states to higher dimensions. We define the \index{generalized Bell basis}{\emph{generalized Bell basis}} density operators as a set, $\left \{ \phi^{(m)}_{k_1,k_2} : k_1,k_2 \in \integer_m \right \}$, where
\begin{align} \label{eq:generalized-bell-basis}
	\phi^{(m)}_{k_1,k_2} = \frac{1}{m} \vec \left( W^{(m)}_{k_1, k_2} \right) \vec \left( W^{(m)}_{k_1,k_2} \right)^*.
\end{align}
A quick calculation reveals that for $m = 2$, equation~\eqref{eq:generalized-bell-basis} gives 
\begin{equation}
	\begin{aligned}
		\phi_{0,0}^{(2)} = u_0 u_0^*, &\quad \phi_{0,1}^{(2)} = u_3 u_3^*, \\
		\phi_{1,0}^{(2)} = u_1 u_1^*, &\quad \phi_{1,1}^{(2)} = u_2 u_2^*,
	\end{aligned}
\end{equation}
which are the density operators that correspond to the Bell states from equation~\eqref{eq:Bell-states-vecs}.

\subsection{Teleportation} \label{sec:teleportation}

One of the most intriguing protocols in quantum information is that of \index{teleportation}{\emph{teleportation}}: a process in which one party transmits a qubit to another party using resources consisting of a pair of maximally entangled qubits and two bits of communication~\cite{Bennett1993}. The traditional teleportation process may be generalized.

\begin{figure}[!htpb] 
	\begin{center}
		\includegraphics[scale=1.0]{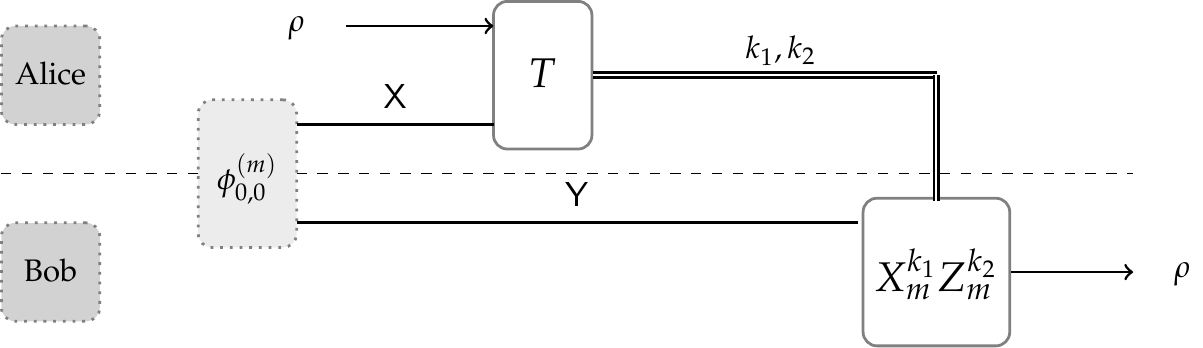}
	\end{center}
		\caption[The teleportation protocol.]{The teleportation protocol. Alice's goal is to teleport the state $\rho$ to Bob. The dashed line in the center separates the actions of Alice and Bob. Alice and Bob prepare a maximally entangled state where part of the state is contained in Alice's register $\reg{X}$ and the other part is contained in Bob's register $\reg{Y}$. Alice performs a Bell measurement and sends $k_1,k_2 \in \integer_m$, from this measurement to Bob. Bob receives $k_1$ and $k_2$ and applies of the generalized Pauli operators to his register $\reg{Y}$. The end result is that Bob now possesses the state $\rho$. }
	\label{fig:teleportation}
\end{figure}

Suppose that Alice and Bob prepare registers $(\reg{X},\reg{Y})$ where Alice holds $\reg{X}$ and Bob holds $\reg{Y}$ such that 
\begin{align}
	\abs{\reg{X}} = m = \abs{\reg{Y}},
\end{align}
where the contents of $(\reg{X},\reg{Y})$ corresponds to the maximally entangled state $\phi_{0,0}^{(m)}$. Alice obtains a new state, $\rho$, contained in register $\reg{Z}$ that she desires to send to Bob. In order to do so, both parties abide by the generalized teleportation protocol, that is depicted in Figure~\ref{fig:teleportation}.

\begin{enumerate}
\item Alice measures $(\reg{Z},\reg{X})$ with respect to the generalized Bell basis as defined from equation~\eqref{eq:generalized-bell-basis}
\begin{align}
	\left \{ \phi^{(m)}_{k_1, k_2} : k_1, k_2 \in \integer_m \right \},
\end{align}
where the outcomes of performing this measurement are given by $(k_1,k_2) \in \integer_{m} \times \integer_{m}$. 

\item Alice then sends measurement outcomes $(k_1,k_2)$ to Bob.

\item Bob receives $(k_1,k_2)$ from Alice and applies the generalized Pauli operator 
	\begin{align}
		W_{k_1,k_2}^{(m)},
	\end{align}
as defined in equation~\eqref{eq:generalized-pauli-ops} to his register, $\reg{Y}$, which completes the protocol, and teleports $\reg{Z}$ to Bob. 
\end{enumerate}
To see why the state $\rho$ from Alice is teleported to Bob, one may consider a generalization of the case for $m = 2$. The scenario where $m = 2$ is the most standard teleportation setup, and has been covered, for instance, in~\cite{Nielsen2001}, whereas the generalization is covered in~\cite{Wilde2013}.

\section{The nonlocal game model}
\label{sec:the-nonlocal-games-model}

The nonlocal game model is built upon the notion of \index{interactive proof system}{\emph{interactive proof systems}}, initially introduced in~\cite{Goldwasser1985} and independently in~\cite{Babai1985}, and further studied in classical complexity theory~\cite{Ben-Or1988,Fortnow1989,Babai1991,Feige1991,Feige1994,Raz1998}. Informally, an interactive proof system is an abstract model of computation where two parties, referred to as the \index{prover (interactive proof system)}{\emph{prover}} and the \index{verifier (interactive proof system)}{\emph{verifier}}, exchange messages to determine the validity of a mathematical statement. The interactive proof system model was made more powerful in~\cite{Ben-Or1988}, where the authors introduced a multi-prover interactive proof system that consisted of at least two independent provers, and one verifier. When considering two provers, we refer to them by the names of \emph{Alice} and \emph{Bob}, and we call the verifier the \emph{referee}. We refer to a one-round multi-prover interactive proof system with at least two provers (Alice and Bob) that play cooperatively against a referee as a \index{nonlocal game}{\emph{nonlocal game}}. In~\cite{Cleve2004}, the authors formally introduced the notion of a nonlocal game where the provers may share entanglement. The nonlocal game model served to embody the notion of a \index{Bell inequality}{\emph{Bell inequality}}, an inequality that illustrated the inability of a local hidden variable theory to account for certain consequences of entanglement~\cite{Bell1964}. In~\cite{Clauser1969}, the authors Clauser, Horne, Shimony, and Holt presented a special type of Bell inequality that has since been named after the authors as the \index{CHSH inequality}{\emph{CHSH inequality}}. In~\cite{Cleve2004}, the CHSH inequality was first formulated in the language of nonlocal games. Nonlocal games have since been studied in the context of quantum information, and the result has been an active topic of research~\cite{Cleve2004, Brassard2005, Cleve2008, Doherty2008,Kempe2010,Kempe2010a,Kempe2011,Junge2011a,Buhrman2013,
Regev2013,Dinur2013,Vidick2013,Cleve2014}.

More formally, a nonlocal game begins by the referee selecting a pair of questions $(x,y)$ according to a fixed probability distribution that is known to all parties. The referee then sends question $x$ to Alice and question $y$ to Bob. While we assume that Alice and Bob may confer prior to the start of the game, when the game begins, the players are forbidden from communicating with each other. So Alice is unaware of the question that Bob received, and vice versa. Alice and Bob then respond to the referee with answers $a$ and $b$, respectively. Upon receiving these answers, the referee evaluates some predicate based on the questions and answers to determine whether Alice and Bob win or lose. In addition to having complete knowledge of the probability distribution used to select $x$ and $y$, we also assume that Alice and Bob have complete knowledge of the predicate.

The goal of Alice and Bob is to maximize their probability of obtaining a winning outcome. Prior to the start of the game, Alice and Bob may corroborate on a joint \index{strategy}{\emph{strategy}} to achieve this goal. One may consider a number of strategies for nonlocal games. For example, if Alice and Bob make use of classical resources, we call this a \emph{classical strategy}. In such a strategy, the players answer \emph{deterministically} with answers $a$ and $b$ determined by functions of $x$ and $y$ respectively. The players may also make use of randomness, but doing so provides no advantage over simply playing deterministically. 

Another type of strategy that the players may adopt are \emph{quantum strategies}. In a quantum strategy, Alice and Bob prepare and share a joint quantum system prior to the start of the game. We also assume that the players have local sets of measurement operators that they perform on their share of the state after the game has begun and they have received their questions from the referee to determine their answers $a$ and $b$. 

One may consider a number of sub-classifications of quantum strategies as well. For instance, the size of the shared quantum system may make a difference in how well Alice and Bob can perform, and indeed one can ask whether or not the size of the state yields any advantage. Another sub-classification of a quantum strategy is referred to as a \emph{commuting measurement strategy}. In this type of strategy, the bipartite tensor product structure of a shared quantum system between Alice and Bob is relaxed to one in which the local measurements of Alice and Bob pairwise commute. 

An even more general type of strategy that Alice and Bob may adopt is referred to as a \emph{non-signaling strategy}. In this type of strategy, the only constraint on Alice and Bob is that they cannot communicate during the game, but may make use of any type of resource, even possibly those outside of the scope of resources described by quantum mechanics. 

We refer to the \index{value (nonlocal game)}{\emph{value}} of a nonlocal game as the supremum value of the probability for the players to win over all strategies of a specified type. 

\subsection{Strategies for nonlocal games}
\label{sec:strategies-for-nonlocal-games}

\subsubsection*{Nonlocal games and correlation functions}
\label{sec:correlation-operators-for-nonlocal-games}

We specify a nonlocal game, $G$, as a pair $(\pi,V)$ where $\pi$ is a probability distribution of the form 
\begin{align}
	\pi : \SigmaA \times \SigmaB \rightarrow \left[0,1\right]
\end{align}
on the Cartesian product of two alphabets $\SigmaA$ and $\SigmaB$, and $V$ is a function of the form 
\begin{align}
	V : \GammaA \times \GammaB \times \SigmaA \times \SigmaB \rightarrow \left[0,1\right],
\end{align}
for $\SigmaA$ and $\SigmaB$ as above and $\GammaA$ and $\GammaB$ being alphabets. We use
\begin{align}
	\Sigma = \SigmaA \times \SigmaB \quad \textnormal{and} \quad \Gamma = \GammaA \times \GammaB
\end{align}
to denote the respective sets of questions asked to Alice and Bob and the sets of answers sent from Alice and Bob to the referee. 

For any type of strategy, the output probability distributions produced by Alice and Bob may be described by a function
\begin{align}
	C : \GammaA \times \GammaB \times \SigmaA \times \SigmaB \rightarrow [0,1],
\end{align}
where the function $C$ is referred to as a \index{correlation function (nonlocal game)}{\emph{correlation function}}. The entry $C(a,b|x,y)$ corresponds to the probability that Alice and Bob output $a \in \GammaA$ and $b \in \GammaB$ given the input $x \in \SigmaA$ and $y \in \SigmaB$. Since a correlation function represents a collection of probability distributions, the operator $C$ must satisfy 
\begin{align}
	\sum_{(a,b) \in \Gamma} C(a,b|x,y) = 1
\end{align}
for all $x \in \SigmaA$ and $y \in \SigmaB$. In particular, Alice and Bob's winning probability is represented as 
\begin{align} \label{eq:correlation-function}
	\sum_{(x,y) \in \Sigma} \pi(x,y) \sum_{(a,b) \in \Gamma} V(a,b|x,y) C(a,b|x,y),
\end{align}
where the correlation function is defined with respect to the corresponding strategy implemented by Alice and Bob. 

In the coming sections, we shall make the notions of the value of a nonlocal game and their corresponding strategies more concrete. 

\subsubsection*{Quantum strategies for nonlocal games}
\label{sec:quantum-strategies-for-nonlocal-games}

A \index{quantum strategy (nonlocal game)}{\emph{quantum strategy}} for a nonlocal game consists of complex Euclidean spaces $\U$ for Alice and $\V$ for Bob, a quantum state $\sigma \in \Density(\U \otimes \V)$ contained in registers $(\reg{U},\reg{V})$, and two collections of measurements, 
\begin{align}
	\{ A_a^x : a \in \GammaA \} \subset \Pos(\U) \quad \textnormal{and} \quad \{ B_b^y : b \in \GammaB \} \subset \Pos(\V), 
\end{align}
for each $x \in \SigmaA$ and $y \in \SigmaB$ respectively. The measurement operators satisfy the constraint that 
\begin{align}
	\sum_{a \in \GammaA} A_a^x = \I_{\U} \quad \textnormal{and} \quad \sum_{b \in \GammaB} B_b^y = \I_{\V}
\end{align}
for each $x \in \SigmaA$ and $y \in \SigmaB$. 

At the beginning of the game, Alice and Bob prepare a quantum system represented by the bipartite state $\sigma \in \Density(\U \otimes \V)$. The referee then selects questions $(x,y) \in \Sigma$ according to the probability distribution $\pi$ that is known to Alice, Bob, and the referee. The referee then sends $x$ to Alice and $y$ to Bob. Alice and Bob then generate answers $a \in \GammaA$ and $b \in \GammaB$, by making measurements on their portion of the state $\sigma$. That is to say, Alice makes a measurement on her part of $\sigma$ with respect to the measurement operators $\{ A_a^x : a \in \GammaA \}$. Similarly, Bob also performs a measurement on his part of $\sigma$ using the set of measurement operators $\{ B_b^y : b \in \GammaB \}$. The answers $(a,b)$ are then sent to the referee. The referee now possesses the questions $(x,y)$ in addition to the responses sent by Alice and Bob, $(a,b)$. The referee uses this information to evaluate the predicate $V(a,b|x,y)$, resulting in either a winning or losing outcome, represented by a $1$ or a $0$, respectively. A depiction of a nonlocal game is given in Figure~\ref{fig:nonlocal-game}. 

\begin{figure}[!htpb] 
	\begin{center}
		\includegraphics[scale=1.0]{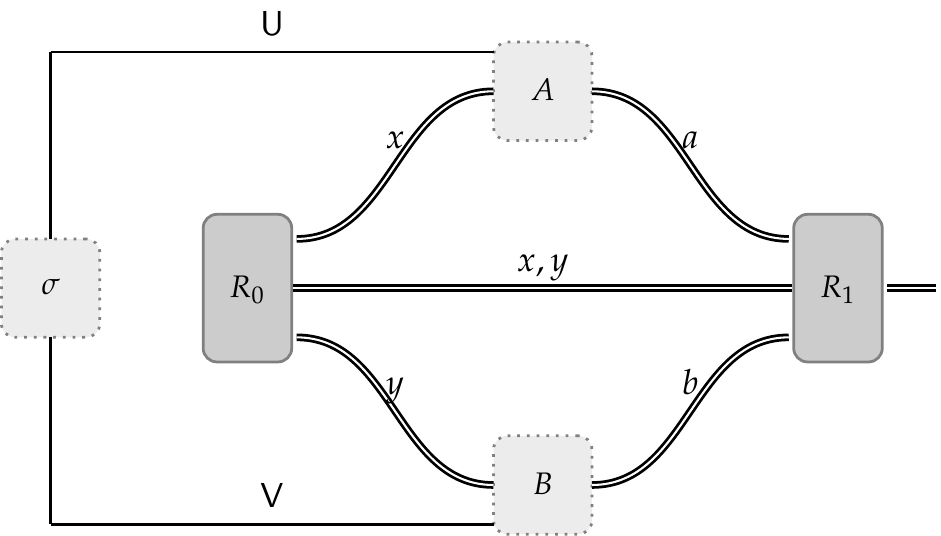}
	\end{center}
		\caption[A two-player nonlocal game.]{A two-player nonlocal game. In a nonlocal game, the players, Alice and Bob, first select a strategy. In the case of a quantum strategy, Alice and Bob may share a state $\sigma \in \Density(\U \otimes \V)$ in registers $(\reg{U},\reg{V})$. We assume that after this point, Alice and Bob are space-like separated and unable to communicate with each other for the remainder of the game. The referee then selects and sends questions $x \in \SigmaA$ for Alice and $y \in \SigmaB$ for Bob according to the publicly known probability distribution, $\pi$. The referee also keeps a copy of $x$ and $y$ after sending. Alice and Bob generate their answers $a \in \GammaA$ and $b \in \GammaB$ respectively, and send their answers to the referee, where the predicate $V(a,b|x,y)$ is computed to determine the probability that Alice and Bob win or lose.}
\label{fig:nonlocal-game}
\end{figure}

The winning probability for such a strategy in this game $G = (\pi,V)$ is given by equation~\eqref{eq:correlation-function} where $C$ is a \index{quantum correlation function (nonlocal game)}{\emph{quantum correlation function}} defined as 
\begin{align}
		C(a,b|x,y) = \bigip{A_a^x \otimes B_b^y}{\sigma},
\end{align}
for all $x \in \SigmaA$, $y \in \SigmaB$, $a \in \GammaA$, and $b \in \GammaB$. 

The \index{quantum value (nonlocal game)}{\emph{quantum value}} of a nonlocal game $G$, denoted as $\omega^*(G)$, is the supremum value of the winning probability of $G$ taken over all quantum strategies for Alice and Bob. We may also write $\omega_N^*(G)$ to denote the quantum value of $G$ when the dimension of Alice's space and the dimension of Bob's space is equal to $N$. Note that we can make the assumption on Alice and Bob's spaces that
\begin{align}
	\dim(\A) = \dim(\B),
\end{align} 
since whichever strategy Alice and Bob use, the probability of winning is always going to be maximized when $\sigma$ is a pure state. That is, Alice and Bob will not perform any better for any possible convex combination of $\sigma$, so we may as well assume $\sigma$ to be pure, that is $\sigma = uu^*$ for some nonzero vector $u \in \U \otimes \V$. It holds that one can always take the Schmidt decomposition of $u$, where it can be observed that the state is supported on spaces of equal dimension. 

We use $\Q_N(\GammaA, \GammaB | \SigmaA, \SigmaB)$ to denote the set of all quantum correlation functions when the dimension of Alice and Bob's system is equal to $N$.

\subsubsection*{Classical strategies for nonlocal games}
\label{sec:classical-strategies-for-nonlocal-games}

A \index{classical strategy (nonlocal game)}{\emph{classical strategy}} for a nonlocal game consists of functions $f : \SigmaA \rightarrow \GammaA$ and $g : \SigmaB \rightarrow \GammaB$ that deterministically produce an output for every input. This type of classical strategy is referred to as a \index{deterministic strategy}{\emph{deterministic strategy}}, as the outputs are produced deterministically. Provided that we are interested in maximizing the winning probability, there is no loss in generality in restricting our attention to deterministic strategies for any classical strategy, as the classical value of any nonlocal game will always be obtained by such a deterministic strategy. This can be observed by the fact that any probabilistic strategy may be expressed as a convex combination of deterministic strategies, so Alice and Bob gain no benefit from using randomness. In other words, the average is never bigger than the maximum. The winning probability for such a strategy in this game $G = (\pi,V)$ is given by
\begin{align}
	\sum_{(x,y) \in \Sigma} \pi(x,y) \sum_{(a,b) \in \Gamma} V(a,b|x,y) C(a,b|x,y),
\end{align}
where $C$ is the \index{deterministic correlation function (nonlocal game)}{\emph{deterministic correlation function}} defined as 
\[
 C(a,b|x,y) =
  \begin{cases} 
      \hfill 1 \hfill & \text{ if $a = f(x) \textnormal{ and } b = g(y)$,} \\
      \hfill 0 \hfill & \text{ otherwise,} \\
  \end{cases}
\]
for all $x \in \SigmaA$, $y \in \SigmaB$, $a \in \GammaA$, and $b \in \GammaB$. We use $\L(\GammaA,\GammaB|\SigmaA,\SigmaB)$ to denote the set of all deterministic correlation functions, including all convex combinations of deterministic correlation functions as well.

The \index{classical value (nonlocal game)}{\emph{classical value}} of a nonlocal game $G$, denoted as $\omega(G)$ is the supremum value of the winning probability of $G$ taken over all classical strategies for Alice and Bob. As argued above the supremum value is necessarily achieved by some deterministic strategy, and therefore we may write $\omega(G)$ as 
\begin{align}
	\omega(G) = \max_{f,g} \sum_{(x,y) \in \Sigma} \pi(x,y) V(f(x),g(y)|x,y),
\end{align}
where the maximum is over all functions $f: \SigmaA \rightarrow \GammaA$ and $g: \SigmaB \rightarrow \GammaB$. 

\subsubsection*{Commuting measurement strategies for nonlocal games}
\label{sec:commuting-measurement-strategies-for-nonlocal-games}

A \index{commuting measurement strategy (nonlocal game)}{\emph{commuting measurement strategy}} consists of a single (possibly infinite-dimensional) Hilbert space, $\H$, a quantum state $\sigma \in \Density(\H)$, and two collections of measurements, 
\begin{align}
	\{ A_a^x : a \in \GammaA \} \subset \Pos(\H) \quad \textnormal{and} \quad \{ B_b^y : b \in \GammaB \} \subset \Pos(\H),
\end{align}
such that 
\begin{align}
	\sum_{a \in \GammaA} A_a^x = \sum_{b \in \GammaB} B_b^y = \I_{\H}
\end{align}
for all $x \in \SigmaA$ and $y \in \SigmaB$, and that satisfy
\begin{align}
	\left[ A_a^x, B_b^y \right] = 0
\end{align}
for all $x \in \SigmaA, y \in \SigmaB, a \in \GammaA,$ and $b \in \GammaB$. For a nonlocal game, $G = (\pi,V)$, the winning probability for a commuting measurement strategy is given by equation~\eqref{eq:correlation-function} where $C$ is a \index{commuting measurement correlation function (nonlocal game)}{\emph{commuting measurement correlation function}} defined as 
\begin{align}
	C(a,b|x,y) = \bigip{A_a^x B_b^y}{\sigma}
\end{align}
for all $x \in \SigmaA$, $y \in \SigmaB$, $a \in \GammaA$, and $b \in \GammaB$. We use $\C(\GammaA, \GammaB | \SigmaA, \SigmaB)$ to denote the set of all commuting measurement correlation function. 

The \index{commuting measurement value (nonlocal game)}{\emph{commuting measurement value}} of a nonlocal game $G$, denoted as $\omega_c(G)$, is the supremum value of the winning probability of $G$ taken over all commuting measurement strategies for Alice and Bob. Elsewhere in the literature, the commuting measurement value is also referred to as the field-theoretic value~\cite{Doherty2008}. 

\subsubsection*{Non-signaling strategies for nonlocal games}
\label{sec:non-signaling-strategies-for-nonlocal-games}

For a nonlocal game $G = (\pi,V)$, the winning probability for a \index{non-signaling strategy}{\emph{non-signaling strategy}} is given by equation~\eqref{eq:correlation-function} where $C$ is a \index{non-signaling correlation function (nonlocal game)}{\emph{non-signaling correlation function}} that satisfies the following non-signaling properties
\begin{align} \label{eq:ns-constraints-1}
	\sum_{b \in \GammaB} C(a,b|x,y) = \sum_{b \in \GammaB} C(a,b|x,y^{\prime}),
\end{align}
for all $a \in \GammaA$, $x \in \SigmaA$, $y \in \SigmaB$, and $y^{\prime} \in \SigmaB$ and 
\begin{align} \label{eq:ns-constraints-2}
	\sum_{a \in \GammaA} C(a,b|x,y) = \sum_{a \in \GammaA} C(a,b|x^{\prime},y),
\end{align}
for all $b \in \GammaB$, $x \in \SigmaA$, $x^{\prime} \in \SigmaA$, and $y \in \SigmaB$ and where $C$ is normalized and nonnegative. We use $\NS(\GammaA, \GammaB | \SigmaA, \SigmaB)$ to denote the set of all non-signaling correlation functions. 

The \index{non-signaling value (nonlocal game)}{\emph{non-signaling value}} of a nonlocal game, $G$, denoted as $\omega_{\ns}(G)$, is the supremum value of the winning probability of $G$ taken over all non-signaling strategies for Alice and Bob.

If one wishes, one may even consider a more general type of strategy, indeed the most general strategy one may consider in the realm of nonlocal games. This most general type of strategy, referred to as a \index{global strategy}{\emph{global strategy}} is one in which the correlation functions need only satisfy
\begin{align}
	\sum_{(a,b) \in \Gamma} C(a,b|x,y) = 1, 
\end{align}
for all $x \in \SigmaA$ and $y \in \SigmaB$ and that the entries of $C$ be nonnegative. Indeed, these two constraints are in all of the strategies we have considered thus far, as they are implicit from the definition of a correlation function from Section~\ref{sec:correlation-operators-for-nonlocal-games}. Another way to think about non-signaling strategies therefore is to consider them as strategies that satisfy these two implicit restrictions of a global strategy, as well as the non-signaling constraints from equations~\eqref{eq:ns-constraints-1}~and~\eqref{eq:ns-constraints-2}.  

\subsection{Relationships between different strategies and values} \label{sec:relationships-between-different-strategies-and-values}


In order to determine how well the players can expect to do for a particular choice of strategy, we consider the corresponding values for each strategy. There exist algorithms that allow one to calculate the classical and non-signaling values of an arbitrary nonlocal game by optimizing over the respective classical and non-signaling correlation functions~\cite{Brunner2014}. These algorithms are not particularly efficient however, as there are exponentially many possible functions for Alice and Bob to consider. In general, with the exception of a specific class of nonlocal games~\cite{Cleve2008}, there is no known efficient algorithm to exactly compute the quantum value of an arbitrary nonlocal game. There is, however, an approach that allows one to approximate the quantum values of arbitrary nonlocal games~\cite{Doherty2008,Navascues2007,Navascues2008}, a technique we will investigate in greater detail in Chapter~\ref{chap:extended_npa_hierarchy}.  

The sets of correlation functions for the strategies we have covered thus far have the following relationship 
\begin{align}
	\L(\GammaA, \GammaB | \SigmaA, \SigmaB) \subseteq \Q(\GammaA, \GammaB | \SigmaA, \SigmaB) \subseteq \C(\GammaA, \GammaB | \SigmaA, \SigmaB) \subseteq \NS(\GammaA, \GammaB | \SigmaA, \SigmaB),
\end{align}
for alphabets $\GammaA, \GammaB, \SigmaA$ and $\SigmaB$. The relationship of $\L(\GammaA, \GammaB | \SigmaA, \SigmaB) \subseteq \Q(\GammaA, \GammaB | \SigmaA, \SigmaB)$ follows since Alice and Bob could use their shared entangled state only as a source of shared randomness. Recall, that Alice and Bob gain no benefit from using randomness in a classical strategy, so one may restrict attention to classical strategies defined in terms of deterministic ones. Should Alice and Bob use their quantum state in a quantum strategy as a source of shared randomness, this is no better than having them use a classical strategy, and gives the relationship between correlation functions. The relationship that $\Q(\GammaA, \GammaB | \SigmaA, \SigmaB) \subseteq \C(\GammaA, \GammaB | \SigmaA, \SigmaB)$ holds due to the fact that bipartite operators where the identity operator is on either side of the operator obey the commutation relationship, that is
\begin{align}
	\left[ A_a^x \otimes \I_{\B}, \I_{\A} \otimes B_b^y \right] = 0
\end{align}
for sets of operators $\{A_a^x : a \in \GammaA\}$ and $\{B_b^y : b \in \GammaB\}$ over all $x \in \SigmaA$, $y \in \SigmaB$, $a \in \GammaA$, and $b \in \GammaB$. The relationship that $\Q(\GammaA, \GammaB | \SigmaA, \SigmaB) \subseteq \NS(\GammaA, \GammaB | \SigmaA, \SigmaB)$ comes from observing that for a commuting measurement correlation function
\begin{align}
	C(a,b|x,y) = \bigip{A_a^x B_b^y}{\sigma}
\end{align}
we have that
\begin{align}
	\sum_{b \in \GammaB} C(a,b|x,y) = \sum_{b \in \GammaB} \bigip{A_a^x B_b^y}{\sigma} = \bigip{A_a^x}{\sigma},
\end{align}
or in other words, that there is no dependence on $y$. Similarly, we have that
\begin{align}
	\sum_{a \in \GammaA} C(a,b|x,y) = \sum_{a \in \GammaA} \bigip{A_a^x B_b^y}{\sigma} = \bigip{B_b^y}{\sigma}.
\end{align}


Given that the correlation functions obey these relationships, it then follows that the corresponding values of these operators must also satisfy a similar inequality relationship
\begin{align}
	0 \leq \omega(G) \leq \omega^*(G) \leq \omega_c(G) \leq \omega_{\ns}(G) \leq 1. 
\end{align}

\chapter{Extended Nonlocal Games}
\label{chap:extended_nonlocal_games}

In this chapter, we introduce the \emph{extended nonlocal game} model. This model is a generalization of the nonlocal game model in which the referee now also holds a quantum system provided to it by Alice and Bob at the start of the game. In Section~\ref{sec:the-model-of-extended-nonlocal-games} we shall present the extended nonlocal game protocol, and in Section~\ref{sec:strategies-extended-nonlocal-games}, we define the corresponding strategies that Alice and Bob may adopt during the course of the game. 

The general notion of extended nonlocal games was previously considered by Fritz~\cite{Fritz2012}. In particular, Fritz considered a class of games, called \emph{bipartite steering games}, which are essentially extended nonlocal games in which the referee randomly chooses to ask either Alice or Bob a question. Extended nonlocal games may also be viewed as being equivalent to multipartite steering inequalities, in a similar way to the equivalence between nonlocal games and Bell inequalities. Multipartite steering inequalities and related notions were studied in the papers~\cite{Cavalcanti2015} and~\cite{Sainz2015}. The term ``extended nonlocal game'' along with a treatment more focused in the nonlocal game setting was carried out in~\cite{Johnston2015a}. 

This chapter is based on joint work with Nathaniel Johnston, Rajat Mittal, and John Watrous~\cite{Johnston2015a}

\minitoc

\section{The extended nonlocal game model}
\label{sec:the-model-of-extended-nonlocal-games}

Extended nonlocal games are a generalization of nonlocal games in which the \emph{referee also holds a quantum system}, provided to it by Alice and Bob at the start of the game. Similar to an ordinary nonlocal game, one may consider a variety of possible strategies for Alice and Bob in an extended nonlocal game. In particular, there are classes of strategies that are analogous to classical, quantum, commuting measurement, and non-signaling strategies from the nonlocal game model. Further details on how these are adapted for the case of extended nonlocal games will be elaborated on in this chapter. 

An \index{extended nonlocal game}{\emph{extended nonlocal game}} is similar to a nonlocal game in the sense that it is a cooperative game played between two players, Alice and Bob, against a referee. The game begins much like a nonlocal game, with the referee selecting and sending a pair of questions $(x,y)$ according to a fixed probability distribution. Once Alice and Bob receive $x$ and $y$, they respond with respective answers $a$ and $b$. Unlike a nonlocal game, the outcome of an extended nonlocal game is determined by measurements performed by the referee on its share of the state initially provided to it by Alice and Bob. Specifically, Alice and Bob's winning probability is determined by a collection of measurements, $V(a,b|x,y) \in \Pos(\R)$, where $\R = \complex^m$ is a complex Euclidean space with $m$ denoting the dimension of the referee's quantum system---so if Alice and Bob's response $(a,b)$ to the question pair $(x,y)$ leaves the referee's system in the quantum state 
\begin{align}
	\sigma_{a,b}^{x,y} \in \Density(\R),
\end{align}
then their winning and losing probabilities are given by 
\begin{align}
	\biggip{V(a,b|x,y)}{\sigma_{a,b}^{x,y}} \quad \textnormal{and} \quad \biggip{\I - V(a,b|x,y)}{\sigma_{a,b}^{x,y}}.
\end{align}


\section{Strategies for extended nonlocal games} \label{sec:strategies-extended-nonlocal-games}

\subsection{Extended nonlocal games and assemblage operators} \label{sec:extended-nonlocal-games-and-assemblage-operators}


An extended nonlocal game $H$ is defined by a pair $(\pi,V)$, where $\pi$ is a probability distribution of the form 
\begin{align}
	\pi : \SigmaA \times \SigmaB \rightarrow \left[0,1\right]
\end{align}
on the Cartesian product of two alphabets $\SigmaA$ and $\SigmaB$, and $V$ is a function of the form 
\begin{align}
	V : \GammaA \times \GammaB \times \SigmaA \times \SigmaB \rightarrow \Pos(\R), 
\end{align}
for $\SigmaA$ and $\SigmaB$ as above, $\GammaA$ and $\GammaB$ being alphabets, and $\R$ refers to the referee's space. Just as in the case for nonlocal games, we shall use the convention that
\begin{align}
	\Sigma = \SigmaA \times \SigmaB \quad \textnormal{and} \quad \Gamma = \GammaA \times \GammaB
\end{align}
to denote the respective sets of questions asked to Alice and Bob and the sets of answers sent from Alice and Bob to the referee. 

When analyzing a strategy for Alice and Bob, it may be convenient to define a function 
\begin{align}
	K : \GammaA \times \GammaB \times \SigmaA \times \SigmaB \rightarrow \Pos(\R).
\end{align}
We will refer to the function $K$ as an \index{assemblage}{\emph{assemblage}}. The operators output by this function represent the \emph{unnormalized} states of the referee's quantum system when Alice and Bob respond to the question pair $(x,y)$ with the answer pair $(a,b)$. 


We can however, if we wish, normalize these states by noting that the quantity $\tr\left(K(a,b|x,y)\right)$ refers to the probability with which Alice and Bob answer $(a,b)$ for the question pair $(x,y)$. Assuming that $\tr \left( K(a,b|x,y) \right) > 0$, we define a set of normalized states 
\begin{align} \label{eq:enlg-conditional-states}
	\sigma_{a,b}^{x,y} = \frac{K(a,b|x,y)}{\tr \left( K(a,b|x,y) \right)}
\end{align}
of the referee's system conditioned on this question and answer pair. Note that the function $K$ completely determines the performance of Alice and Bob's strategy for $H$ as it encodes the probability that Alice and Bob obtain answers $a \in \GammaA$ and $b \in \GammaB$ given questions $x \in \SigmaA$ and $y \in \SigmaB$ as
\begin{align}
	\tr \left( K(a,b|x,y)	\right),
\end{align} 
along with the conditional states from equation~\eqref{eq:enlg-conditional-states}. In particular, Alice and Bob's winning probability is represented as 
\begin{align} \label{eq:quantum-winning-probability-kabxy}
	\sum_{(x,y) \in \Sigma} \pi(x,y) \sum_{(a,b) \in \Gamma} \biggip{V(a,b|x,y)}{K(a,b|x,y)}. 
\end{align}

\subsection{Standard quantum strategies for extended nonlocal games} \label{sec:standard-quantum-strategies-extended-nonlocal-games}

A \index{standard quantum strategy (extended nonlocal game)}{\emph{standard quantum strategy}} for an extended nonlocal game consists of finite-dimensional complex Euclidean spaces $\U$ for Alice and $\V$ for Bob, a quantum state $\sigma \in \Density(\U \otimes \R \otimes \V)$, and two collections of measurements, 
\begin{align}
	\{ A_a^x : a \in \GammaA \} \subset \Pos(\U) \quad \textnormal{and} \quad \{ B_b^y : b \in \GammaB \} \subset \Pos(\V),
\end{align}
for each $x \in \SigmaA$ and $y \in \SigmaB$ respectively. As usual, the measurement operators satisfy the constraint that 
\begin{align}
	\sum_{a \in \GammaA} A_a^x = \I_{\U} \quad \textnormal{and} \quad \sum_{b \in \GammaB} B_b^y = \I_{\V},
\end{align}
for each $x \in \SigmaA$ and $y \in \SigmaB$. 

When the game is played, Alice and Bob present the referee with a quantum system so that the three parties share the state $\sigma \in \Density(\U \otimes \R \otimes \V)$. The referee selects questions $(x,y) \in \Sigma$ according to the distribution $\pi$ that is known to all participants in the game. The referee then sends $x$ to Alice and $y$ to Bob. At this point, Alice and Bob make measurements on their respective portions of the state $\sigma$ using their measurement operators to yield an outcome to send back to the referee. Specifically, Alice measures her portion of the state $\sigma$ with respect to her set of measurement operators $\{A_a^x : a \in \GammaA\}$, and sends the result $a \in \GammaA$ of this measurement to the referee. Likewise, Bob measures his portion of the state $\sigma$ with respect to his measurement operators $\{ B_b^y : b \in \GammaB \}$ to yield the outcome $b \in \GammaB$, that is then sent back to the referee. At the end of the protocol, the referee measures its quantum system with respect to the measurement $\{ V(a,b|x,y), \I - V(a,b|x,y) \}$. 

\begin{figure}[!htpb] \label{fig:extended-nonlocal-game}
	\begin{center}
		\includegraphics[scale=0.9]{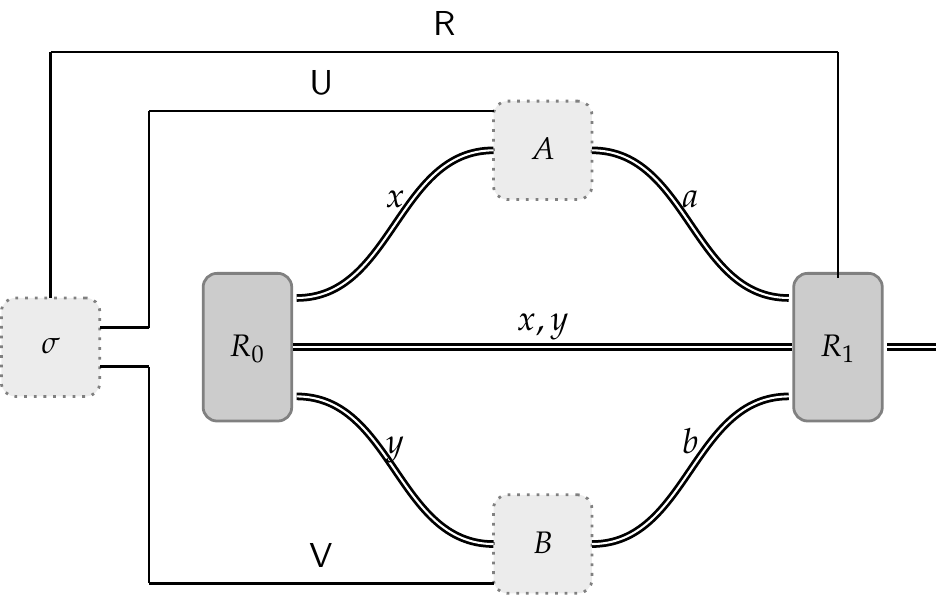}
	\end{center}
		\caption[A two-player extended nonlocal game.]{A two-player extended nonlocal game. Alice, Bob, and the referee all share a tripartite state, $\sigma \in \Density(\U \otimes \R \otimes \V)$, contained in registers $(\reg{U},\reg{R},\reg{V})$. The referee selects questions $(x,y) \in \Sigma$ according to the probability distribution $\pi$, and sends $x$ to Alice and $y$ to Bob. Upon receiving $x$ and $y$, Alice and Bob respond with answers $a \in \GammaA$ and $b \in \GammaB$. After receiving $a$ and $b$, the referee performs a measurement on its system $\{ V(a,b|x,y), \I - V(a,b|x,y) \}$ to determine the probability with which Alice and Bob win the game.}
\end{figure}

The winning probability for such a strategy in this game $H = (\pi,V)$ is given by 
\begin{align} \label{eq:enlg-quantum-winning-probability}
	\sum_{(x,y) \in \Sigma} \pi(x,y) \sum_{(a,b) \in \Gamma} \biggip{A_a^x \otimes V(a,b|x,y) \otimes B_b^y}{\sigma},
\end{align}  
or equivalently the winning probability for such a strategy is given by 
\begin{align}
	\sum_{(x,y) \in \Sigma} \pi(x,y) \sum_{(a,b) \in \Gamma} \biggip{V(a,b|x,y)}{K(a,b|x,y)},
\end{align}
where the operator $K : \GammaA \times \GammaB \times \SigmaA \times \SigmaB \rightarrow \Pos(\R)$ is a \index{standard quantum assemblage}{\emph{standard quantum assemblage}} operator defined as 
\begin{align}
	K(a,b|x,y) = \tr_{\U \otimes \V} \left( \left( A_a^x \otimes \I_{\R} \otimes B_b^y \right) \sigma \right).
\end{align}
This may be observed by noting that 
\begin{align}
	 &\sum_{(x,y) \in \Sigma} \pi(x,y) \sum_{(a,b) \in \Gamma} \biggip{V(a,b|x,y)}{K(a,b|x,y)} \\
	 &=\sum_{(x,y) \in \Sigma} \pi(x,y) \sum_{(a,b) \in \Gamma} \biggip{V(a,b|x,y)}{\tr_{\U \otimes \V} \left( \left(  A_a^x \otimes \I_{\R} \otimes B_b^y \right) \sigma \right)} \\
	 &=\sum_{(x,y) \in \Sigma} \pi(x,y) \sum_{(a,b) \in \Gamma} \tr \left( V(a,b|x,y) \tr_{\U \otimes \V} \left( \left( A_a^x \otimes \I_{\R} \otimes B_b^y \right) \sigma \right) \right)  \label{eq:ip-to-tr-1} \\
	 &=\sum_{(x,y) \in \Sigma} \pi(x,y) \sum_{(a,b) \in \Gamma} \tr \left( \left( A_a^x \otimes V(a,b|x,y) \otimes B_b^y \right) \sigma \right) \label{eq:factor-out-pt} \\
	 &=\sum_{(x,y) \in \Sigma} \pi(x,y) \sum_{(a,b) \in \Gamma} \biggip{ A_a^x \otimes V(a,b|x,y) \otimes B_b^y}{\sigma}, \label{eq:ip-to-tr-2} 
\end{align}
where in equations~\eqref{eq:ip-to-tr-1} and~\eqref{eq:ip-to-tr-2}, we used the relationship between the inner product and trace operations, and in equation~\eqref{eq:factor-out-pt}, we factor out the partial trace operator from the overall trace. 

For any strategy, there is an equivalent strategy where $\sigma$ is a pure state and the sets of measurements that Alice and Bob possess are projective operators. This can be shown through a two step process. First, either party may purify the state. It makes no difference whether Alice or Bob hold the purification, but for the sake of argument, we assume that Alice purifies the state. Second, the non-projective measurements can be simulated by projective measurements in a standard way that is described by Naimark's theorem~\cite{Paulsen2003}.

For a given extended nonlocal game $H = (\pi,V)$, we write $\omega^*(H)$ to denote the \index{standard quantum value}{\emph{standard quantum value}} of $H$, which is the supremum value of Alice and Bob's winning probability over all standard quantum strategies for $H$. We may wish to consider the standard quantum value of $H$ when the dimension on Alice's space and Bob's space are equal to $N$, which we denote as $\omega_N^*(H)$. 

\subsection{Unentangled strategies for extended nonlocal games} \label{sec:unentangled-strategies-extended-nonlocal-games}

An \index{unentangled strategy (extended nonlocal game)}{\emph{unentangled strategy}} for an extended nonlocal game is simply a standard quantum strategy for which the state $\sigma \in \Density(\U \otimes \R \otimes \V)$ initially prepared by Alice and Bob is fully separable. Equivalently, there exists an alphabet $\Delta$ and collections of states 
\begin{align}
	\{ \sigma_j^{\reg{U}} : j \in \Delta \} \subseteq \Density(\U), \quad \{ \sigma_j^{\reg{R}} : j \in \Delta \} \subseteq \Density(\R), \quad \textnormal{and} \quad \{ \sigma_j^{\reg{V}} : j \in \Delta \} \subseteq \Density(\V),
\end{align}
and a probability vector $p \in \P(\Delta)$ such that 
\begin{align}
	\sigma = \sum_{j \in \Delta} p(j) \sigma_j^{\reg{U}} \otimes \sigma_j^{\reg{R}} \otimes \sigma_j^{\reg{V}}.
\end{align} 

Note that any unentangled strategy is equivalent to a strategy where Alice and Bob store only classical information after the referee's quantum system has been provided to it. This is because the state that Alice and Bob share between themselves and the referee is fully separable, that is, there are no quantum correlations that may arise between the constituent subsystems held by the parties. Alice and Bob are therefore justified in following a deterministic strategy on their local systems in a similar way that was considered in classical strategies for nonlocal games.  


Furthermore, any such strategy is equivalent to one given by a convex combination of deterministic strategies, in which Alice and Bob initially provide the referee with a fixed pure state $\sigma = uu^* \in \Density(\R)$, and respond to questions deterministically, with Alice responding to $x \in \SigmaA$ with $a = f(x)$ and Bob responding to $y \in \SigmaB$ with $b = g(y)$ for functions $f: \SigmaA \rightarrow \GammaA$ and $g : \SigmaB \rightarrow \GammaB$.


For a given extended nonlocal game $H = (\pi,V)$, we write $\omega(H)$ to denote the \index{unentangled value}{\emph{unentangled value}} of $H$, which is the supremum value for Alice and Bob's winning probability in $H$ over all unentangled strategies. It follows by convexity and compactness that this supremum value is necessarily achieved by some deterministic strategy. The unentangled value for such a game is therefore given by
\begin{align} \label{eq:enlg-unentangled-value}
\omega(G) = \max_{f,g} \biggnorm{ \sum_{(x,y) \in \Sigma} \pi(x,y) V(f(x),g(y)|x,y) },
\end{align}
where the maximum is over all functions $f : \SigmaA \rightarrow \GammaA$ and $g: \SigmaB \rightarrow \GammaB$. 

\subsection{Commuting measurement strategies for extended nonlocal games} \label{sec:commuting-measurement-strategies-extended-nonlocal-games}

A \index{commuting measurement strategy (extended nonlocal game)}{\emph{commuting measurement strategy}} for an extended nonlocal game consists of a single (possibly infinite-dimensional) Hilbert space, $\H$, a quantum state $\sigma \in \Density(\R \otimes \H)$, and two collections of measurements, 
\begin{align}
	\{ A_a^x : a \in \GammaA \} \subset \Pos(\H) \quad \textnormal{and} \quad \{ B_b^y : b \in \GammaB \} \subset \Pos(\H),
\end{align}
such that 
\begin{align}
	\sum_{a \in \GammaA} A_a^x = \sum_{b \in \GammaB} B_b^y = \I_{\H}
\end{align}
for all $x \in \SigmaA$ and $y \in \SigmaB$ and that
\begin{align}
	\left[ A_a^x, B_b^y \right] = 0
\end{align}
for all $x \in \SigmaA, y \in \SigmaB, a \in \GammaA,$ and $b \in \GammaB$. 

For an extended nonlocal game, $H = (\pi,V)$, the winning probability for a commuting measurement strategy is given by 
\begin{align}
	\sum_{(x,y) \in \Sigma} \pi(x,y) \sum_{(a,b) \in \Gamma} \biggip{V(a,b|x,y) \otimes A_a^x B_b^y}{\sigma},
\end{align}
or equivalently the winning probability for such a strategy is given by 
\begin{align}
	\sum_{(x,y) \in \Sigma} \pi(x,y) \sum_{(a,b) \in \Gamma} \biggip{V(a,b|x,y)}{K(a,b|x,y)},
\end{align}
where the operator $K : \GammaA \times \GammaB \times \SigmaA \times \SigmaB$ is a \index{commuting measurement assemblage}{\emph{commuting measurement assemblage}} operator defined as 
\begin{align}
	K(a,b|x,y) = \tr_{\H} \left( \left( \I_{\R} \otimes A_a^x B_b^y  \right) \sigma \right).
\end{align}
This may be observed by noting that
\begin{align}
	 &\sum_{(x,y) \in \Sigma} \pi(x,y) \sum_{(a,b) \in \Gamma} \biggip{V(a,b|x,y)}{K(a,b|x,y)} \\
	 &=\sum_{(x,y) \in \Sigma} \pi(x,y) \sum_{(a,b) \in \Gamma} \biggip{V(a,b|x,y)}{\tr_{\H} \left( \left( \I_{\R} \otimes A_a^x B_b^y \right) \sigma \right)} \\
	 &=\sum_{(x,y) \in \Sigma} \pi(x,y) \sum_{(a,b) \in \Gamma} \tr \left( V(a,b|x,y) \tr_{\H} \left( \left( \I_{\R} \otimes A_a^x B_b^y  \right) \sigma \right) \right)  \label{eq:ip-to-tr-1-com} \\
	 &=\sum_{(x,y) \in \Sigma} \pi(x,y) \sum_{(a,b) \in \Gamma} \tr \left( \left( V(a,b|x,y) \otimes A_a^x B_b^y \right) \sigma \right) \label{eq:factor-out-pt-com} \\
	 &=\sum_{(x,y) \in \Sigma} \pi(x,y) \sum_{(a,b) \in \Gamma} \biggip{V(a,b|x,y) \otimes A_a^x B_b^y  }{\sigma}, \label{eq:ip-to-tr-2-com}
\end{align}
where the analysis follows in a similar manner to the case of standard quantum strategies for extended nonlocal games as described in Section~\ref{sec:standard-quantum-strategies-extended-nonlocal-games}. 

The \index{commuting measurement value (extended nonlocal game)}{\emph{commuting measurement value}} of $H$, which is denoted $\omega_c(H)$, is the supremum value of the winning probability of $H$ taken over all commuting measurement strategies for Alice and Bob. 



\subsection{Non-signaling strategies for extended nonlocal games} \label{sec:non-signaling-strategies-extended-nonlocal-games}

A \index{non-signaling strategy (extended nonlocal game)}{\emph{non-signaling strategy}} for an extended nonlocal game consists of a function
\begin{align}
	K: \GammaA \times \GammaB \times \SigmaA \times \SigmaB \rightarrow \Pos(\R)
\end{align}
such that
\begin{align} \label{eq:enlg-ns-assemblage}
	\sum_{a \in \GammaA} K(a,b|x,y) = \xi_b^y \quad \textnormal{and} \quad \sum_{b \in \GammaB} K(a,b|x,y) = \rho_a^x,
\end{align}
for all $x \in \SigmaA$ and $y \in \SigmaB$ where $\left\{ \xi_b^y : y \in \SigmaB, \ b \in \GammaB \right \}$ and $\left \{ \rho_a^x : x \in \SigmaA, \ a \in \GammaA \right \}$ are collections of operators satisfying 
\begin{align}
	\sum_{a \in \GammaA} \rho_a^x = \tau = \sum_{b \in \GammaB} \xi_b^y,
\end{align}
for every choice of $x \in \SigmaA$ and $y \in \SigmaB$ and where $\tau \in \Density(\R)$ is a density operator. We refer to the function $K$ satisfying equation~\eqref{eq:enlg-ns-assemblage} as a \index{non-signaling assemblage}{\emph{non-signaling assemblage}}. For any extended nonlocal game, $H = (\pi,V)$, the winning probability for a non-signaling strategy is given by
\begin{align}
	\sum_{(x,y) \in \Sigma} \pi(x,y) \sum_{(a,b) \in \Gamma} \biggip{V(a,b|x,y)}{K(a,b|x,y)},
\end{align}
where $K(a,b|x,y)$ is a non-signaling assemblage. The \index{non-signaling value (extended nonlocal game)}{\emph{non-signaling value}} of $H$, which is denoted as $\omega_{\ns}(H)$, is the supremum value of the winning probability of $H$ taken over all non-signaling strategies for Alice and Bob. Note that the supremum is achieved since the set of non-signaling assemblages is compact which implies the that the supremum is achieved. 

\subsubsection*{Relationships between different strategies and values}

It is worth noting that the same inequality chain that holds for nonlocal games also holds for extended nonlocal games, 
\begin{align}
	0 \leq \omega(H) \leq \omega^*(H) \leq \omega_c(H) \leq \omega_{\ns}(H) \leq 1. 
\end{align}
Due to the similarity in definitions of strategies, this line of reasoning is nearly identical to that of Section~\ref{sec:relationships-between-different-strategies-and-values}. 

\chapter{On the properties of the extended nonlocal game model}
\label{chap:infinite_entanglement}

This chapter is focused on studying the relationship between \emph{quantum-classical games} and extended nonlocal games. In Section~\ref{sec:quantum-classical-games}, we formally define the model of quantum-classical games, which is a variant of an ordinary nonlocal game, where now, in this model, the referee sends quantum registers to Alice and Bob in place of sending classical messages. This variant was considered by Buscemi~\cite{Buscemi2012} under the name of \emph{semi-quantum games}, and was also considered by Regev and Vidick~\cite{Regev2013}, where they studied a class of quantum-classical games, referred to as \emph{quantum XOR games}, where the winning condition is predicated upon an XOR function. 

One of the main results of Regev and Vidick's paper was to show that there exists a class of quantum XOR games for which no finite-dimensional quantum strategy can be optimal. In Section~\ref{sec:constructing-extended-nonlocal-games-from-quantum-classical-games}, we analyze this result in the context of extended nonlocal games, and building on their framework, show that there also exists a class of extended nonlocal games where no finite-dimensional quantum strategy can be optimal. We then use the relationship between extended nonlocal games and tripartite steering to arrive at the result that there exists a tripartite steering inequality for which an infinite-dimensional quantum state is required in order to achieve a maximal violation. From this, we conclude that there exists extended nonlocal games for which no finite-dimensional standard quantum strategy can be optimal. 

Finally, in Section~\ref{sec:variations-on-enlg}, we consider variants on the extended nonlocal game model. As we have covered, an extended nonlocal game is composed of three rounds of communication; where the type of communication in the first round from Alice and Bob to the referee is quantum, and the remaining two question and answer rounds are composed of classical communication. We ask here what happens if we exchange the type of communication for certain rounds and investigate these variations on the extended nonlocal game model. 

This chapter is based on joint work with John Watrous in~\cite{Russo2016}.

\minitoc

\section{Quantum-classical games}  \label{sec:quantum-classical-games}

\index{quantum-classical games (QC games)}{\emph{Quantum-classical games}} or \emph{QC games} for short, differ from nonlocal games in that the referee begins the game by preparing a tripartite quantum state and sends one part of it to each player, keeping a part of the state for itself. (This step replaces the generation of a classical question pair $(x,y)$ in an ordinary nonlocal game.) Once the players receive their portion of the tripartite state in a QC game, the players respond with classical answers $a$ and $b$ (as they would in a nonlocal game as well), and finally the referee determines whether the players win or lose by measuring its part of the original quantum state it initially prepared. (This step replaces the evaluation of a predicate $V(a,b|x,y)$ in an ordinary nonlocal game.) Games of this form, with slight variations from the general class just described, were considered by Buscemi~\cite{Buscemi2012} and Regev and Vidick~\cite{Regev2013}. 

Formally, a quantum-classical game (QC game) is specified by the following objects:
\begin{itemize}
	\item A state $\rho \in \Density(\X \otimes \S \otimes \Y)$ of a triple of registers $(\reg{X},\reg{S},\reg{Y})$.
	\item A collection of measurement operators $\{ Q_{a,b} : a \in \GammaA, \ b \in \GammaB \} \subset \Pos(\S)$, for alphabets $\GammaA$ and $\GammaB$.  
\end{itemize}
Viewing a QC game from the referee's perspective, it is played in the following manner:
\begin{enumerate}
	\item The referee prepares $(\reg{X},\reg{S},\reg{Y})$ in the state $\rho$, then sends $\reg{X}$ to Alice and $\reg{Y}$ to Bob. 
	\item Alice responds with $a \in \GammaA$ and Bob responds with $b \in \GammaB$. 
	\item The referee measures $\reg{S}$ with respect to the binary-valued measurement 
		\begin{align}
			\{Q_{a,b}, \I - Q_{a,b}\}.
		\end{align} 
	The outcome corresponding to the measurement operator $Q_{a,b}$ indicates that Alice and Bob \emph{win}, while the other measurement result indicates that they \emph{lose}. 
\end{enumerate}

\begin{figure}[!htpb] 
	\begin{center}
		\includegraphics[scale=1.0]{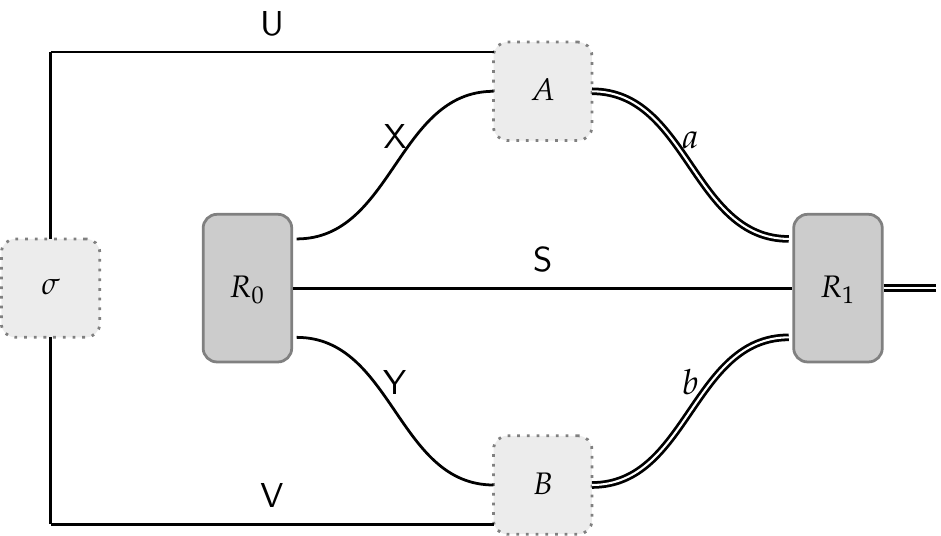}
	\end{center}
		\caption[A quantum strategy for a quantum-classical game.]{A quantum strategy for a quantum-classical game. Just as in a nonlocal game, if Alice and Bob are using a quantum strategy, they may prepare a state $\sigma \in \Density(\U \otimes \V)$ prior to the start of the game. Unlike a nonlocal game where the referee sends classical information, the referee in a quantum-classical game prepares a state $\rho \in \Density(\X \otimes \S \otimes \Y)$ in registers $(\reg{X},\reg{S},\reg{Y})$ and sends registers $\reg{X}$ and $\reg{Y}$ to Alice and Bob. Alice and Bob perform measurements to generate their answers $a \in \GammaA$ and $b \in \GammaB$, which are then sent back to the referee. The referee then evaluates whether or not Alice and Bob win or lose by making a measurement on its register $\reg{S}$.}
		\label{fig:quantum-classical-game}
\end{figure}

Just as there are various strategies that one may consider for the class of extended nonlocal games, one may also consider various classes of strategies for QC games. We will, however, restrict our attention to \index{quantum strategies (QC games)}{\emph{quantum strategies}} for a QC game. That is, a strategy that consists of a shared quantum state between Alice and Bob, as well as respective sets of measurement operators for Alice and Bob. This type of strategy is similar to a quantum strategy for a nonlocal game, but where now we take into account the fact that the questions that Alice and Bob receive in a QC game are provided via quantum registers. 

More precisely, a quantum strategy for a QC game specified by 
\begin{align}
	\rho \in \Density(\X \otimes \S \otimes \Y)	\quad \textnormal{and} \quad \{Q_{a,b} : a \in \GammaA \ b \in \GammaB \} \subset \Pos(\S)
\end{align}
as above, consists of the following objects:
\begin{enumerate}
	\item A state $\sigma \in \Density(\U \otimes \V)$, for $\U$ being the space corresponding to a register $\reg{U}$ held by Alice and $\V$ being the space corresponding to a register $\reg{V}$ held by Bob. 
	\item A measurement $\{A_a : a \in \GammaA\} \subset \Pos(\U \otimes \X)$ for Alice, performed on the pair $(\reg{U},\reg{X})$ after she receives $\reg{X}$ from the referee, and a measurement $\{ B_b : b \in \GammaB \} \subset \Pos(\Y \otimes \V)$ for Bob, performed on the pair $(\reg{Y},\reg{V})$ after he receives $\reg{Y}$ from the referee. 
\end{enumerate}
A quantum-classical game where Alice and Bob use a quantum strategy is depicted in Figure~\ref{fig:quantum-classical-game}. 

One may express the winning probability for a QC game when Alice and Bob adopt a quantum strategy as 
\begin{align}
	\sum_{(a,b) \in \GammaA \times \GammaB} \biggip{ A_a \otimes Q_{a,b} \otimes B_b }{W \left( \sigma \otimes \rho \right) W^*},
\end{align}
where $W$ is the unitary operator that corresponds to the natural re-ordering of registers consistent with each of the tensor product operators $A_a \otimes Q_{a,b} \otimes B_b$ (i.e. the permutation $( \reg{U}, \reg{V}, \reg{X}, \reg{S}, \reg{Y} ) \mapsto ( \reg{U}, \reg{X}, \reg{S}, \reg{Y}, \reg{V} )$).
The \index{quantum value (quantum-classical game)}{\emph{quantum value}} of a QC game represents the supremum of the winning probabilities, taken over all quantum strategies. If $G_{qc}$ is the name assigned to a QC game having a specification as above, then we write $\omega^*_N(G_{qc})$ to denote the \emph{maximum} winning probability taken over all quantum strategies for which $\dim(\U) = N = \dim(\V)$, so that the quantum value of $G_{qc}$ is 
\begin{align}
	\omega^*(G_{qc}) = \lim_{N \rightarrow \infty} \omega_N^*(G_{qc}).
\end{align}

Regev and Vidick~\cite{Regev2013} proved that certain QC games have the following peculiar property: if Alice and Bob make use of an entangled state of two finite-dimensional quantum systems, initially shared between them, they can never achieve perfect optimality---it is always possible for them to do better (meaning that they win with a strictly larger probability) using some different shared entangled state on two larger quantum systems. Thus, it is only in the limit, as the local dimensions of their shared entangled states goes to infinity, that they can approach an optimal performance in these specific examples of games. This was previously established for analogues of nonlocal games for which both the questions and answers are quantum~\cite{Leung2013}, and it is an open question to determine if the same property holds for any ordinary nonlocal game, where both the questions and answers must be classical.

In particular, Regev and Vidick considered a specific class of QC games called \index{quantum XOR games}{\emph{quantum XOR games}}, where the winning condition in such a game is predicated on an XOR function. Regev and Vidick showed that there exists a family of quantum XOR games such that if the dimension of Alice and Bob's quantum system, $N$, is finite, then the quantum value will be strictly less than $1$. However, taking the limit as $N$ goes to infinity, the quantum value approaches $1$. A restatement of their result follows. 
\begin{theorem}[Theorem 1.2 of~\cite{Regev2013}] \label{thm:regev-vidick-qcg}
	There exists a quantum-classical game $G_{qc}$ such that 
	\begin{align}
		\omega^*(G_{qc}) = 1,
	\end{align}
and for every positive integer $N$ it holds that 
	\begin{align}
		\omega_N^*(G_{qc}) < 1.
	\end{align}
\end{theorem}

\section[Constructing extended nonlocal games from quantum-classical games]{Constructing extended nonlocal games \\ from quantum-classical games} \label{sec:constructing-extended-nonlocal-games-from-quantum-classical-games}

In this section, we will state and prove an analogous theorem to Theorem~\ref{thm:regev-vidick-qcg} for an extended nonlocal game. That is, we will show that there exists an extended nonlocal game where the standard quantum value approaches $1$ when the dimension of the quantum systems shared by Alice and Bob approach infinity. 
\begin{theorem} \label{thm:enlg-from-qcg}
	Given a quantum-classical game, $G_{qc}$ with question registers $\reg{X}$ and $\reg{Y}$, there exists an extended nonlocal game, labelled as $H_t$, such that
	\begin{align} \label{eq:enlg-from-qcg}
		\omega_N^*(G_{qc}) \leq 1 - \abs{\reg{X}} \abs{\reg{Y}} \left( 1- \omega^*_{N \abs{\reg{X}} \abs{\reg{Y}}}(H_t) \right) \quad \textnormal{and} \quad \omega_N^*(H_t) \leq 1 - \frac{ 1 - \omega^*_{N \abs{\reg{X}} \abs{\reg{Y}}}(G_{qc}) }{\abs{\reg{X}} \abs{\reg{Y}}}.
	\end{align}
\end{theorem}
The main idea for proving Theorem~\ref{thm:enlg-from-qcg} will involve a successive reduction from a quantum-classical game to an intermediate type of game, called a \emph{teleportation game} (that we will define formally in the next section), and finally to an extended nonlocal game. Sections~\ref{sec:teleportation-games-and-quantum-classical-games} and~\ref{sec:extended-nonlocal-games-and-teleportation-games} are dedicated to proving Theorem~\ref{thm:enlg-from-qcg}. Specifically, in Section~\ref{sec:teleportation-games-and-quantum-classical-games}, we will show how quantum-classical games are related to teleportation games, and in Section~\ref{sec:extended-nonlocal-games-and-teleportation-games}, we will show how teleportation games are related to extended nonlocal games. Once these relationships are established, we will be able to prove Theorem~\ref{thm:enlg-from-qcg}.

\subsection{Teleportation games and quantum-classical games} \label{sec:teleportation-games-and-quantum-classical-games}

In this section we will introduce \index{teleportation game}{\emph{teleportation games}}. A teleportation game is similar to an extended nonlocal game in that the referee receives a state prepared by Alice and Bob, the referee sends questions to Alice and Bob, and the referee receives answers from them as well. The one key difference now is that after the referee receives the state from Alice and Bob, it will produce registers and perform a Bell measurement on the parts of the state sent by Alice and Bob along with the registers that it produced. The reason we refer to this class of games as teleportation games is because the registers that the referee produces are the registers that the referee desires to teleport to Alice and Bob. A teleportation game is depicted in Figure~\ref{fig:teleportation-game}.

Formally, a \emph{teleportation game} is specified by the following objects:
\begin{itemize}
	\item A state $\rho \in \Density(\X \otimes \S \otimes \Y)$ of a triple of registers $(\reg{X},\reg{S},\reg{Y})$. 
	\item A collection of measurement operators $\{ Q_{a,b} : a \in \GammaA, \ b \in \GammaB \} \subset \Pos(\S)$, where $\GammaA$ and $\GammaB$ are alphabets and $\S$ is the space corresponding to register $\reg{S}$. 
\end{itemize}
From the referee's perspective, such a game is played as follows:
\begin{enumerate}
	\item The referee is presented with the register $\reg{R} = (\reg{X}_1,\reg{Y}_1)$ where $\reg{X}_1$ and $\reg{Y}_1$, are copies of the registers $\reg{X}$ and $\reg{Y}$. (The register $\reg{R}$ might, for instance, be entangled with systems possessed by Alice and Bob.)
	\item The referee prepares $(\reg{X},\reg{S},\reg{Y})$ in the state $\rho$ and performs Bell measurements
		\begin{align} \label{eq:bell-basis-telep-game}
			\left \{ \phi_{x}^{(\abs{\reg{X}})} : x \in \SigmaA \right \} \subset \Pos(\X \otimes \X_1) \quad \textnormal{and} \quad \left \{ \phi_{y}^{(\abs{\reg{Y}})} : y \in \SigmaB \right \} \subset \Pos(\Y \otimes \Y_1)
		\end{align}		
	 and where 
	\begin{align}
	 	\SigmaA = \integer_{\abs{\reg{X}}} \times \integer_{\abs{\reg{X}}} \quad \textnormal{and} \quad \SigmaB = \integer_{\abs{\reg{Y}}} \times \integer_{\abs{\reg{Y}}},
	 \end{align}
	 on the respective pairs $(\reg{X},\reg{X}_1)$ and $(\reg{Y},\reg{Y}_1)$ yielding outcomes $x \in \SigmaA$ and $y \in \SigmaB$ which are sent to Alice and Bob.
	\item Alice and Bob respond with $a \in \GammaA$ and $b \in \GammaB$. 
	\item The referee measures $\reg{S}$ with respect to the binary-valued measurement 
		\begin{align}
			\{Q_{a,b}, \I - Q_{a,b}\} \subset \Pos(\S).		
		\end{align} 
	The outcome corresponding to the measurement operator $Q_{a,b}$ indicates that Alice and Bob \emph{win}, while the other measurement result indicates that they \emph{lose}.  
\end{enumerate}

Just as is the case for both extended nonlocal games and quantum-classical games, one may consider various types of strategies for Alice and Bob. For the purposes of this discussion, we will be focusing on \index{quantum strategy (teleportation game)}{\emph{quantum strategies}} in which Alice and Bob begin the game in possession of finite-dimensional quantum systems that have been initialized as they choose. They may then measure these systems in order to obtain answers to the referee's questions. 

In more precise terms, a quantum strategy for a teleportation game, specified by 
	\begin{align}
		\rho \in \Density(\X \otimes \S \otimes \Y) \quad \textnormal{and} \quad \{Q_{a,b} : a \in \GammaA, \ b \in \GammaB \} \subset \Pos(\S)
	\end{align}
as above, consists of these objects:
\begin{enumerate}
	\item A state $\sigma \in \Density(\U \otimes (\X_1 \otimes \Y_1) \otimes \V)$ where $(\X_1 \otimes \Y_1)$ is the space corresponding to registers $(\reg{X}_1, \reg{Y}_1)$ presented to the referee at the start of the game, where $\U$ is the space corresponding to register $\reg{U}$ held by Alice, and where $\V$ is the space corresponding to register $\reg{V}$ held by Bob.
	\item A measurement $\{A_a^x : a \in \GammaA\} \subset \Pos(\U)$ for each $x \in \SigmaA$, performed by Alice, when she receives the question $x$, and a measurement $\{B_b^y : b \in \GammaB\} \subset \Pos(\V)$ for each $y \in \SigmaB$, performed by Bob when he receives the question $y$. 
\end{enumerate}


If $G_t$ is the name assigned to a teleportation game having the specifications as above, then we write $\omega_N^*(G_t)$ to denote the \emph{maximum} winning probability taken over all quantum strategies for which $\dim(\U \otimes \V) = N$, so that the quantum value of $G_t$ is 
\begin{align}
	\omega^*(G_t) = \lim_{N \rightarrow \infty} \omega_N^*(G_t).
\end{align}

\begin{figure}[!htpb] 
	\begin{center}
		\includegraphics[scale=1.0]{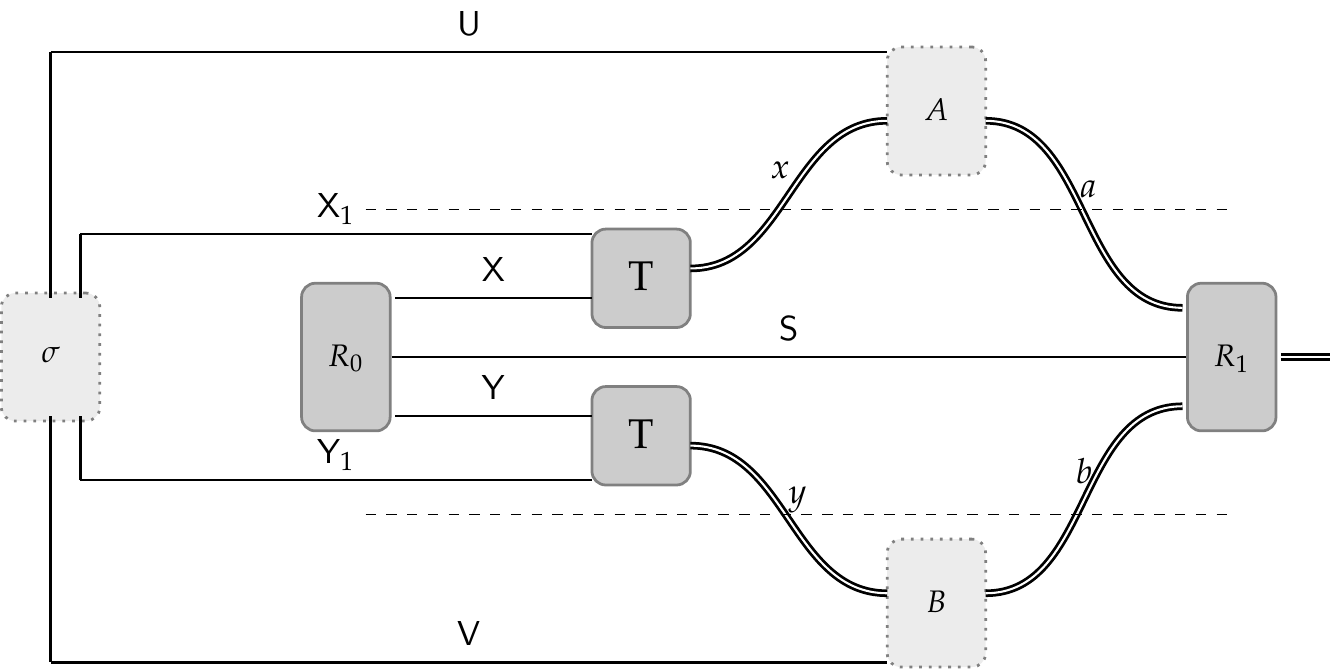}
	\end{center}
		\caption[A quantum strategy for a teleportation game.]{A quantum strategy for a teleportation game. Prior to the start of the game, the state $\sigma \in \Density(\U \otimes (\X_1 \otimes \Y_1) \otimes \V)$ is prepared. The referee obtains registers $(\reg{X}_1,\reg{Y}_1)$. The referee prepares a state $\rho \in \Density(\X \otimes \S \otimes \Y)$ contained in registers $(\reg{X},\reg{S},\reg{Y})$ and performs a generalized Bell measurement on registers $(\reg{X},\reg{X}_1)$ and $(\reg{Y},\reg{Y}_1)$. The outcomes of this measurement, $x \in \SigmaA$ is sent to Alice and $y \in \SigmaB$ is sent to Bob, who in turn respond with answers $a \in \GammaA$ and $b \in \GammaB$ to the referee. The referee then performs a measurement on the register $\reg{S}$ to determine whether Alice and Bob win or lose.}
		\label{fig:teleportation-game}
\end{figure}

\begin{lemma} \label{lem:qcg_eq_telep}
	Given any quantum-classical game, $G_{qc}$, with question registers $\reg{X}$ and $\reg{Y}$, there exists a teleportation game $G_t$ such that 
	\begin{align}
		\omega_N^*(G_{qc}) \leq \omega_{N\abs{\reg{X}}\abs{\reg{Y}}}^*(G_t) \quad \textnormal{and} \quad \omega^*_{N}(G_t) \leq \omega^*_{N\abs{\reg{X}}\abs{\reg{Y}}}(G_{qc}),
	\end{align}		
for all $N \geq 1$.
\end{lemma}

Prior to proceeding to the proof, we give a brief sketch to provide some intuition. In order to prove Lemma~\ref{lem:qcg_eq_telep}, we must prove both that $\omega_N^*(G_{qc}) \leq \omega_{N \abs{\reg{X}} \abs{\reg{Y}}}^*(G_t)$ and that $\omega_{N}^*(G_t) \leq \omega_{N\abs{\reg{X}}\abs{\reg{Y}}}^*(G_{qc})$. 

In the first inequality, we assume that Alice and Bob play honestly. That is to say, we assume that Alice and Bob play along and allow the referee to teleport registers to Alice and Bob. For this to happen, the initial state is prepared as a maximally entangled state and Alice and Bob also apply the appropriate Pauli teleportation corrections on their respective systems after they receive the questions from the referee. This direction of the proof is simply illustrating how such a strategy is carried out when Alice and Bob play honestly and is depicted in Figure~\ref{fig:teleportation-game-strategy}. 

In the second inequality, we remove the assumption that Alice and Bob play honestly. That is to say that we are not guaranteed that Alice and Bob prepare maximally entangled states, nor are we to assume that the registers they possess are not entangled in some arbitrarily complex manner. In other words, we are concerned now with the possibility that Alice and Bob may attempt to cheat, and play dishonestly. The general idea of this direction is that Alice and Bob will perform what may be thought of a teleportation protocol to themselves. That is, after Alice and Bob perform measurements in the Bell basis on their registers, they will use the outcome of these measurements to apply the appropriate generalized Pauli correction operator to their systems. 

\begin{proof}
	Let $G_t$ be the teleportation game that is defined in terms of the same state and measurement operators
	\begin{align}
		\rho \in \Density(\X \otimes \S \otimes \Y) \quad \textnormal{and} \quad \{ Q_{a,b} : a \in \GammaA \ b \in \GammaB \} \subset \Pos(\S)
	\end{align}	
that also define $G_{qc}$.  

	Let us first show that $\omega_N^*(G_{qc}) \leq \omega_{N\abs{\reg{X}}\abs{\reg{Y}}}^*(G_t)$. Consider an arbitrary strategy for any quantum-classical game $G_{qc}$. We show how one may adapt this strategy into a strategy for the teleportation game $G_t$. The following strategy is depicted in Figure~\ref{fig:teleportation-game-strategy}
	
	The state $\sigma$ is prepared in the following manner 
	\begin{align}
		\sigma \in \Density( (\U_1 \otimes \X_0) \otimes (\X_1 \otimes \Y_1) \otimes (\Y_0 \otimes \V_1))
	\end{align}		
in registers $(\reg{U}_1, \reg{X}_0, \reg{X}_1,\reg{Y}_1, \reg{Y}_0, \reg{V}_1)$ such that
	\begin{align}
		 \abs{\reg{X}_0} = \abs{\reg{X}} = \abs{\reg{X}_1} \quad \textnormal{and} \quad  \abs{\reg{Y}_0} = \abs{\reg{Y}} = \abs{\reg{Y}_1},	
	\end{align}		
where the contents of $(\reg{X}_0,\reg{X}_1)$ and $(\reg{Y}_0,\reg{Y}_1)$ are respective maximally entangled states 
\begin{align}
	\psi_{\reg{X}} = \frac{1}{\sqrt{\abs{\reg{X}}}} \sum_{c \in \integer_{\abs{\reg{X}}}} e_c \otimes e_c \quad \textnormal{and} \quad \psi_{\reg{Y}} = \frac{1}{\sqrt{\abs{\reg{Y}}}} \sum_{d \in \integer_{\abs{\reg{Y}}}} e_d \otimes e_d.
\end{align}
When the referee receives $\reg{X}_1$ and $\reg{Y}_1$, it prepares the quantum state $\rho \in \Density(\X \otimes \S \otimes \Y)$ contained in registers $(\reg{X},\reg{S},\reg{Y})$.

The referee then measures each pair $(\reg{X},\reg{X}_1)$ and $(\reg{Y},\reg{Y}_1)$ with respect to the Bell basis as from equation~\eqref{eq:bell-basis-telep-game} and obtains outcomes $x$ and $y$, where
\begin{align}
	x = (k_1,k_2) \in \SigmaA \quad \textnormal{and} \quad y = (l_1,l_2) \in \SigmaB,
\end{align}
which are then sent to Alice and Bob. 
Alice and Bob then apply one of the generalized Pauli operators 
\begin{align} \label{eq:gen-pauli-telep}
	\left \{ W_{k_1,k_2}^{(\abs{\reg{X}})} : (k_1, k_2) \in \SigmaA \right \} \quad \textnormal{and} \quad \left \{ W_{l_1, l_2}^{(\abs{\reg{Y}})} : (l_1, l_2) \in \SigmaB \right \},
\end{align}
to their registers $\reg{X}_0$ and $\reg{Y}_0$. This completes the teleportation protocol, and teleports the register $\reg{X}$ to Alice and $\reg{Y}$ to Bob. Finally, Alice and Bob respond with $a \in \GammaA$ and $b \in \GammaB$ by performing measurements from the sets 
\begin{align}
	\{ A_a^x : a \in \GammaA \} \subset \Pos(\U_1 \otimes \X_0) \quad \textnormal{and} \quad \{B_b^y : b \in \GammaB \} \subset \Pos(\V_1 \otimes \Y_0), 
\end{align}
for each $x \in \SigmaA$ and $y \in \SigmaB$. The referee then performs a measurement from the set 
\begin{align}
	\{Q_{a,b}, \I - Q_{a,b} \} \subset \Pos(\S). 
\end{align}
Since Alice and Bob receive registers $\reg{X}$ and $\reg{Y}$ by the teleportation protocol, it is clear that they win with at least the same probability as in $G_{qc}$. It follows that $\omega_N^*(G_{qc}) \leq \omega_{N\abs{\reg{X}}\abs{\reg{Y}}}^*(G_t)$. 

\begin{figure}[!htpb] 
	\begin{center}
		\includegraphics[scale=1.0]{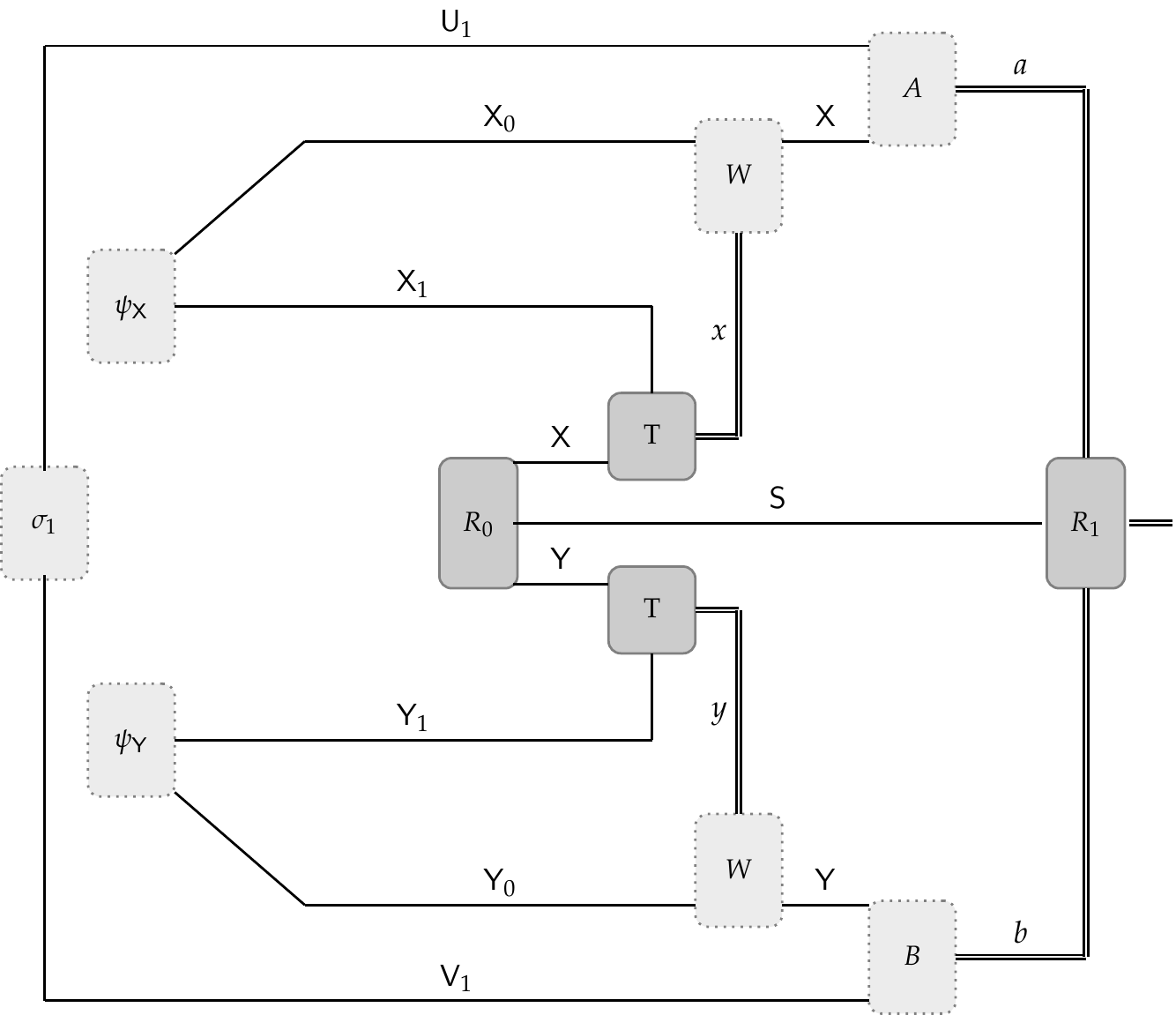}
	\end{center}
		\caption[A teleportation game strategy.]{The strategy that Alice and Bob abide by for a teleportation game when they play honestly. Alice and Bob prepare registers $( \reg{U}_1, \reg{X}_0, \reg{X}_1 )$ and $( \reg{V}_1, \reg{Y}_0, \reg{Y}_1 )$ where register pairs $(\reg{X}_0, \reg{X}_1)$ and $(\reg{Y}_0, \reg{Y}_1)$ consist of pairs of maximally entangled states. The referee receives registers $(\reg{X}_1,\reg{Y}_1)$ and prepares registers $(\reg{X},\reg{S},\reg{Y})$ and performs a measurement in the Bell basis on register pairs $(\reg{X},\reg{X}_1)$ and $(\reg{Y},\reg{Y}_1)$ in order to teleport $\reg{X}$ to Alice and $\reg{Y}$ to Bob. The outcome of these measurements result in $(x,y)$ where $x$ is sent to Alice and $y$ is sent to Bob. Alice and Bob then apply the appropriate Pauli corrections on their registers $\reg{X}_0$ and $\reg{Y}_0$, which teleports the registers $\reg{X}$ and $\reg{Y}$ into their possession. Finally, Alice and Bob respond with answers $a$ and $b$ to the referee, which is followed by the referee performing a measurement $\{Q_{a,b}, \I - Q_{a,b} \} \subset \Pos(\S)$.}
		\label{fig:teleportation-game-strategy}
\end{figure}

Now we show that $\omega_{N}^*(G_t) \leq \omega_{N\abs{\reg{X}}\abs{\reg{Y}}}^*(G_{qc})$. Consider an arbitrary strategy for the teleportation game $G_t$ from above. We show how one may adapt this strategy into a strategy for a quantum-classical game $G_{qc}$. 

Let the referee prepare a quantum state $\rho \in \Density(\X \otimes \S \otimes \Y)$ contained in registers $(\reg{X},\reg{S},\reg{Y})$, and let 
\begin{align}
	\sigma \in \Density( (\U \otimes \X_1) \otimes (\Y_1 \otimes \V) )
\end{align}
be the state shared between Alice, Bob, and the referee contained in registers $(\reg{U},\reg{X}_1,\reg{Y}_1,\reg{V})$. The registers $\reg{X}$ and $\reg{Y}$ are sent to Alice and Bob respectively. Once Alice and Bob receive $\reg{X}$ and $\reg{Y}$, they prepare a two step measurement:

\begin{enumerate}
	\item Alice and Bob measure $(\reg{X},\reg{X}_1)$ and $(\reg{Y},\reg{Y}_1)$ in the Bell basis yielding measurement outcomes 
		\begin{align}
			x \in \SigmaA \quad \textnormal{and} \quad y \in \SigmaB.
		\end{align}
	\item Alice and Bob perform measurements 
		\begin{align}
			\{ A_a^x : a \in \GammaA \} \subset \Pos(\U) \quad \textnormal{and} \quad \{ B_b^y : b \in \GammaB \} \subset \Pos(\V)
		\end{align}
		and obtain respective outcomes $a \in \GammaA$ and $b \in \GammaB$. 
\end{enumerate}
The two-step measurement operators corresponding to outcomes $a$ and $b$ are written as 
\begin{align}
	\sum_{x \in \SigmaA} A_a^x \otimes \phi_x^{\abs{\reg{X}}} \in \Pos(\U \otimes \X_1 \otimes \X) \quad \textnormal{and} \quad \sum_{y \in \SigmaB} B_b^y \otimes \phi_y^{\abs{\reg{Y}}} \in \Pos(\V \otimes \Y_1 \otimes \Y),
\end{align}
where $\phi_x^{\abs{\reg{X}}}$ and $\phi_y^{\abs{\reg{Y}}}$ are Bell measurements. 
Finally, the referee performs a measurement from the set 
\begin{align}
	\{ Q_{a,b}, \I - Q_{a,b} \} \subset \Pos(\S).
\end{align}
It is evident from the above procedure that the information stored in $x \in \SigmaA$ and $y \in \SigmaB$ is precisely what the referee would have sent in $G_t$. Furthermore, the cost of this procedure is given by the dimension of the state $\sigma$, that is
\begin{align}
	N \abs{\reg{X}} \abs{\reg{Y}} = \dim(\U \otimes (\X_1 \otimes \Y_1) \otimes \V).
\end{align}
It then follows that $\omega_{N\abs{\reg{X}}\abs{\reg{Y}}}^*(G_t) \leq \omega_N^*(G_{qc})$.
\end{proof}

\subsection{Extended nonlocal games and teleportation games} \label{sec:extended-nonlocal-games-and-teleportation-games}

In the previous section, we established a relationship between certain quantum-classical games and teleportation games. Building on this, we now show how teleportation games and certain extended nonlocal games are related. Once we have this chain of relationships, we will be able to prove Theorem~\ref{thm:enlg-from-qcg}. The following lemma establishes a relationship between teleportation games and extended nonlocal games. 

\begin{lemma} \label{lem:telep-to-enlg}
	Given any teleportation game, $G_t$, with teleported registers $\reg{X}$ and $\reg{Y}$, there exists an extended nonlocal game, $H_t$, such that 
	\begin{align}
		\omega_N^*(H_t) = 1 - \frac{ 1 - \omega_{N}^*(G_t)  }{ \abs{\reg{X}}^2 \abs{\reg{Y}}^2 },
	\end{align}
	for all $N$.
\end{lemma}

In order to prove Lemma~\ref{lem:telep-to-enlg}, as was done previously in the proof of Lemma~\ref{lem:qcg_eq_telep}, we assume that $G_t$ is defined by 
\begin{align}
	\rho \in \Density(\X \otimes \S \otimes \Y) \quad \textnormal{and} \quad \{Q_{a,b} : a \in \GammaA, \ b \in \GammaB \} \subset \Pos(\S).
\end{align} 
We shall also define a specific extended nonlocal game, $H_t$, that consists of a teleportation procedure. From the referee's perspective, such a game is played as follows:
\begin{enumerate}
	\item Alice and Bob present the referee with the register $\reg{R} = (\reg{X}_1, \reg{Y}_1)$ such that 
		\begin{align}
			\abs{\reg{X}_1} = \abs{\reg{X}} \quad \textnormal{and} \quad \abs{\reg{Y}_1} = \abs{\reg{Y}}.
		\end{align}
	Note that the register $\reg{R}$ may be entangled with systems possessed by Alice and Bob, as is the case for ordinary extended nonlocal games. 
	
	\item The referee randomly generates a pair $(x,y) \in \SigmaA \times \SigmaB$ where
	\begin{align}
		\SigmaA = \integer_{\abs{\reg{X}}} \times \integer_{\abs{\reg{X}}} \quad \textnormal{and} \quad \SigmaB = \integer_{\abs{\reg{Y}}} \times \integer_{\abs{\reg{Y}}}	
	\end{align}		
according to the uniform distribution and sends $x \in \SigmaA$ to Alice and $y \in \SigmaB$ to Bob. Alice responds with $a \in \GammaA$ and Bob responds with $b \in \GammaB$. 
	
	\item The referee prepares a state $\rho \in \Density(\X \otimes \S \otimes \Y)$ and then performs a measurement with respect to the binary-valued measurement $\{ P_{a,b,x,y}, \I - P_{a,b,x,y} \}$ where
\begin{equation} \label{eq:ref-meas-telep-enlg}
	\begin{aligned}
		P_{a,b,x,y} &= \I - \phi^{(\abs{\reg{X}})}_{x} \otimes \left( \I - Q_{a,b} \right) \otimes \phi^{(\abs{\reg{Y}})}_y, \\
		\I - P_{a,b,x,y} &= \phi^{(\abs{\reg{X}})}_{x} \otimes \left( \I - Q_{a,b} \right) \otimes \phi^{(\abs{\reg{Y}})}_y, \\
	\end{aligned}	
\end{equation}	 
where $\{Q_{a,b}, \I - Q_{a,b} \} \subset \Pos(\S)$. The outcome corresponding to the measurement $P_{a,b,x,y}$ indicates that Alice and Bob win, while the other measurement indicates that they lose. 
\end{enumerate}
As further intuition for the above protocol, we shall see that the last step may be thought of as a form of \index{post-selected teleportation}{\emph{post-selcted teleportation}} where the randomly selected questions $x$ and $y$ are compared to $x_1$ and $y_1$ which are hypothetical measurement results that would be obtained if the referee were to perform teleportation. Implicit in the winning and losing measurements is the relationship between $(x,y)$ and $(x_1,y_1)$ that
\begin{enumerate}
		\item \emph{If $x \not= x_1$ or $y \not= y_1$}: The referee immediately accepts, and therefore Alice and Bob win. 
		\item \emph{If $x = x_1$ and $y = y_1$}: The referee performs a measurement with respect to the binary-valued measurement $\{Q_{a,b}, \I - Q_{a,b} \}$ on register $\reg{S}$.
	\end{enumerate}
That is, in the event where $x \not= x_1$ or $y \not= y_1$, this corresponds to a failure to teleport $\reg{X}$ or $\reg{Y}$ to Alice or Bob. Likewise, the event where $x = x_1$ and $y = y_1$ corresponds to the event where teleportation protocol would have succeeded, since if the referee \emph{were} to teleport, it would have sent $x_1$ and $y_1$ to Alice and Bob, which would influence the measurement that they would apply to their system. Since in this case $x = x_1$ and $y = y_1$ it is \emph{as if} the referee were to teleport $\reg{X}$ to Alice and $\reg{Y}$ to Bob.

\begin{proof}[Proof of Lemma~\ref{lem:telep-to-enlg}]
Let $H_t$ be the extended nonlocal game as introduced above, and let it be defined in terms of the same state and measurement operators
	\begin{align}
		\rho \in \Density(\X \otimes \S \otimes \Y) \quad \textnormal{and} \quad \{Q_{a,b} : a \in \GammaA \ b \in \GammaB\} \subset \Pos(\S)
	\end{align}
that also define $G_{t}$ such that the measurement operators $\{P_{a,b,x,y}, \I - P_{a,b,x,y}\}$ in $H_t$ are defined in terms of $\{Q_{a,b}, \I - Q_{a,b} \}$ as in equation~\eqref{eq:ref-meas-telep-enlg}. 

Note that in both games $G_t$ and $H_t$, Alice and Bob's strategy is defined by the state
\begin{align}
	\sigma \in \Density(\U \otimes (\X_1 \otimes \Y_1) \otimes \V), \end{align}
as well as their respective measurement operators 
\begin{align}
	\{A_a^x : a \in \GammaA \} \subset \Pos(\U) \quad \textnormal{and} \quad  \{B_b^y : b \in \GammaB \} \subset \Pos(\V).
\end{align}
Let $p$ denote the winning probability for Alice and Bob in $G_t$
\begin{align}
	p = \sum_{\substack{(x,y) \in \SigmaA \times \SigmaB \\ (a,b) \in \GammaA \times \GammaB}} \biggip{ A_a^x \otimes \phi_x^{\abs{\reg{X}}} \otimes Q_{a,b} \otimes \phi_y^{\abs{\reg{Y}}} \otimes B_b^y }{ W (\rho \otimes \sigma) W^* },
\end{align}
where $W$ is the unitary operator that corresponds to the permutation of registers 
\begin{align}
	(\reg{X},\reg{S},\reg{Y},\reg{U},\reg{X}_1,\reg{Y}_1,\reg{V}) \mapsto (\reg{U},\reg{X},\reg{X}_1,\reg{S},\reg{Y},\reg{Y}_1,\reg{V}).
\end{align}
The losing probability for $G_t$ may be derived from the above equation by considering the losing measurement, that is
\begin{align} \label{eq:telep-losing-prob}
	1 - p = \sum_{\substack{(x,y) \in \SigmaA \times \SigmaB \\ (a,b) \in \GammaA \times \GammaB}} \biggip{ A_a^x \otimes \phi_x^{\abs{\reg{X}}} \otimes (\I - Q_{a,b}) \otimes \phi_y^{\abs{\reg{Y}}} \otimes B_b^y }{ W (\rho \otimes \sigma) W^* }.
\end{align}
We will show how the losing probability of $H_t$ may be written in terms of the losing probability of $G_t$. 

Consider an arbitrary strategy for any teleportation game $G_t$. We show how one may adapt this strategy into a strategy for the extended nonlocal game $H_t$. 

Let $(\reg{X}_0, \reg{X}_1)$ and $(\reg{Y}_0,\reg{Y}_1)$ be quantum registers such that
\begin{align}
	 \abs{\reg{X}_0} = \abs{\reg{X}} = \abs{\reg{X}_1} \quad \textnormal{and} \quad  \abs{\reg{Y}_0} = \abs{\reg{Y}} = \abs{\reg{Y}_1},
\end{align}
where the contents of $(\reg{X}_0,\reg{X}_1)$ and $(\reg{Y}_0,\reg{Y}_1)$ are respective maximally entangled states 
\begin{align}
	\psi_{\reg{X}} = \frac{1}{\sqrt{\abs{\reg{X}}}} \sum_{c \in \integer_{\abs{\reg{X}}}} e_c \otimes e_c \quad \textnormal{and} \quad \psi_{\reg{Y}} = \frac{1}{\sqrt{\abs{\reg{Y}}}} \sum_{d \in \integer_{\abs{\reg{Y}}}} e_d \otimes e_d. 
\end{align}
Registers $\reg{X}_1$ and $\reg{Y}_1$ are sent to the referee. 

The referee then chooses $(x,y) \in \SigmaA \times \SigmaB$ at random, according to the uniform probability distribution. The referee makes a local copy of $x$ and $y$ as usual, and then sends $x$ to Alice and $y$ to Bob. Alice and Bob then perform measurements from the sets 
\begin{align}
	\{ A_a^x : a \in \GammaA \} \subset \Pos(\U) \quad \textnormal{and} \quad \{B_b^y : b \in \GammaB \} \subset \Pos(\V),
\end{align}
for each $x \in \SigmaA$ and $y \in \SigmaB$ yielding outcomes $a \in \GammaA$ and $b \in \GammaB$, which are then sent to the referee. 

\begin{figure}[!htpb] 
	\begin{center}
		\includegraphics[scale=1.0]{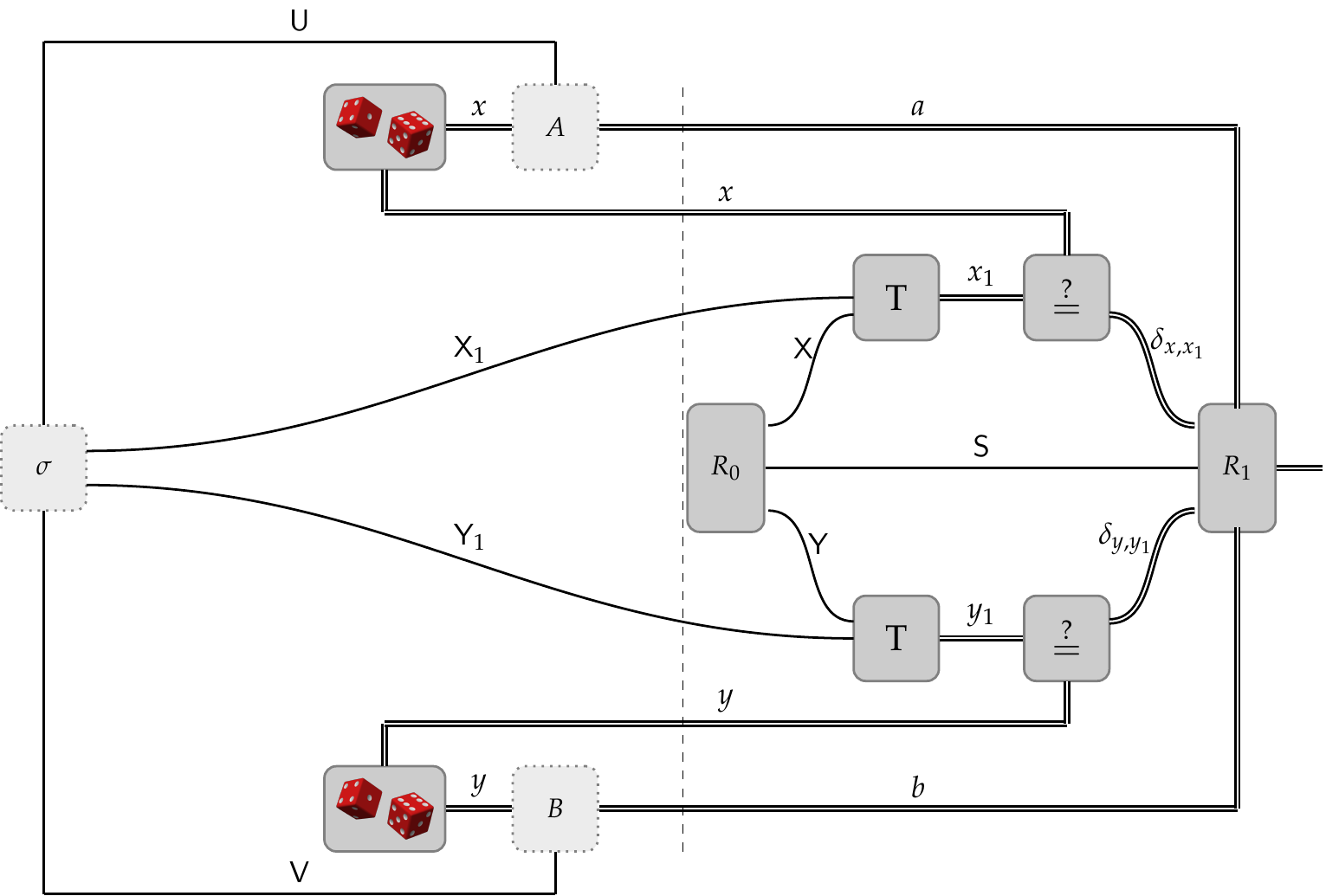}
	\end{center}
		\caption[The extended nonlocal game $H_t$.]{The extended nonlocal game $H_t$, an extended nonlocal game where the referee initiates a teleportation procedure. That is, once the referee sends questions $(x,y) \in \SigmaA \times \SigmaB$ to Alice and Bob and receives registers $(\reg{X}_1,\reg{Y}_1)$, it shall create a state $\rho \in \Density(\X \otimes \S \otimes \Y)$ in registers $(\reg{X},\reg{S},\reg{Y})$ and perform teleportation using $(\reg{X},\reg{X}_1)$ and $(\reg{Y},\reg{Y}_1)$. The outcome of these teleportation procedures will yield $x_1$ and $y_1$. If both $x_1 = x$ and $y_1 = y$, the referee will perform a measurement on his system, $\reg{S}$, to determine the outcome of the game. Otherwise if $x_1 \not= x$ or $y_1 \not= y$, Alice and Bob automatically win. The dark gray shapes correspond to actions performed by the referee.}
		\label{fig:G2}
\end{figure}

Once the referee receives $a \in \GammaA$ and $b \in \GammaB$, it prepares a state $\rho \in \Density(\X \otimes \S \otimes \Y)$ in registers $(\reg{X},\reg{S},\reg{Y})$ such that 
\begin{align}
	\abs{\reg{X}} = \abs{\reg{X}_1} \quad \textnormal{and} \quad \abs{\reg{Y}} = \abs{\reg{Y}_1}.
\end{align}
The referee now measures $(\reg{X},\reg{X}_1)$ and $(\reg{Y},\reg{Y}_1)$ in the Bell basis, which yields respective outcomes of 
\begin{align}
	x_1 = (k_1,k_2) \in \SigmaA \quad \textnormal{and} \quad y_1 = (l_1,l_2) \in \SigmaB.
\end{align}
We now consider the post-selected teleportation protocol to be a success if $x = x_1$ and $y = y_1$. For each $x_1$ there is a $1/\abs{\reg{X}}^2$ chance that the referee obtains a matching outcome in $x$. Similarly, for each $y_1$, there is a $1/\abs{\reg{Y}}^2$ chance that the referee obtains a matching outcome for $y$. Therefore, the total probability with which the post-selected teleportation protocol is performed successfully is with probability $1/\abs{\reg{X}}^2\abs{\reg{Y}}^2$. 

Depending on the outcome of the referee's measurement, the game proceeds accordingly:
\begin{enumerate}
	\item If $x \not= x_1$ or $y \not= y_1$, this implies that at least of of the two teleportation protocols has failed. In this case, the referee accepts, and Alice and Bob win the game. 
	\item If $x = x_1$ and $y = y_1$, this implies that both teleportation protocols are successful.
\end{enumerate}
The referee proceeds to perform the measurement $\{P_{a,b,x,y}, \I - P_{a,b,x,y} \}$ defined as in equation~\eqref{eq:ref-meas-telep-enlg}. Let $q$ denote the winning probability of $H_t$
\begin{align}
	q = \frac{1}{\abs{\reg{X}}^2\abs{\reg{Y}}^2} \sum_{\substack{(x,y) \in \SigmaA \times \SigmaB \\ (a,b) \in \GammaA \times \GammaB}} \biggip{A_a^x \otimes P_{x,y,a,b} \otimes B_b^y}{W (\rho \otimes \sigma) W^*},
\end{align}
where again $W$ is the unitary operator that corresponds to the permutation 
\begin{align}
	(\reg{X},\reg{S},\reg{Y},\reg{U},\reg{X}_1,\reg{Y}_1,\reg{V}) \mapsto (\reg{U},\reg{X},\reg{X}_1,\reg{S},\reg{Y},\reg{Y}_1,\reg{V}).
\end{align}
We may write the losing probability of $H_t$ as 
\begin{equation}
	\begin{aligned}
		1-q &= \frac{1}{\abs{\reg{X}}^2\abs{\reg{Y}}^2} \sum_{\substack{(x,y) \in \SigmaA \times \SigmaB \\ (a,b) \in \GammaA \times \GammaB}} \biggip{A_a^x \otimes (\I - P_{x,y,a,b}) \otimes B_b^y}{W (\rho \otimes \sigma) W^*} \\
		&= \frac{1}{\abs{\reg{X}}^2\abs{\reg{Y}}^2} \sum_{\substack{(x,y) \in \SigmaA \times \SigmaB \\ (a,b) \in \GammaA \times \GammaB}} \biggip{A_a^x \otimes \phi_x^{\abs{\reg{X}}} \otimes (\I - Q_{a,b}) \otimes \phi_y^{\abs{\reg{Y}}} \otimes B_b^y}{W (\rho \otimes \sigma) W^*}  \\
		&= \frac{1}{\abs{\reg{X}}^2\abs{\reg{Y}}^2} (1-p).
	\end{aligned}
\end{equation}
Since in both cases we have that $N = \dim(\U \otimes \V)$, optimizing over strategies of cost $N$ gives 
\begin{align}
	\omega_N^*(H_t) = 1 - \frac{ 1 - \omega_N^*(G_t) }{\abs{\reg{X}}^2 \abs{\reg{Y}}^2},
\end{align}
which concludes the proof.

\end{proof}

\subsubsection*{Proof of Theorem~\ref{thm:enlg-from-qcg}}

\begin{proof}[Proof of Theorem~\ref{thm:enlg-from-qcg}]
	Recall from Lemma~\ref{lem:qcg_eq_telep} we have that 
	\begin{align}
		\omega_N^*(G_{qc}) \leq \omega_{N\abs{\reg{X}}\abs{\reg{Y}}}^*(G_t) \quad \textnormal{and} \quad \omega^*_{N}(G_t) \leq \omega^*_{N\abs{\reg{X}}\abs{\reg{Y}}}(G_{qc}),
	\end{align}		
	for all $N \geq 1$. From Lemma~\ref{lem:telep-to-enlg} it follows that 
	\begin{align}
		1 - \frac{1 - \omega_{N}^*(G_t)}{\abs{\reg{X}}^2\abs{\reg{Y}}^2} = \omega_N^*(H_t).
	\end{align}
	It then follows that the inequalities from equation~\eqref{eq:enlg-from-qcg} hold.
	Furthermore, it follows from Theorem~\ref{thm:regev-vidick-qcg} that there exists a quantum-classical game $G_{qc}$ where $\omega_N^*(G_{qc}) = 1$ is achieved in the limit as $N$ goes to infinity. It then follows that there also exists an extended nonlocal game $H_t$, where $\omega_N^*(H_t) = 1$ as $N$ approaches infinity. 
\end{proof}

\subsubsection*{Implications and discussion of Theorem~\ref{thm:enlg-from-qcg}}

We briefly discuss the broader context of Theorem~\ref{thm:enlg-from-qcg}. As previously mentioned, Regev and Vidick~\cite{Regev2013} proved that a certain class of QC games have the property that if Alice and Bob make use of an entangled state initially shared between them, they can never achieve perfect optimality, it is always possible for them to do better (meaning that they win with a strictly larger probability) using some different shared entangled state on larger quantum systems. We found in the previous section that there also exists a class of extended nonlocal games with this property as well. 
 
We can also ask whether or not the above property holds more generally for some class of nonlocal games. Formally, 
\begin{question}
	Does there exist a nonlocal game $G$ such that 
	\begin{align}
		\omega^*(G) = 1,
	\end{align}
	and for every positive integer $N$ it holds that 
	\begin{align}
		\omega_N^*(G) < 1.
	\end{align}
\end{question}
For the traditional nonlocal game case with classical questions and classical answers, this question remains open. The so-called \index{I3322 inequality}{\emph{I3322 inequality}}, when formulated as a nonlocal game, has been conjectured to have the property just described, in which increasing degrees of entanglement admit strategies with strictly increasing success rates~\cite{Pal2009}.

It is also worth noting that Theorem~\ref{thm:enlg-from-qcg} implies the existence of a tripartite steering inequality for which an infinite-dimensional quantum state is required in order to achieve a maximal violation. This follows from the fact that extended nonlocal games may be equivalently viewed as a tripartite steering scenario (as considered in~\cite{Cavalcanti2015} and~\cite{Sainz2015}), as was mentioned in Chapter~\ref{chap:extended_nonlocal_games}. 

\section{Variations on the extended nonlocal game model} \label{sec:variations-on-enlg}

Recall that a nonlocal game consists of two rounds of communication: one from the referee to the players and one from the players to the referee. The standard definition of a nonlocal game assumes that both rounds of communication consist of classical messages. We saw a variation on that model in Section~\ref{sec:quantum-classical-games}, in which the question round was replaced with the referee sending quantum questions to both Alice and Bob. In a similar manner, we may also consider such variations on the extended nonlocal game model. The standard extended nonlocal game consists of three rounds of communication, that is
\begin{enumerate}
	\item (Quantum): Alice and Bob prepare a state $\sigma \in \Density(\U \otimes \R \otimes \V)$ shared between themselves and the referee. 
	\item (Classical): The referee randomly generates classical questions for Alice and Bob, $(x,y) \in \SigmaA \times \SigmaB$, respectively.
	\item (Classical): Alice and Bob respond with answers $a \in \GammaA$ and $b \in \GammaB$. 
\end{enumerate}
In complete generality, any variation on the type of communication used in an extended nonlocal game may be specified in terms of a tuple $(t_1,t_2,t_3) \in \{q,c\}$ where each round of communication consists of either a transmission of classical or quantum information as denoted by either $c$ or $q$ in the tuple. For instance, the type of communication in each round of a standard extended nonlocal game corresponds to the tuple $(q,c,c)$. We therefore equivalently may refer to the standard definition of an extended nonlocal game as a quantum-classical-classical extended nonlocal game or just a QCC extended nonlocal game for short. 


\subsection{Quantum-classical-quantum extended nonlocal games}

Consider the class of \index{quantum-classical-quantum extended nonlocal games}{\emph{quantum-classical-quantum extended nonlocal games}} or QCQ extended nonlocal games for short. This class of game is defined precisely as a standard extended nonlocal game, only now the last round of communication is replaced with Alice and Bob sending quantum registers in place of the classical message $a \in \GammaA$ and $b \in \GammaB$ to the referee. 

Specifically, a QCQ extended nonlocal game is specified by the following objects:
\begin{itemize}
	\item A probability distribution $\pi : \SigmaA \times \SigmaB \rightarrow [0,1]$, for alphabets $\SigmaA$ and $\SigmaB$. 
	\item A collection of measurement operators $\{P_{x,y} : x \in \SigmaA, \ y \in \SigmaB \} \subset \Pos(\A \otimes \R \otimes \B)$ where $\A,\B$, and $\R$ are complex Euclidean spaces corresponding to registers $\reg{A},\reg{B},$ and $\reg{R}$. 
\end{itemize}
From the referee's perspective, such a game is played as follows: 
\begin{enumerate}
	\item Alice and Bob present the referee with the register $\reg{R}$, which has been initialized in a state of Alice and Bob's choosing. (The register $\reg{R}$ might, for instance, be entangled with systems possessed by Alice and Bob.)
	\item The referee randomly generates a pair $(x,y) \in \SigmaA \times \SigmaB$ according to the distribution $\pi$, and sends $x$ to Alice and $y$ to Bob. Alice and Bob then send registers $\reg{A}$ and $\reg{B}$ corresponding to spaces $\A$ and $\B$ to the referee. 
	\item The referee measures registers $(\reg{A},\reg{R},\reg{B})$ with respect to the binary-valued measurement $\{P_{x,y}, \I - P_{x,y}\} \subset \Pos(\A \otimes \R \otimes \B)$. The outcome corresponding to the measurement operator $P_{x,y}$ indicates that Alice and Bob win, while the other measurement result indicates that they lose. 
\end{enumerate}
For any QCQ extended nonlocal game, there are various classes of strategies that may be adapted from the standard extended nonlocal game case for Alice and Bob, including unentangled strategies, standard quantum strategies, commuting measurement strategies, and non-signaling strategies. In this section, we only consider standard quantum strategies for QCQ extended nonlocal games.

A standard quantum strategy for a QCQ extended nonlocal game, specified by 
\begin{align}
	\pi : \SigmaA \times \SigmaB \rightarrow [0,1] \quad \textnormal{and} \quad \{P_{x,y} : x \in \SigmaA, \ y \in \SigmaB \} \subset \Pos(\A \otimes \R \otimes \B)
\end{align}
as above, consists of these objects: 

\begin{enumerate}
	\item A state $\sigma \in \Density(\U \otimes \R \otimes \V)$, for $\U$ being the space corresponding to a register $\reg{U}$ held by Alice and $\V$ being the space corresponding to a register $\reg{V}$ held by Bob. This state represents Alice and Bob's initialization of the tripe $(\reg{U},\reg{R},\reg{V})$ immediately before $\reg{R}$ is sent to the referee. 
	\item A collection of channels $\{ \Phi^x \} \subset \Channel(\U,\A)$ for each $x \in \SigmaA$, applied by Alice when she receives the question $x$, and a collection of channels $\{\Phi^y \} \subset \Channel(\V,\B)$ for each $y \in \SigmaB$, applied by Bob when he receives the question $y$. Alice and Bob then send their portions of the state after they have applied their channels to the referee. 
\end{enumerate} 
When Alice and Bob utilize such a strategy, their winning probability may be expressed as 
\begin{align}
	  \sum_{\substack{
      (x,y) \in \SigmaA \times \SigmaB 
  }}
	\biggip{P_{x,y}}{ \left( \Phi^x \otimes \I_{\Lin(\R)} \otimes \Phi^y \right) \left(\sigma\right) }.
\end{align}
\begin{figure}[!htpb] 
	\begin{center}
			\includegraphics[scale=1.0]{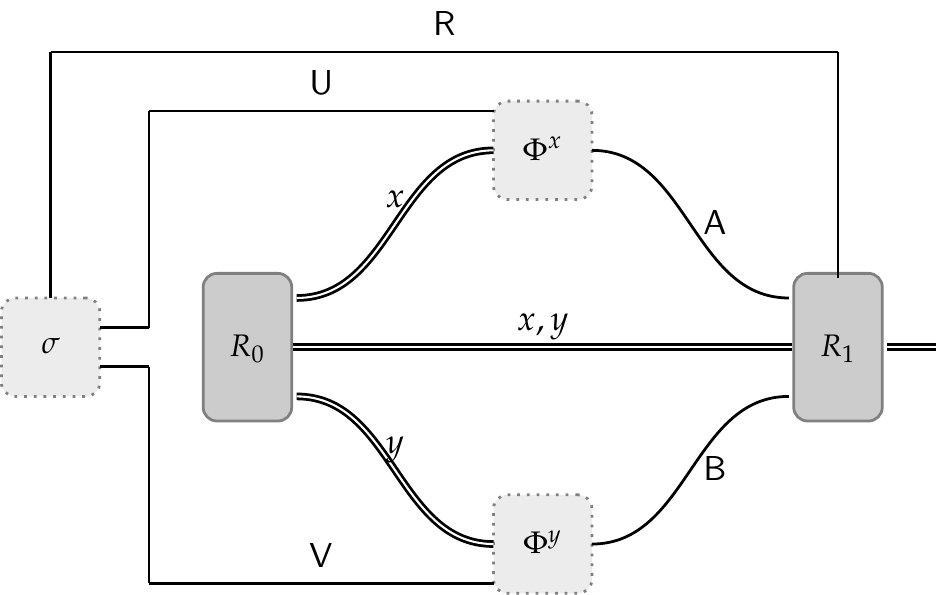}
	\end{center}
		\caption[A quantum-classical-quantum extended nonlocal game.]{A quantum-classical-quantum extended nonlocal game. The referee selects questions $(x,y) \in \SigmaA \times \SigmaB$ according to the probability distribution $\pi$, and sends $x$ to Alice and $y$ to Bob. Upon receiving $(x,y)$, Alice and Bob apply channels $\Phi^x$ and $\Phi^y$ to their respective systems and respond with answers in the form of quantum registers $(\reg{A},\reg{B})$ over complex Euclidean spaces $\A$ and $\B$. After receiving $(\reg{A},\reg{B})$ from Alice and Bob, the referee performs a measurement on $\{P_{x,y}, \I - P_{x,y}\}\subset \Pos(\A \otimes \R \otimes \B)$ to determine the probability with which Alice and Bob win or lose.}
		\label{fig:quantum-classical-quantum-enlg}
\end{figure}
Using a similar teleportation trick that we have used in Section~\ref{sec:extended-nonlocal-games-and-teleportation-games}, one may show that an arbitrary strategy for a QCC extended nonlocal game may be adapted for a QCQ extended nonlocal game. There is not much to be gained from going through the explicit details of this, as they are nearly identical to the process we have seen in the section already.

It is, however, relevant to note that in Lemma 32 of~\cite{Kobayashi2015}, the authors use a similar teleportation trick to prove a relationship between two different subclasses of complexity classes arising from what they refer to as the ``generalized-QAM'' complexity class. These two subclasses are similar in some sense to the QCC extended nonlocal game model and the QCQ extended nonlocal game model, as the authors also analyze variants of the QAM complexity class with similar properties.  

\chapter{Bounding the standard quantum value of extended nonlocal games}
\label{chap:extended_npa_hierarchy}

In this chapter, we shall present a number of heuristic methods that obtain bounds on the standard quantum value of an extended nonlocal game. For placing upper bounds, we take inspiration from the QC hierarchy~\cite{Doherty2008, Navascues2007,Navascues2008}; a hierarchy of semidefinite programs that yield progressively better upper bounds on the quantum value for nonlocal games as one computes higher levels of the hierarchy. Indeed, we adopt these results and provide what we refer to as the \index{extended QC hierarchy}{\emph{extended QC hierarchy}} and apply this technique to the class of extended nonlocal games to obtain upper bounds on the standard quantum value. 

In Section~\ref{sec:extended-npa-hierarchy}, we present the extended QC hierarchy in greater detail. We begin in Section~\ref{sec:intuitive-description-of-the-extended-qc-hierarchy} by giving an informal description of how the extended QC hierarchy is structured. In Section~\ref{sec:construction-of-the-extended-qc-hierarchy}, we make this description more formal and show in Section~\ref{sec:convergence-of-the-extended-qc-hierarchy} that the extended QC hierarchy has a similar convergence property as does the original QC hierarchy. In Section~\ref{sec:examples-upper-bounds-extended-npa}, we shall provide some explicit examples of how one may apply the extended QC hierarchy to extended nonlocal games. 

In Section~\ref{sec:lower-bound-extended-nonlocal-games}, we present our method to lower bound the value of extended nonlocal games which is inspired by the work of Liang and Doherty~\cite{Liang2007}, where they consider a method that can be applied to lower bound the quantum value in the nonlocal game setting. We adopt their technique and apply it to the case of extended nonlocal games. In Section~\ref{sec:examples-lower-bounds}, we provide explicit examples of applying this lower bound technique to an extended nonlocal game. 

This chapter is based on joint work with Nathaniel Johnston, Rajat Mittal, and John Watrous~\cite{Johnston2015a}.  

\minitoc

\section{Upper bounds for extended nonlocal games: the extended QC hierarchy}
\label{sec:extended-npa-hierarchy}

In this section we describe how the original \index{QC hierarchy}{\emph{QC hierarchy}}~\cite{Doherty2008,Navascues2007, Navascues2008}, may be generalized to extended nonlocal games. The QC hierarchy is a method that allows one to obtain upper bounds on the quantum value of a nonlocal game. Specifically, for a finite level, the commuting measurement value of a nonlocal game is guaranteed to be obtained, which serves as a natural upper bound to the quantum value of a nonlocal game. Directly calculating the quantum value of a nonlocal game is probably intractable, but in most cases, the first few levels of the QC hierarchy are numerically tractable to compute on current hardware, and in many cases the first few levels are sufficient~\cite{Pal2009}. 

\subsection{Intuitive description of the extended QC hierarchy} \label{sec:intuitive-description-of-the-extended-qc-hierarchy}

In this section, we shall provide some intuition on how one may interpret and use the extended QC hierarchy. Many of these ideas are also found in the QC hierarchy for nonlocal games. First, let us establish what the extended QC hierarchy is used for: it is a technique to allow one to place upper bounds on the standard quantum value of a given extended nonlocal game. More precisely, it is a method that allows us to obtain the commuting measurement value of a given extended nonlocal game, where it may be recalled from Chapter~\ref{chap:extended_nonlocal_games}, that the commuting measurement value is an upper bound on the standard quantum value for every extended nonlocal game, that is $\omega^*(G) \leq \omega_c(G)$ holds for any extended nonlocal game $G$. 

Recall that the commuting measurement value of an extended nonlocal game is given by a maximization over the following equation
\begin{align} \label{eq:commuting-meas-val-eqchier}
	\sum_{(x,y) \in \Sigma} \pi(x,y) \sum_{(a,b) \in \Gamma} \biggip{V(a,b|x,y)}{K(a,b|x,y)},
\end{align}
where $K$ is a commuting measurement assemblage operator. What the extended QC hierarchy allows us to do is to consider the following equation instead to compute the commuting measurement value
\begin{align} \label{eq:eqchier-val}
	\sum_{(x,y) \in \Sigma} \pi(x,y) \sum_{(a,b) \in \Gamma} \biggip{V(a,b|x,y)}{M^{(k)}(a,b|x,y)},
\end{align}
where now $M^{(k)}$ is some matrix parametrized by some integer $k$ with entries indexed by $a \in \GammaA$, $b \in \GammaB$, $x \in \SigmaA$, and $y \in \SigmaB$ satisfying certain constraints, which we will elaborate on shortly. The benefit of this is that we can optimize over the matrix $M^{(k)}$ and these constraints for some level $k$ by way of a semidefinite program, and thereby compute the commuting measurement value of an extended nonlocal game. Of course, showing that such a correspondence exists between equations~\eqref{eq:commuting-meas-val-eqchier} and~\eqref{eq:eqchier-val} is the difficult part, and is what we will be showing in Theorem~\ref{thm:extended-npa-convergence} in Section~\ref{sec:convergence-of-the-extended-qc-hierarchy}. 

Assuming that there does exist such a correspondence though, let us consider what the $M^{(k)}$ matrices look like, and how to compute the constraints on these matrices as the level $k$ increases. These matrices can be thought to embody certain properties that one would expect from a commuting measurement strategy. That is to say that the entries in the matrices $M^{(k)}$ are indexed by strings which correspond to operators coming from a commuting measurement strategy. It may be recalled that the measurements for such a strategy obey pair-wise commutative properties, sum to the identity when summing over the outputs, and may be considered to be projective without any loss of generality. The strings within these matrices possess these qualities in ways that we will describe. 

Assume that question and answer alphabets $\SigmaA$, $\SigmaB$, $\GammaA$, and
$\GammaB$, as well as a positive integer $m$ representing the dimension of the
referee's quantum system, have been fixed. The symbol $\cupdot$ denotes the disjoint union, meaning that $\SigmaA \times \GammaA$ and
$\SigmaB \times \GammaB$ are to be treated as disjoint sets when forming
\begin{align}
\Delta = \left( \SigmaA \times \GammaA \right) \cupdot \left( \SigmaB \times \GammaB \right).
\end{align}
We write $\Delta^{\ast}$ to denote the set of all strings (of finite length) over
$\Delta$, and we write $\varepsilon$ to denote the empty string.

For simplicity, we will restrict our attention to $k = 1$, the first level of the extended QC hierarchy, and consider an extended nonlocal game where the dimension of the referee's space is $r$ and the game consists of $n$ possible questions and $m$ possible answers for each player. The matrix $M^{(1)}$ consists $r \times r$ blocks
\begin{align}
	M^{(1)} = \begin{pmatrix}
				M_{1,1}^{(1)} & \cdots & M^{(1)}_{1,r} \\
				\vdots & \ddots & \vdots \\
				M_{m,1}^{(1)} & \cdots & M_{r,r}^{(1)}
			   \end{pmatrix}.
\end{align}
We construct each block $M_{i,j}^{(1)}$ by lining all tuples of strings that correspond to measurement operators and the identity operator used in the extended nonlocal game. For instance, the tuple $(x,a)$ can be thought of a string pair that corresponds to Alice's measurement operator $A_a^x$. We use $\epsilon$ as the empty string that relates to the identity operator. More specifically, for $k = 1$, we consider all strings of length at most one from the set
\begin{align}
	\Delta^{\leq 1} = \{\varepsilon\} \cup \left \{ (x,a) \right \} \cup \left \{ (y,b) \right \}.
\end{align}
The block matrices are then formed from this set as $M_{i,j}^{(1)} : \Delta^{\leq 1} \times \Delta^{\leq 1}$ where $1 \leq i,j \leq r$.  This is a bit clearer if we simply write out what we have described thus far 
\begin{footnotesize}
\[
\begin{aligned}
M_{i,j}^{(1)} =
\left(
\begin{array}{c||cccc|ccc}
 & \epsilon & (x_1,a_1) & \cdots & (x_{nm},a_{nm}) & (y_1,b_1) & \cdots & (y_{nm},b_{nm}) \\
\hline\hline 
\epsilon & & & & & & & \\
(x_1,a_1) & & & & & & & \\
\vdots & & & & & & & \\
(x_{nm},a_{nm}) & & & & & & & \\
\hline 
(y_1,b_1) & & & & & & & \\
\vdots & & & & & & & \\
(y_{nm},b_{nm}) & & & & & & &
\end{array}
\right). 
\end{aligned}
\]
\end{footnotesize}
Now we have not actually placed an entry into this matrix yet. The outer row and outer column are simply guides we will use to fill in the matrix. Specifically, we fill the matrix by composing the outer row element with the outer column element. Again, writing this out explicitly may enhance the explanation,
\begin{tiny}
\[
\begin{aligned}
M_{i,j}^{(1)} =
\left(
\begin{array}{c||cccc|ccc}
 & \epsilon & (x_1,a_1) & \cdots & (x_{nm},a_{nm}) & (y_1,b_1) & \cdots & (y_{nm},b_{nm}) \\
\hline\hline 
\epsilon & \epsilon & (x_1,a_1) & \cdots & (x_{nm},a_{nm}) & (y_1,b_1) & \cdots & (y_{nm},b_{nm}) \\
(x_1,a_1) & (x_1,a_1) & (x_1,a_1) & \cdots & (x_1,a_1) (x_{nm},a_{nm}) & (x_1,a_1) (y_1,b_1) & \cdots & (x_1,a_1) (y_{nm},b_{nm}) \\
\vdots & \vdots & \vdots & \ddots & \vdots & \vdots  & \ddots & \vdots \\
(x_{nm},a_{nm}) & (x_{nm},a_{nm}) & (x_{nm},a_{nm}) (x_1,a_1) & \cdots & (x_{nm},a_{nm}) & (x_{nm},a_{nm}) (y_1,b_1) & \cdots & (x_{nm},a_{nm}) (y_{nm},b_{nm}) \\
\hline 
(y_1,b_1) & (y_1,b_1) & (y_1,b_1) (x_1,a_1) & \cdots & (y_1,b_1) (x_{nm},a_{nm}) & (y_1,b_1) & \cdots & (y_1,b_1) (y_{nm},b_{nm}) \\
\vdots & \vdots & \vdots & \ddots & \vdots & \vdots  & \ddots & \vdots \\
(y_{nm},b_{nm}) & (y_{nm},b_{nm}) & (y_{nm},b_{nm}) (x_1,a_1) & \cdots & (y_{nm},b_{nm}) (x_{nm},a_{nm}) & (y_{nm},b_{nm}) (y_1,b_1) & \cdots & (y_{nm},b_{nm})
\end{array}
\right). 
\end{aligned}
\]
\end{tiny}
Consider the last entry in the second row with entry $(x_1,a_1)(y_{nm},b_{nm})$. To obtain this entry, we multiplied, or more precisely, concatenated string pairs $(x_1,a_1)$, coming from the outer column, with the pair $(y_{nm},b_{nm})$, coming from the outer row. 

It is also essential to consider the first row, first column, and diagonal of $M_{i,j}^{(1)}$. Note how there is just a single element in these spots instead of two concatenated tuples. The reasoning behind this is that, as mentioned before, the properties of this matrix are meant to embody those of commuting measurement strategy. Take for instance the first entry in $M_{i,j}^{(1)}$ which is equal to just $\epsilon$. We would expect an entry of $\epsilon \epsilon$, but recall $\epsilon$ corresponds to the identity operator and $\I \I = \I$. In a similar way, the second entry in the second row is just $(x_1,a_1)$. Again, we would expect an entry of $(x_1,a_1)(x_1,a_1)$, however since we may assume that strings representing the measurements are projective, that is $A_{a_1}^{x_1} A_{a_1}^{x_1} = A_{a_1}^{x_1}$, this property is conveyed in a similar way. The same idea applies to the entire diagonal of $M_{i,j}^{(1)}$. 

We now consider how the commutation relationships are conveyed in this matrix. In $M_{i,j}^{(1)}$, this property is represented as enforcing that 
\begin{align}
	M^{(1)}_{i,j}((x,a),(y,b)) = M^{(1)}_{i,j}((y,b),(x,a)),
\end{align}
for all $(i,j)$ blocks. For instance, consider the entries $(y_1,b_1)(x_1,a_1)$ and $(x_1,a_1)(y_1,b_1)$. These entries are equal since the strings represent operators coming from a commuting measurement strategy, that is to say they represent the property 
\begin{align}
	\left[A_{a_1}^{x_1}, B_{b_1}^{y_1} \right] = \left[ B_{b_1}^{y_1}, A_{a_1}^{x_1} \right] = 0.
\end{align}
We also need to convey the property that the measurements of Alice and Bob are equal to the identity when summing over all answers. This is conveyed by observing that
\begin{equation}
	\begin{aligned}
		\sum_{a \in \GammaA} M^{(1)}((x,a),(y,b)) &= M^{(1)}(\epsilon,(y,b)), \\
		\sum_{b \in \GammaB} M^{(1)}((x,a),(y,b)) &= M^{(1)}((x,a),\epsilon).
	\end{aligned}
\end{equation}
The last constraint we place on the matrix $M^{(1)}$ is that it must be positive semidefinite. If this constraint, along with all of the other constraints regarding the blocks of $M^{(1)}$ are enforced, we refer to $M^{(1)}$ as a first-order admissible matrix, which will be formally defined in the coming sections for any $k$. Optimizing over such a matrix subject to the above conditions while attempting to maximize equation~\eqref{eq:eqchier-val} will provide us with our desired upper bound on the standard quantum value for some extended nonlocal game. 

It may happen that the first level of the hierarchy is not sufficient in attaining the true commuting measurement value. Specifically, the value at the first level may be higher than the actual commuting measurement value. In this case, computing higher levels of $k$ will help us in getting closer to the true commuting measurement value. When constructing the block matrices $M^{(k)}_{i,j}$ for some level $k$, each block will have the following form
\begin{align}
	M_{i,j}^{(k)} : \Delta^{\leq k} \times \Delta^{\leq k} \rightarrow \complex,
\end{align} 
where 
\begin{equation}
	\begin{aligned}
		\Delta^{\leq 0} &= \{ \varepsilon \}, \\
		\Delta^{\leq 1} &= \Delta^{\leq 0} \cup \{ (x,a) \} \cup \{ (y,b) \}, \\
		\Delta^{\leq 2} &= \Delta^{\leq 1} \cup \{ (x,a)(x^{\prime},a^{\prime}) \} \cup \{ (x,a)(y,b) \} \cup \{ (y,b)(y^{\prime},b^{\prime}) \}, \\
		\vdots
	\end{aligned}
\end{equation}
where $x \not= x^{\prime}$, $a \not= a^{\prime}$, $y \not= y^{\prime}$, and $b \not= b^{\prime}$. It is apparent that as $k$ increases, the alphabets have the following property that 
\begin{align}
	\Delta^{\leq 0} \subseteq \Delta^{\leq 1} \subseteq \cdots \subseteq \Delta^{\leq k}.
\end{align}
That is to say, the blocks in $M^{(k)}$ will consist of more and more entries as $k$ increases. One of the main ideas of the extended QC hierarchy that is also similar to the original QC hierarchy is that as $k$ increases, this leads to better and better approximations of the commuting measurement value of some extended nonlocal game $G$, that is
\begin{align}
	\omega_c^k(G) \leq \cdots \leq \omega_c^{2}(G) \leq \omega_c^1(G),
\end{align}
for some value of $k$. This relationship is depicted in Figure~\ref{fig:qc-hierarchy-levels}. 

\begin{figure}[!htpb] 
	\begin{center}
		\includegraphics[scale=1.5]{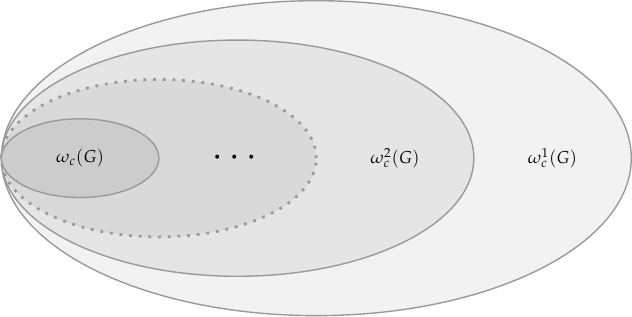}
	\end{center}
		\caption[Levels of the extended QC hierarchy.]{A visual representation of computing levels of the extended QC hierarchy. The outermost ellipse corresponds to the value attained when one computes the first level of the extended QC hierarchy. We represent this as $\omega^1_c(G)$, where $k=1$ represents the level computed. For certain games, this may indeed by equal to the true commuting measurement value of the game, that is $\omega_c(G)$.}
\label{fig:qc-hierarchy-levels}		
\end{figure}  

As previously mentioned, one nice property of the extended QC hierarchy that is also enjoyed by the QC hierarchy for nonlocal games is that obtaining the commuting measurement value is guaranteed for any extended nonlocal game for some finite level $k$. The downside to this, of course, is that $k$ may be particularly large. From an algorithm analysis perspective, the original QC hierarchy scales exponentially with respect to the level computed, that is $(nm)^k$, where again $n$ represents the total number of questions and $m$ represents the total number of answers. Indeed, for the extended QC hierarchy, the complexity fares even worse as we also have the referee's space to be concerned about now. It is therefore sometimes helpful to consider intermediate levels that are between integer values of $k$. For instance
\begin{equation}
	\begin{aligned}
		\Delta^{\leq 0} &= \{ \varepsilon \}, \\
		\Delta^{\leq 1} &= \Delta^{\leq 0} \cup \{ (x,a) \} \cup \{ (y,b) \}, \\
		\Delta^{\leq 1+AB} &= \Delta^{\leq 1} \cup \{ (x,a)(y,b) \}, \\
		\vdots
	\end{aligned}
\end{equation}
where $x \not= x^{\prime}$, $a \not= a^{\prime}$, $b \not= b^{\prime}$, and $y \not= y^{\prime}$. The $k = 1 + AB$ level is not quite as restrictive computationally as the complete second level of the hierarchy, and this may be enough in certain cases to obtain the commuting measurement value for a given extended nonlocal game. These intermediate levels have also been considered for the original QC hierarchy as well and serve a similar purpose. 

In practice, for many nonlocal games and extended nonlocal games, the values emerging from low levels of the QC hierarchy and extended QC hierarchy agree with the true commuting measurement value of the game. While there do exist some exceptions to this property, such as the I3322 game (based on the I3322 inequality~\cite{Collins2004}), the authors here~\cite{Pal2009} for instance were able to show a wide variety of Bell inequalities (or nonlocal games) that the QC hierarchy was able to obtain the commuting measurement value at low levels of the hierarchy.
 
In the next section, we will make our intuition developed in this section more formal, and we will further prove that the extended QC hierarchy allows one to obtain the commuting measurement value of any extended nonlocal game. Our proof technique for this follows very closely the technique used in~\cite{Navascues2007, Navascues2008} to prove convergence for the QC hierarchy for nonlocal games.

\subsection{Construction of the extended QC hierarchy} \label{sec:construction-of-the-extended-qc-hierarchy}

Define $\sim$ to be the equivalence relation on $\Delta^{\ast}$
generated by the following rules:\vspace{2mm}

\noindent
\begin{tabular}{@{\hspace*{1.2mm}}lll}
  1. \hspace*{-3mm} & $s \sigma t \sim s \sigma \sigma t$ & (for every
  $s,t\in\Delta^{\ast}$ and $\sigma \in \Delta$).\\[1mm]
  2. \hspace*{-3mm} & $s \sigma \tau t \sim s \tau \sigma t$ & 
  (for every $s,t\in\Delta^{\ast}$, $\sigma \in \SigmaA \times \GammaA$, and 
  $\tau\in \SigmaB \times \GammaB$).\\[2mm]
\end{tabular}

\noindent
That is, two strings are equivalent with respect to the relation $\sim$
if and only if one can be obtained from the other by a finite number of
applications of the above rules.

Now, a function of the form 
\begin{equation}
  \phi: \Delta^{\ast} \rightarrow \complex
\end{equation}
will be said to be \index{admissible}{\emph{admissible}} if and only if the following conditions are
satisfied:
\begin{mylist}{\parindent}
\item[1.] For every choice of strings $s,t\in\Delta^{\ast}$ it holds that
  \begin{equation} \label{eq:enpa-sum-strings}
    \sum_{a\in \GammaA} \phi(s (x,a) t) = \phi(st)
    \quad\text{and}\quad
    \sum_{b\in \GammaB} \phi(s (y,b) t) = \phi(st)
  \end{equation}
  for every $x\in \SigmaA$ and $y\in \SigmaB$.
\item[2.]
  For every choice of strings $s,t\in \Delta^{\ast}$, it holds that
  \begin{equation}
    \phi(s (x,a) (x,a') t) = 0 \quad\text{and}\quad
    \phi(s (y,b) (y,b') t) = 0
  \end{equation}
  for every choice of $x\in \SigmaA$ and $a,a'\in \GammaA$ satisfying $a \not= a'$, 
  and every choice of $y\in \SigmaB$ and $b,b'\in \GammaB$ satisfying $b \not= b'$, 
  respectively.
\item[3.]
  For all strings $s,t\in\Delta^{\ast}$ satisfying $s\sim t$ it
  holds that $\phi(s) = \phi(t)$.
\end{mylist}
Along similar lines, a function of the form
\begin{equation}
  \phi: \Delta^{\leq k} \rightarrow \complex
\end{equation}
is said to be \emph{admissible} if and only if the same conditions listed above
hold, provided that $s$ and $t$ are sufficiently short so that $\phi$ is defined
on the arguments indicated within each condition.

Finally, for each positive integer $k$ (representing a level of approximation in
the hierarchy to be constructed), we consider the set of all block matrices of
the form
\begin{equation}
  \label{eq:block-matrix}
  M^{(k)}
  = \begin{pmatrix}
    M^{(k)}_{1,1} & \cdots & M^{(k)}_{1,m}\\
    \vdots & \ddots & \vdots\\
    M^{(k)}_{m,1} & \cdots & M^{(k)}_{m,m}
  \end{pmatrix},
\end{equation}
where each of the blocks takes the form
\begin{equation}
  M^{(k)}_{i,j} : \Delta^{\leq k} \times \Delta^{\leq k} \rightarrow \complex,
\end{equation}
and for which the following conditions are satisfied:
\begin{mylist}{\parindent}
\item[1.]
  For every choice of $i,j\in\{1,\ldots,m\}$, there exists an admissible
  function
  \begin{equation}
    \phi_{i,j}:\Delta^{\leq 2k}\rightarrow \complex
  \end{equation}
  such that
  \begin{equation} \label{eq:enpa-mat-varphi}
    M^{(k)}_{i,j}(s,t) = \phi_{i,j}(s^{\mathsmaller{R}}t)
  \end{equation}
  for every choice of strings $s,t\in\Delta^{\leq k}$.
  (Here, the notation $s^{\mathsmaller{R}}$ means the \emph{reverse} of the
  string $s$.)
\item[2.]
  It holds that
  \begin{equation}
    \label{eq:pseudo-correlation-constraint-3}
    M^{(k)}_{1,1}(\varepsilon,\varepsilon) + \cdots +
    M^{(k)}_{m,m}(\varepsilon,\varepsilon) = 1.
  \end{equation}
\item[3.]
  The matrix $M^{(k)}$ is positive semidefinite.
\end{mylist}
Matrices of the form \eqref{eq:block-matrix} obeying the listed constraints will
be called \index{$k$-th order admissible matrix}{\emph{$k$-th order admissible matrices}}.
For such a matrix, we write $M^{(k)}(s,t)$ to denote the $m\times m$ complex
matrix
\begin{equation}
  M^{(k)}(s,t)
  = \begin{pmatrix}
    M^{(k)}_{1,1}(s,t) & \cdots & M^{(k)}_{1,m}(s,t)\\
    \vdots & \ddots & \vdots\\
    M^{(k)}_{m,1}(s,t) & \cdots & M^{(k)}_{m,m}(s,t)
  \end{pmatrix},
\end{equation}
for each choice of strings $s,t\in\Delta^{\leq k}$.
With respect to this notation, the second and third conditions on $M^{(k)}$
imply that $M^{(k)}(\varepsilon,\varepsilon)$ is an $m\times m$ density matrix.

We observe that an optimization over all $k$-th order admissible
matrices can be represented by a semidefinite program: a matrix of the form
\eqref{eq:block-matrix} is a $k$-th order admissible matrix if and only if it
is positive semidefinite and satisfies a finite number of linear constraints
imposed by the first two conditions on $M^{(k)}$.
In particular, for an extended nonlocal game $G = (\pi,V)$,
where $\pi$ is a distribution over $\SigmaA \times \SigmaB$ and $V$ is a function
$V: \GammaA \times \GammaB \times \SigmaA \times \SigmaB \rightarrow\Pos(\complex^m)$, 
one may consider the maximization of the quantity
\begin{equation}
  \sum_{(x,y) \in \SigmaA \times \SigmaB} \pi(x,y) \sum_{(a,b) \in \GammaA \times \GammaB} \Bigip{V(a,b|x,y)}{M^{(k)}((x,a),(y,b))}
\end{equation}
subject to $M^{(k)}$ being a $k$-th order admissible matrix.

We also note that the original QC hierarchy
corresponds precisely to the $m=1$ case of the hierarchy just described.

\subsection{Convergence of the extended QC hierarchy} \label{sec:convergence-of-the-extended-qc-hierarchy}

In this section, we show that the extended QC hierarchy converges to the set of commuting measurement assemblages. Our convergence proof follows a very similar trajectory to the convergence proof of the original QC hierarchy outlined in~\cite{Navascues2008}. The primary idea here, and in the original proof, is that for some finite $k$, there exists a $k$-th order admissible matrix that represents a commuting measurement assemblage. The beneficial property of this, as we have previously stated, is that the properties of this matrix are amenable to optimization via a semidefinite program. This is appealing from a computational standpoint, as we can leverage this property along with convex optimization software (such as CVX~\cite{Grant2008a}) to compute upper bounds on the standard quantum value of extended nonlocal games. 

We now give some intuition for how the proof of convergence proceeds. The easier direction of the proof is to show that if you are given a commuting measurement assemblage, then this assemblage satisfies the properties specified by a $k$-th order pseudo commuting measurement assemblage for every level $k$. The harder and more interesting direction is the converse. The main idea of proving this direction is very similar to the idea presented in~\cite{Navascues2008}, where one needs to show that there must exist collections of measurement operators $\{A_a^x\} \subset \Pos(\H)$ and $\{B_b^y\} \subset \Pos(\H)$ for some Hilbert space $\H$ belonging to Alice and Bob, as well as a state $\rho \in \Density(\R \otimes \H)$ that satisfy the conditions of a commuting measurement strategy arising from a $k$-th order pseudo commuting measurement assemblage. The basic idea here is to consider the $k$-th order pseudo commuting measurement assemblage as a matrix and show how the shared state and collections of measurements arise from the definition of this matrix. 

Now, for a fixed choice of alphabets $\SigmaA$, $\SigmaB$, $\GammaA$, and $\GammaB$, as well as positive
integers $m$ and $k$, let us consider the set of all functions of the form
\begin{equation}
  K : \GammaA \times \GammaB \times \SigmaA \times \SigmaB \rightarrow \Lin(\complex^m)
\end{equation}
for which there exists a $k$-th order admissible matrix $M^{(k)}$ that satisfies
\begin{equation}
  K(a,b|x,y) = M^{(k)}((x,a),(y,b))
\end{equation}
for every $x\in \SigmaA$, $y\in \SigmaB$, $a\in \GammaA$, and $b\in \GammaB$.

The set of all such functions will be called \index{$k$-th order pseudo commuting measurement assemblage}{\emph{$k$-th order pseudo commuting measurement assemblages}}. 

\begin{theorem} \label{thm:extended-npa-convergence}
  Let $\SigmaA$, $\SigmaB$, $\GammaA$, and $\GammaB$ be alphabets, let $m$ be a positive integer, let $\R = \complex^m$ be a complex Euclidean space, and
  let
  \begin{equation}
    K : \GammaA \times \GammaB \times \SigmaA \times \SigmaB \rightarrow \Lin(\R)
  \end{equation}
  be a function.
  The following statements are equivalent:
  \begin{mylist}{\parindent}
  \item[1.]
    The function $K$ is a commuting measurement assemblage.
  \item[2.]
    The function $K$ is a $k$-th order pseudo commuting measurement assemblage for every
    positive integer $k$.
  \end{mylist}
\end{theorem}
We require the following lemma to prove Theorem~\ref{thm:extended-npa-convergence}. This lemma will allow us to claim that the entries of a $k$-th order admissible matrix, $M^{(k)}$, are bounded above by one. 

\begin{lemma} \label{lem:enpa-bounded-entries}
	Let $m,k \geq 1$ be positive integers. Then a $k$-th order admissible matrix, $M^{(k)}$, satisfies 
	\begin{align}
		\abs{M_{i,j}^{(k)}(s,t)} \leq 1,
	\end{align}
for every $i,j \in \{1, \ldots, m\}$ and all $s,t \in \Delta^{\leq k}$. 
\end{lemma}

\begin{proof}
	It follows that since $M^{(k)}$ is positive semidefinite, then the $2 \times 2$ principal submatrix of $M^{(k)}$, written as 
  \begin{equation}
    \begin{pmatrix}
      M_{i,i}^{(k)}(s,s) & M_{i,j}^{(k)}(s,t) \\[2mm]
      M_{j,i}^{(k)}(t,s) & M_{j,j}^{(k)}(t,t)
    \end{pmatrix}
  \end{equation}
must also be positive semidefinite for each $i,j \in \{1,\ldots,m\}$ and $s,t \in \Delta^{\ast}$. It follows then that 
\begin{align}
	\abs{M_{i,j}^{(k)}(s,t)} \leq \sqrt{M_{i,i}^{(k)}(s,s)} \sqrt{M_{j,j}^{(k)}(t,t)}
\end{align}
for each $i,j \in \{1, \ldots, m\}$ and $s,t \in \Delta^{\ast}$. It now remains to show that
	\begin{align} \label{eq:extended-npa-psd-bound}
		M_{i,i}^{(k)}(s,s) \leq 1	
	\end{align}	 
	for every $i \in \{1,\ldots,m\}$ and $s \in \Delta^{\leq k}$. We prove equation~\eqref{eq:extended-npa-psd-bound} by induction on the length of $s$. For the base case, it holds that $M_{i,i}^{(k)}(\varepsilon,\varepsilon) \leq 1$ by the property of equation~\eqref{eq:pseudo-correlation-constraint-3} that 
	\begin{align}
		\sum_{i=1}^m M_{i,i}^{(k)}(\varepsilon,\varepsilon) = 1,
	\end{align}
and that the diagonal entries of $M^{(k)}$ are nonnegative. For the general case, for any string $t \in \Delta^{\ast}$ and any choice of $(z,c) \in \Delta$, it holds that 
	\begin{align}
		M_{i,i}^{(k)}((z,c)t,(z,c)t) &\leq \sum_d M_{i,i}^{(k)}((z,d)t,(z,d)t) \\
		&= \sum_d \phi_{i,i}^{(k)}(t^\mathsmaller{R}(z,d)(z,d)t) \label{eq:enpa-psd-2} \\
		&= \sum_d \phi_{i,i}^{(k)}(t^\mathsmaller{R}(z,d)t) \label{eq:enpa-psd-3} \\
		&= \phi_{i,i}^{(k)}(t^\mathsmaller{R}t) \label{eq:enpa-psd-4} \\
		&= M_{i,i}^{(k)}(t,t) \label{eq:enpa-psd-5},
	\end{align}	 
	where $d \in \GammaA$ if $z \in \SigmaA$ or $d \in \GammaB$ if $z \in \SigmaB$ and where equation~\eqref{eq:enpa-psd-2} follows from equation~\eqref{eq:enpa-mat-varphi}, equation~\eqref{eq:enpa-psd-3} follows from the equivalence relation on strings that $s \sigma t \sim s \sigma \sigma t$ for every $s,t \in \Delta^{\ast}$ and $\sigma \in \Delta$, equation~\eqref{eq:enpa-psd-4} follows from equation~\eqref{eq:enpa-sum-strings}, and equation~\eqref{eq:enpa-psd-5} follows again from equation~\eqref{eq:enpa-mat-varphi}. The proof follows by the hypothesis of induction. 
\end{proof}

With Lemma~\ref{lem:enpa-bounded-entries} in hand, we proceed to the proof of Theorem~\ref{thm:extended-npa-convergence}.

\begin{proof}[Proof of Theorem~\ref{thm:extended-npa-convergence}]
  The simpler implication is that statement 1 implies statement 2.
  Under the assumption that statement 1 holds, it must be that $K$ is defined
  by a strategy in which Alice and Bob use projective measurements,
  \begin{align}
  	\{ A_a^x : a \in \GammaA \} \subset \Pos(\H) \quad \textnormal{and} \quad \{ B_b^y : b \in \GammaB \} \subset \Pos(\H)
  \end{align}
  for Alice and Bob, on a shared (possibly infinite-dimensional) complex Euclidean space $\H$, along with a pure state $u\in\R\otimes\H$.
  Let $u_1,\ldots,u_m \in \H$ be vectors for which
  \begin{equation}
    u = \sum_{j = 1}^m e_j \otimes u_j.
  \end{equation}
  Also let $\Pi^z_c$ denote $A^z_c$ if $z\in \SigmaA$ and $c\in \GammaA$, or $B^z_c$ if
  $z\in \SigmaB$ and $c\in \GammaB$.
  With respect to this notation, one may consider the $k$-th order
  admissible matrix $M^{(k)}$ defined by
  \begin{equation}
    M^{(k)}_{i,j}(s,t) = \phi_{i,j}(s^{\mathsmaller{R}}t),
  \end{equation}
  where the functions $\{\phi_{i,j}\}$ are defined as
  \begin{equation}
    \phi_{i,j} \bigl((z_1, c_1) \cdots (z_\ell, c_\ell)\bigr) 
    = u_i^* \Pi_{c_1}^{z_1} \cdots \Pi_{c_\ell}^{z_\ell} u_j
  \end{equation}
  for every string $(z_1, c_1) \cdots (z_\ell, c_\ell)\in\Delta^{\leq 2k}$.
  A verification reveals that this matrix is consistent with $K$, and therefore
  $K$ is a $k$-th order pseudo commuting measurement
    assemblage.

  The more difficult implication is that statement 2 implies statement 1.
  The basic methodology of the proof is similar to the $m=1$ case proved in
  \cite{Navascues2008}, and we will refer to arguments made in that paper
  when they extend to the general case.
  For every positive integer $k$, let $M^{(k)}$ be a $k$-th order admissible
  matrix satisfying $K(a,b|x,y) = M^{(k)}((x,a),(y,b))$ for every $x\in \SigmaA$,
  $y\in \SigmaB$, $a\in \GammaA$, and $b\in \GammaB$.
	First, by Lemma~\ref{lem:enpa-bounded-entries} it follows that for every choice of $k \geq 1$ that 
  \begin{equation}
    \Bigabs{M_{i,j}^{(k)}(s,t)} \leq 1
  \end{equation}
  for every choice of $i,j\in\{1,\ldots,m\}$ and $s,t\in\Delta^{\leq k}$.


  Next, reasoning in the same way as~\cite{Navascues2008}, we may assume that there exists an infinite matrix $\hat{M}^{(k)}$ created from $M^{(k)}$ by padding blocks $M_{i,j}^{(k)}$ to make them infinite. This sequence of infinite matrices $\{\hat{M}^{(k)} : k = 1,2,\ldots \}$ admits a subsequence $\{k_l\}$ that weak-* converges to a limit when $l$ approaches infinity. Recall this fact follows from the Banach--Alaoglu theorem mentioned in Chapter~\ref{chap:preliminaries}. This implies that 
  \begin{align}
  	\lim_{l \rightarrow \infty} \hat{M}^{(k_l)} \rightarrow M,
  \end{align}
where $M$ is an infinite matrix of the form
  \begin{equation}
    M
    = \begin{pmatrix}
      M_{1,1} & \cdots & M_{1,m}\\
      \vdots & \ddots & \vdots\\
      M_{m,1} & \cdots & M_{m,m}
    \end{pmatrix},
  \end{equation}
  where
  \begin{equation}
    M_{i,j} : \Delta^{\ast} \times \Delta^{\ast} \rightarrow \complex
  \end{equation}
  for each $i,j\in\{1,\ldots,m\}$, satisfying similar constraints to the finite
  matrices $M^{(k)}$.
  In particular, it must hold that 
  \begin{equation}
    M_{i,j}(s,t) = \phi_{i,j}(s^{\mathsmaller{R}}t)
  \end{equation}
  for a collection of admissible functions $\{\phi_{i,j}\}$ taking the form
  \begin{equation}
    \phi_{i,j}:\Delta^{\ast} \rightarrow \complex,
  \end{equation}
  it must hold that all finite submatrices of $M$ are positive semidefinite,
  and it must hold that
  $M_{1,1}(\varepsilon,\varepsilon) + \cdots +
  M_{m,m}(\varepsilon,\varepsilon) = 1$.
  Consequently, there must exist a collection of vectors
  \begin{equation}
    \label{eq:vectors-in-H}
    \bigl\{u_{i,s}\,:\,i\in\{1,\ldots,m\},\;s\in\Delta^{\ast}\}\subset\H
  \end{equation}
  chosen from a separable Hilbert space $\H$ for which it holds that
  \begin{equation}
    M_{i,j}(s,t) = \bigip{u_{i,s}}{u_{j,t}}
  \end{equation}
  for every choice of $i,j\in\{1,\ldots,m\}$ and $s,t\in\Delta^{\ast}$.
  Furthermore, it must hold that
  \begin{equation}
    K(a,b|x,y) = M((x,a),(y,b))
  \end{equation}
  where, as for the matrices $M^{(k)}$, we write
  \begin{equation}
    M(s,t) = 
    \begin{pmatrix}
      M_{1,1}(s,t) & \cdots & M_{1,m}(s,t)\\
      \vdots & \ddots & \vdots\\
      M_{m,1}(s,t) & \cdots & M_{m,m}(s,t)
    \end{pmatrix}
  \end{equation}
  for each $s,t\in\Delta^{\ast}$.
  There is no loss of generality in assuming $\H$ is spanned by
  the vectors \eqref{eq:vectors-in-H}, for otherwise $\H$ can simply be replaced
  by the (possibly finite-dimensional) subspace spanned by these vectors.

  Now we will define a commuting measurement strategy for Alice and Bob
  certifying that $K$ is a commuting measurement assemblage.
  The state initially prepared by Alice and Bob, and shared with the referee,
  will be the pure state corresponding to the vector
  \begin{equation}
    u = \sum_{j = 1}^m e_j \otimes u_{j,\varepsilon} \in \R \otimes\H.
  \end{equation}
  This is a unit vector, as a calculation reveals:
  \begin{equation}
    \norm{u}^2 = \sum_{j = 1}^m 
    \ip{u_{j,\varepsilon}}{u_{j,\varepsilon}}
    = M_{1,1}(\varepsilon,\varepsilon) + \cdots +
    M_{m,m}(\varepsilon,\varepsilon) = 1.
  \end{equation}

  Next, we define projective measurements on $\H$ for Alice and Bob.
  For each $(z,c)\in\Delta$, define $\Pi^z_c$ to be the projection operator
  onto the span of the set
  \begin{equation} \label{eq:enpa-projection-onto-span}
    \bigl\{u_{j,(z,c)s}\,:\,j\in\{1,\ldots,m\},\;s\in\Delta^{\ast}\bigr\}.
  \end{equation}
  It must, of course, be proved that these projections do indeed form projective
  measurements, and that Alice's measurements commute with Bob's.
  Toward these goals, consider any choice of $i,j \in \{1,\ldots,m\}$, $s,t \in \Delta^*$, and $(z,c) \in \Delta$, and observe that 
  \begin{equation}
	  \begin{aligned}
  	\ip{u_{i,(z,c)t}}{u_{j,s}} &= M_{i,j}((z,c)t,s) \\
  	&= \phi_{i,j}(t^{\mathsmaller{R}}(z,c)s) \\
  	&= \phi_{i,j}(t^{\mathsmaller{R}}(z,c)(z,c)s) \\
  	&= M_{i,j}((z,c)t,(z,c)s) \\
  	&= \ip{ u_{i,(z,c)t} }{ u_{j,(z,c)s} }
	  \end{aligned}
  \end{equation}
  It follows that $u_{j,s}$ and $u_{j,(z,c)s}$ have the same inner product with every vector in the image of $\Pi_c^z$. As every vector in the orthogonal complement of the image of $\Pi_c^z$ is orthogonal to $u_{j,(z,c)s}$, as this vector is contained in the image of $\Pi_c^z$, if follows that
  \begin{equation}
    \label{eq:projection-on-vectors}
    \Pi^z_c u_{j,s} = u_{j,(z,c)s}.
  \end{equation}
  This formula greatly simplifies the required verifications.
  For instance, one has
  \begin{equation}
  	\begin{aligned}
	    \bigip{u_{i,(z,c)t}}{u_{j,(z,d)s}} &= M_{i,j}((z,c)t,(z,d)s) \\
	    & = \phi_{i,j}(t^{\mathsmaller{R}}(z,c)(z,d)s) \\
	    & = 0
    \end{aligned}
  \end{equation}
  for all $i,j\in\{1,\ldots,m\}$, $s,t\in\Delta^{\ast}$, and
  $(z,c),(z,d)\in\Delta$ for which $c\not=d$, and therefore
  $\Pi^z_c \Pi^z_d = 0$
  whenever $(z,c),(z,d)\in\Delta$ satisfy $c\not=d$.
  For each $x\in \SigmaA$, and each $i,j\in\{1,\ldots,m\}$ and
  $s,t\in\Delta^{\ast}$, it holds that
  \begin{equation}
    \sum_{a\in \GammaA} \bigip{u_{i,s}}{\Pi^x_a u_{j,t}}
    = \sum_{a\in \GammaA} \bigip{u_{i,s}}{u_{j,(x,a)t}}
    = \sum_{a\in \GammaA} \phi_{i,j}(s^{\mathsmaller{R}}(x,a)t)
    = \phi_{i,j}(s^{\mathsmaller{R}}t) = \bigip{u_{i,s}}{u_{j,t}}
  \end{equation}
  and therefore
  \begin{equation}
    \sum_{a\in \GammaA}\Pi^x_a = \I, 
  \end{equation}
  for each $x\in \SigmaA$, and along similar lines one finds that
  \begin{equation}
    \sum_{b\in \GammaB}\Pi^y_b = \I
  \end{equation}
  for each $y\in \SigmaA$.
  Finally, for every $i,j\in\{1,\ldots,m\}$, $s,t\in\Delta^{\ast}$, 
  $(x,a)\in\SigmaA \times \GammaA$, and $(y,b)\in\SigmaB \times \GammaB$ we have
  \begin{equation}
    \begin{aligned}
      \Bigip{u_{i,s}}{\Pi^x_a \Pi^y_b u_{j,t}} &= \bigip{u_{i,(x,a)s}}{u_{j,(y,b)t}} \\
      &= \phi_{i,j}\bigl(s^{\mathsmaller{R}}(x,a)(y,b)t\bigr)\\
      &= \phi_{i,j}\bigl(s^{\mathsmaller{R}}(y,b)(x,a)t\bigr)\\
      &= \bigip{u_{i,(y,b)s}}{u_{j,(x,a)t}} \\
      &= \bigip{u_{i,s}}{\Pi^y_b \Pi^x_a u_{j,t}},
    \end{aligned}
  \end{equation}
  and therefore $\bigl[\Pi^x_a,\Pi^y_b\bigr] = 0$.

  It remains to observe that the strategy represented by the pure state
  $u$ and the projective measurements $\{\Pi^x_a\}$ and $\{\Pi^y_b\}$
  yields the commuting measurement assemblage $K$.
  This is also evident from the equation \eqref{eq:projection-on-vectors}, as
  one has
  \begin{equation}
    M_{i,j}((x,a),(y,b)) =
    \bigip{u_{i,(x,a)}}{u_{j,(y,b)}} =
    \Bigip{\Pi^x_a \Pi^y_b}{u_{j,\varepsilon} u_{i,\varepsilon}^{\ast}},
  \end{equation}
  and therefore
  \begin{equation}
    K(a,b|x,y) = \tr_{\H} \Bigl( \bigl(\I\otimes \Pi^x_a \Pi^y_b\bigr) u
    u^{\ast}\Bigr)
  \end{equation}
  for every choice of $x\in \SigmaA$, $y\in \SigmaB$, $a\in \GammaA$, and $b\in \GammaB$.
\end{proof}

\subsection{Examples: Upper-bounding the standard quantum values of extended nonlocal games} \label{sec:examples-upper-bounds-extended-npa}

\subsubsection*{The BB84 extended nonlocal game}

Consider the following extended nonlocal game, $G_{\BB84}$.
\begin{example}\label{ex:bb84-monogamy-game}
	Let $\SigmaA = \SigmaB = \GammaA = \GammaB = \{0,1\}$, define 
	\begin{equation} 
	\begin{aligned}
		V(0,0|0,0) &= E_{0,0} = \begin{pmatrix} 1 & 0 \\ 0 & 0 \end{pmatrix}, \\
		V(1,1|0,0) &= E_{1,1} = \begin{pmatrix} 0 & 0 \\ 0 & 1 \end{pmatrix}, \\
		V(0,0|1,1) &= \frac{1}{2}\left( E_{0,0} + E_{0,1} + E_{1,0} + E_{1,1} \right) = \begin{pmatrix} \frac{1}{2} & \frac{1}{2} \\ \frac{1}{2} & \frac{1}{2} \end{pmatrix} , \\
		V(1,1|1,1) &= \frac{1}{2}\left( E_{0,0} - E_{0,1} - E_{1,0} + E_{1,1} \right) = \begin{pmatrix} \frac{1}{2} & -\frac{1}{2} \\ -\frac{1}{2} & \frac{1}{2} \end{pmatrix}, 
	\end{aligned}
	\end{equation}
	define
	\begin{equation}
		V(a,b|x,y) = \begin{pmatrix} 0 & 0 \\ 0 & 0 \end{pmatrix},	
	\end{equation}		
 	for all $a \not= b$ or $x \not= y$, define $\pi(0,0) = \pi(1,1) = 1/2$, and define $\pi(x,y) = 0$ if $x \not= y$. 
\end{example}
Alice and Bob win $G_{\BB84} = (\pi,V)$ if and only if the responses from Alice and Bob, $a$ and $b$, are equal to the outcome of the measurement that the referee makes on its system. In the event where the referee sends $x = y = 0$ and Alice and Bob respond with either $a = b = 0$ or $a = b = 1$, the referee measures his state against either the measurement $V(0,0|0,0)$ or $V(1,1|0,0)$. These measurements correspond to the well known $0/1$ basis, which is sometimes written in Dirac notation elsewhere in the literature as $\ket{0}\bra{0}$ and $\ket{1}\bra{1}$. Likewise, if the referee sends $x = y = 1$ to Alice and Bob and they respond with either $a = b = 0$ or $a = b = 1$, the referee measures his state against either $V(0,0|1,1)$ or $V(1,1|1,1)$. These measurements correspond to the $+/-$ basis, which is typically denoted as $\ket{+}\bra{+}$ and $\ket{-}\bra{-}$ in Dirac notation elsewhere in the literature. This particular game is one that we will see again in Chapter~\ref{chap:monogamy_games} under the name of the BB84 monogamy-of-entanglement game, and it was initially introduced in~\cite{Tomamichel2013}.

Recall that the extended QC hierarchy states that an upper bound on the standard quantum value of any extended nonlocal game may be obtained by maximizing the following quantity 
\begin{align} \label{eq:npa-sdp-max}
	\sum_{(x,y) \in \SigmaA \times \SigmaB} \pi(x,y) \sum_{(a,b) \in \GammaA \times \GammaB} \biggip{V(a,b|x,y)}{M^{(k)}((x,a),(y,b))},
\end{align}
where $M^{(k)}$ is a $k$-th order admissible matrix. For the game $G_{\BB84}$, this quantity may be explicitly written as 
\begin{equation} \label{eq:bb84-kth-order-admissible-matrix}
	\begin{aligned}
		&\frac{1}{2} \left( \biggip{ V(0,0|0,0) }{ M^{(k)}((0,0),(0,0))} + \biggip{ V(1,1|0,0) }{ M^{(k)}((0,1),(0,1))} \right) + \\
		&\frac{1}{2} \left( \biggip{ V(0,0|1,1) }{ M^{(k)}((1,0),(1,0))} + \biggip{ V(1,1|1,1) }{ M^{(k)}((1,1),(1,1))} \right),
	\end{aligned}
\end{equation}
for some level $k$. For simplicity, we shall just consider the first level at $k = 1$ of the extended QC hierarchy for $G_{\BB84}$. As it turns out, and as we shall see in the coming sections, the first level converges to the standard quantum value, and therefore going any higher in the hierarchy will not yield any better approximations. 

The software listing~\ref{code:first-level-qc-hierarchy-bb84} in Appendix~\ref{chap:AppendixA} maximizes the objective function from equation~\eqref{eq:npa-sdp-max} subject to the constraints that the matrix $M^{(1)}$ is a first-order admissible matrix. Running the listing, one obtains the following matrix for $M^{(1)}$:
\begin{align}
	M^{(1)} = \begin{pmatrix}
				M^{(1)}_{1,1} & M^{(1)}_{1,2} \\[2mm]
				M^{(1)}_{2,1} & M^{(1)}_{2,2}
			  \end{pmatrix},
\end{align}
where 
\begin{equation}
	\begin{aligned}
	M^{(1)}_{1,1} = M^{(1)}_{2,2} =& \begin{pmatrix} \alpha & \beta_+ & \beta_- & \alpha & \alpha & \beta_+ & \beta_- & \alpha & \alpha \\
													\beta_+ & \beta_+ & 0 & \alpha\beta_+ & \alpha\beta_+ & \beta_+ & 0 & \alpha\beta_+ & \alpha \beta_+ \\														\beta_- & 0 & \beta_- & \alpha \beta_- & \alpha \beta_- & 0 & \beta_- & \alpha \beta_- & \alpha \beta_- \\
													\alpha & \alpha \beta_+ & \alpha \beta_- & \alpha & 0 & \alpha \beta_+ & \alpha \beta_- & \alpha & 0 \\
													\alpha & \alpha \beta_+ & \alpha \beta_- & 0 & \alpha & \alpha \beta_+ & \alpha \beta_- & 0 & \alpha \\
													\beta_+ & \beta_+ & 0 & \alpha \beta_+ & \alpha \beta_+ & \beta_+ & 0 & \alpha \beta_+ & \alpha \beta_+ \\
													\beta_- & 0 & \beta_- & \alpha \beta_- & \alpha \beta_- & 0 & \beta_- & \alpha \beta_- & \alpha \beta_- \\
													\alpha & \alpha \beta_+ & \alpha \beta_- & \alpha & 0 & \alpha \beta_+ & \alpha \beta_- & \alpha & 0 \\
													\alpha & \alpha \beta_+ & \alpha \beta_- & 0 & \alpha & \alpha \beta_+ & \alpha \beta_- & 0 & \alpha			
						\end{pmatrix}, \\
	M^{(1)}_{1,2} = M^{(1)}_{2,1} =& \begin{pmatrix} 
										0 & 0 & 0 & \gamma & -\gamma & 0 & 0 & \gamma & -\gamma \\
								0 & 0 & 0 & \alpha \gamma & -\alpha \gamma & 0 & 0 & \alpha \gamma & -\alpha \gamma \\
					0 & 0 & 0 & \alpha \gamma & -\alpha \gamma & 0 & 0 & \alpha \gamma & -\alpha \gamma \\
					0 & 0 & 0 & \alpha \gamma & -\alpha \gamma & 0 & 0 & \alpha \gamma & -\alpha \gamma \\
						\gamma & \alpha \gamma & \alpha \gamma & \gamma & 0 & \alpha \gamma & \alpha \gamma & \gamma & 0 \\
				-\gamma & -\alpha \gamma & -\alpha \gamma & 0 & -\gamma & -\alpha \gamma & -\alpha \gamma & 0 & -\gamma \\
					0 & 0 & 0 & \alpha \gamma & -\alpha \gamma & 0 & 0 & \alpha \gamma & -\alpha \gamma \\ 
				0 & 0 & 0 & \alpha \gamma & -\alpha \gamma & 0 & 0 & \alpha \gamma & -\alpha \gamma \\
				\gamma & \alpha \gamma & \alpha \gamma & \gamma & 0 & \alpha \gamma & \alpha \gamma & \gamma & 0 \\
				-\gamma & -\alpha \gamma & -\alpha \gamma & 0 & -\gamma & -\alpha \gamma & -\alpha \gamma & 0 & -\gamma 															 \end{pmatrix},
	\end{aligned}
\end{equation}
where we define the constants 
\begin{align}
	\alpha = 1/2, \quad \beta_{\pm} = \frac{1}{8} \left( 2 \pm \sqrt{2} \right), \quad \textnormal{and} \quad \gamma = \frac{\sqrt{2}}{8}.
\end{align}
One may verify that the matrix $M^{(1)}$ is a first-order admissible matrix and that the equation~\eqref{eq:bb84-kth-order-admissible-matrix} for $k = 1$ yields the value of $\cos^2(\pi/8) \approx 0.8536$ where 
\begin{equation}
	\begin{aligned}
		M^{(1)}((0,0),(0,0)) = \begin{pmatrix} \beta_+ & 0 \\ 0 & \beta_- \end{pmatrix}, \quad M^{(1)}((0,1),(0,1)) = \begin{pmatrix} \beta_- & 0 \\ 0 & \beta_+ \end{pmatrix}, \\
		M^{(1)}((1,0),(1,0)) = \begin{pmatrix} \alpha & \gamma \\ \gamma & \alpha \end{pmatrix}, \quad M^{(1)}((1,1),(1,1)) = \begin{pmatrix} \alpha & -\gamma \\ -\gamma & \alpha \end{pmatrix}.
	\end{aligned}
\end{equation}
In Section~\ref{sec:lower-bound-extended-nonlocal-games}, we will verify that $\cos^2(\pi/8)$ is indeed the standard quantum value as this value will also arise when calculating the lower bound of $G_{\BB84}$.  

\subsubsection*{The CHSH extended nonlocal game}

Let us now consider another game, $G_{\CHSH}$. 
\begin{example}[CHSH extended nonlocal game]
	Let $\SigmaA = \SigmaB = \GammaA = \GammaB = \{0,1\}$, define a collection of measurements $\{V(a,b | x,y) : a \in \GammaA, \ b \in \GammaB, \ x \in \SigmaA, \ y \in \SigmaB \} \subset \Pos(\R)$ such that 
	\begin{equation} \label{eq:chsh-meas-ops}
	\begin{aligned}
		V(0,0|0,0) = V(0,0|0,1) = V(0,0|1,0) &= E_{0,0}, \\
		V(1,1|0,0) = V(1,1|0,1) = V(1,1|1,0) &= E_{1,1}, \\
		V(0,1|1,1) &= \frac{1}{2} \left( E_{0,0} + E_{0,1} + E_{1,0} + E_{1,1} \right), \\
		V(1,0|1,1) &= \frac{1}{2} \left( E_{0,0} - E_{0,1} - E_{1,0} + E_{1,1} \right),
	\end{aligned}
	\end{equation}
	define 
	\begin{equation} \label{eq:zero-chsh}
		V(a,b|x,y) = \begin{pmatrix} 0 & 0 \\ 0 & 0 \end{pmatrix}
	\end{equation}
for all $a \oplus b \not= x \land y$, and define $\pi(0,0) = \pi(0,1) = \pi(1,0) = \pi(1,1) = 1/4$. 
\end{example}
In the event that $a \oplus b \not= x \land y$, the referee's measurement corresponds to the zero matrix from equation~\eqref{eq:zero-chsh}. If instead it happens that $a \oplus b = x \land y$, the referee then proceeds to measure with respect to one of the measurement operators from equation~\eqref{eq:chsh-meas-ops}. This winning condition is reminiscent of the standard CHSH nonlocal game. For $G_{\CHSH}$, we may again consider the equation~\eqref{eq:npa-sdp-max} where the quantity may be explicitly written as 
\begin{equation}
	\begin{aligned}
		& \frac{1}{4} \left( \biggip{V(0,0|0,0)}{M^{(k)}((0,0),(0,0))} + \biggip{V(1,1|0,0)}{M^{(k)}((0,1),(0,1))} \right) + \\
		& \frac{1}{4} \left( \biggip{V(0,0|0,1)}{M^{(k)}((0,0),(1,0))} + \biggip{V(1,1|0,1)}{M^{(k)}((0,1),(1,1))} \right) + \\
		& \frac{1}{4} \left( \biggip{V(0,0|1,0)}{M^{(k)}((1,0),(0,0))} + \biggip{V(1,1|1,0)}{M^{(k)}((1,1),(0,1))} \right) + \\
		& \frac{1}{4} \left( \biggip{V(0,1|1,1)}{M^{(k)}((1,0),(1,1))} + \biggip{V(1,0|1,1)}{M^{(k)}((1,1),(1,0))} \right),
	\end{aligned}
\end{equation}
for some level $k$. Unlike $G_{\BB84}$, the first level of the extended QC hierarchy is not sufficient for obtaining the standard quantum value of $G_{\CHSH}$. Indeed, running the software listing~\ref{code:first-level-qc-hierarchy-chsh} that implements the first level of the extended QC hierarchy yields a value of $\approx 0.7578$, while the non-signaling value of this game yields a value of $3/4$ (refer to software listing~\ref{code:ns-val-chsh-enlg}) as does the lower bound on the standard quantum value of $G_{\CHSH}$. The lower bound technique will be elaborated on further in Section~\ref{sec:lower-bound-extended-nonlocal-games}.

\section{Lower bounds for extended nonlocal games: the see-saw method}
\label{sec:lower-bound-extended-nonlocal-games}

Our lower bound heuristic for the class of extended nonlocal games is based on the work of Liang and Doherty~\cite{Liang2007}, where they provide a lower bound technique for Bell inequalities, or equivalently, nonlocal games. The primary idea of their algorithm is to note that fixing measurements on one system yields the optimal measurements of the other system via an SDP. The algorithm proceeds in an iterative manner between two SDPs. In the first SDP, we assume that Bob's measurements are fixed, and Alice's measurements are to be optimized over. In the second SDP, we take Alice's optimized measurements from the first SDP and now optimize over Bob's measurements. This method is repeated until the quantum value reaches a desired numerical precision. This ``see-saw'' type iteration was done in~\cite{Werner2001} by Werner and Wolf and served as a basis of inspiration for Liang and Doherty's method. It is also worthwhile to mention that in~\cite{Ito2006} the authors showed concurrently with~\cite{Liang2007} that there exists an SDP that achieves a lower bound on the quantum value for a nonlocal game. 

We must slightly adapt the Liang and Doherty lower bound algorithm to take into account the actions of the referee for our extended nonlocal game. In our scenario, we shall represent Alice's actions in terms of the residual states acting on the referee and Bob as the set $\{ \rho_a^x : x \in \SigmaA, \ a \in \GammaA \} \subset \Pos(\R \otimes \B)$ where 
\begin{align}
	\rho_a^x = \tr_{\A} \left( \left( \I_{\R} \otimes A_a^x \otimes \I_{\B} \right) \rho \right) \in \Pos(\R \otimes \B).
\end{align}   
It is then necessary and sufficient that $\sum_{a \in \GammaA} \rho_a^x = \tau \in \Density(\R \otimes \B)$ for all $x \in \SigmaA$. It then holds that we can express the probability that Alice and Bob's standard quantum strategy wins in as 
\begin{align} \label{eq:lower-bound-objective-function}
	\sum_{(x,y) \in \Sigma} \pi(x,y) \sum_{(a,b) \in \Gamma} \bigip{V(a,b|x,y) \otimes B_b^y}{\rho_a^x}.
\end{align}
Writing the above conditions in terms of an SDP, we have that 
\begin{center}
		\centerline{\underline{Lower bound SDP-1}}\vspace{-7mm}
		\begin{equation} \label{sdp:lower-bound-SDP1}		
  		\begin{split}
  			\text{maximize:}\quad & \sum_{(x,y) \in \Sigma} \pi(x,y) \sum_{(a,b) \in \Gamma} \bigip{V(a,b|x,y) \otimes B_b^y}{\rho_a^x} \\
      \text{subject to:}\quad & \sum_{a \in \GammaA} \rho_a^x = \tau, \qquad \quad \ \forall x \in \SigmaA, \\
      & \rho_a^x \in \Pos(\R \otimes \B), \quad \forall x \in \SigmaA, \ \forall a \in \GammaA, \\
      & \tau \in \Density(\R \otimes \B),
  		\end{split}
		\end{equation}
\end{center}
where the measurements of the referee represented as the collection $\{V(a,b|x,y)\}$ as well as the sets of measurements for Bob, represented as the collection $\{B_b^y\}$ are fixed where the collection of operators $\{\rho_a^x\}$ are the variables that we wish to optimize. Now we consider the second SDP. First, observe that we can write equation~\eqref{eq:lower-bound-objective-function} as 
\begin{align}
	\sum_{(x,y) \in \Sigma} \pi(x,y) \sum_{(a,b) \in \Gamma} \bigip{ \Phi(B_b^y) }{\rho_a^x},
\end{align}
where the mapping $\Phi \in \Trans(\B, \R \otimes \B)$ is defined as 
\begin{align}
	\Phi(B_b^y) = \pi(x,y) V(a,b|x,y) \otimes B_b^y.
\end{align}
We calculate the unique adjoint mapping $\Phi^* \in \Trans(\R \otimes \B, \B)$, as 
\begin{align}
	\Phi^*(\rho_a^x) = \tr_{\R} \left( \left( V(a,b|x,y)^* \otimes \I_{\B} \right) \rho_a^x \right),
\end{align}
which may be verified from 
\begin{equation}
	\begin{aligned}
		\bigip{\Phi(B_b^y)}{\rho_a^x} &= \bigip{V(a,b|x,y) \otimes B_b^y}{\rho_a^x} \\
		&= \tr \left( \left( V(a,b|x,y)^* \otimes (B_b^y)^* \right) \rho_a^x \right) \\
		&= \tr \left( \left( \I_{\B} \otimes (B_b^y)^* \right) \left( V(a,b|x,y)^* \otimes \I_{\B} \right) \rho_a^x \right) \\
		&= \tr \left( (B_b^y)^* \tr_{\R} \left( V(a,b|x,y)^* \otimes \I_{\B} \right) \rho_a^x \right) \\
		&= \bigip{B_b^y}{\tr_{\R}\left( \left(V(a,b|x,y)^* \otimes \I_{\B} \right) \rho_a^x \right) }.
	\end{aligned}
\end{equation}

From this, we can define the second SDP.
\begin{center}
		\centerline{\underline{Lower bound SDP-2}}\vspace{-7mm}
		\begin{equation} \label{sdp:lower-bound-SDP2}
  		\begin{split}
      \text{maximize:}\quad & \sum_{(x,y) \in \Sigma} \pi(x,y) \sum_{(a,b) \in \Gamma} \bigip{B_b^y}{\Phi^*(\rho_a^x)} \\
      \text{subject to:}\quad & \sum_{b \in \GammaB} B_b^y = \I_{\B}, \qquad \ \forall y \in \SigmaB, \\
      & B_b^y \in \Pos(\B), \qquad \ \ \forall y \in \SigmaB, \ b \in \GammaB. 
  		\end{split}
		\end{equation}
\end{center} 
While the optimization procedure is not guaranteed to converge to the actual quantum value, we can perform our extended QC hierarchy to check if the lower and upper bounds are in agreement to determine optimality. If this is indeed the case, we can extract the explicit strategy that Alice and Bob perform via this lower bound method.


Bob's measurements are given directly from the formulation of the SDP, while Alice's measurements may be obtained by the following equation
\begin{align} \label{eq:extract-alice-meas-ops-lb}
	A_a^x = \tau^{-1/2} \tr_{\B}(\rho_a^x) \tau^{-1/2}
\end{align} 
for all $a \in \GammaA$ and $x \in \SigmaA$. 

\subsection{Examples: Lower-bounding the standard quantum values of extended nonlocal games} \label{sec:examples-lower-bounds}

\subsubsection*{The BB84 extended nonlocal game}

We revisit $G_{\BB84}$, the BB84 extended nonlocal game considered in Section~\ref{sec:examples-upper-bounds-extended-npa}. We observed that for $k = 1$, the extended QC hierarchy gave us a value of $\cos^2(\pi/8)$. We also claimed that this value does indeed correspond to the standard quantum value of the game, and going any higher in the hierarchy would not yield any better approximations to this value. In this example, we shall compute the lower bound of $G_{\BB84}$ and verify that the lower and upper bounds agree, which confirms that computing for any higher levels of $k$ in the extended QC hierarchy would not be useful.

The software listing~\ref{code:bb84-enlg-lower-bound} in Appendix~\ref{chap:AppendixA} computes the lower bound of $G_{\BB84}$ using the two semidefinite programs from equations~\eqref{sdp:lower-bound-SDP1} and~\eqref{sdp:lower-bound-SDP2}. In the first semidefinite program, we are given the measurements that the referee uses (as defined in example~\ref{ex:bb84-monogamy-game}) as well as some collection of measurements for Bob. When we start the see-saw algorithm, we simply generate random unitaries of appropriate dimension to represent the measurements that Bob may use. These measurement operators will change as we go back and forth between the semidefinite programs. The variable that we are optimizing with respect to is $\rho_a^x$, which represents Alice's actions. 

We then plug in the variables that we obtain from the first semidefinite program into the second semidefinite program. We repeat this process until the desired threshold is reached. In this example, the value of $\cos^2(\pi/8)$ is obtained almost immediately, and therefore allows one to conclude that since the upper and lower bounds are in agreement, that $\omega^*(G_{\BB84}) = \cos^2(\pi/8)$. Furthermore using equation~\eqref{eq:extract-alice-meas-ops-lb}, we can also obtain the strategy that Alice uses to obtain this value, where the measurement operators of Alice are given explicitly as 
\begin{equation}
	\begin{aligned}
		A_0^0 = A_1^1 &= \begin{pmatrix}
							\cos^2(\pi/8) & -\sin(\pi/8)\cos(\pi/8) \\
							-\sin(\pi/8)\cos(\pi/8) & \sin^2(\pi/8)
   						 \end{pmatrix}, \\
		A_1^0 = A_0^1 &= \begin{pmatrix}		
							\sin^2(\pi/8) & \sin(\pi/8)\cos(\pi/8) \\
							\sin(\pi/8)\cos(\pi/8) & \cos^2(\pi/8)
						 \end{pmatrix}.
	\end{aligned}
\end{equation}
In this particular example, Bob's measurement operators can take the form of any valid measurement operators. 

\chapter{Monogamy-of-Entanglement Games}
\label{chap:monogamy_games}
In this chapter, we shall consider a particular type of extended nonlocal game referred to as a \emph{monogamy-of-entanglement game}, which was initially introduced in~\cite{Tomamichel2013}. 

In Section~\ref{sec:monogamy-of-entanglement-games}, we formally present this model and prove a number of properties about this class of game. In particular, we will study the relationship between the standard quantum and unentangled strategies of certain monogamy-of-entanglement games. In Section~\ref{sec:parallel-rep-moe-games} we will study the parallel repetition of monogamy-of-entanglement games, and in Section~\ref{sec:upper-and-lower-bounds-moe-games} we will present an example of a monogamy-of-entanglement game that Alice and Bob win with higher probability in the event that they use a standard quantum strategy in place of an unentangled strategy. 

This chapter is based on joint work with Nathaniel Johnston, Rajat Mittal, and John Watrous~\cite{Johnston2015a}.

\minitoc

\section{Monogamy-of-entanglement games}
\label{sec:monogamy-of-entanglement-games}

\index{monogamy-of-entanglement game}{\emph{Monogamy-of-entanglement games}} are a special type of extended nonlocal game, and were originally introduced and studied by Tomamichel, Fehr, Kaniewski, and Wehner~\cite{Tomamichel2013}. Monogamy-of-entanglement games received their namesake as they serve as a framework to conceptualize the fundamental monogamy property exhibited by entangled qubits~\cite{Coffman2000}. In short, this property states that for three possibly entangled qubits contained in the registers $\reg{X_0}$, $\reg{X_1}$, and $\reg{X_2}$, that if $\reg{X_i}$ and $\reg{X_j}$ are maximally entangled, then $\reg{X_k}$ is completely unentangled with qubits $\reg{X_i}$ and $\reg{X_j}$ for $i \not= j \not= k$ where $i,j,k \in \{0,1,2\}$. This phenomena has been studied in a number of other works~\cite{Terhal2001, Terhal2004, Koashi2004, Osborne2006}.

The manner in which a monogamy-of-entanglement game proceeds is similar to an extended nonlocal game. After Alice and Bob supply the referee with a quantum system, we now assume that the referee selects a single question at random, and sends this same question to both Alice and Bob. The winning condition of a monogamy-of-entanglement game is predicated on the ability for Alice and Bob to respond with the same answer, and that this answer must agree with the measurement outcome of the referee.

More formally, we specify a monogamy-of-entanglement game as $G = (\pi,R)$ where $\pi : \Sigma \rightarrow \left[0,1\right]$ is a probability distribution defined over an alphabet $\Sigma$ and where $R$ is a function of the form $R : \Gamma \times \Sigma \rightarrow \Pos(\R)$ where $\R = \complex^m$ is a complex Euclidean space of dimension $m$ belonging to the referee and where $\Gamma$ is an alphabet. The function $R$ corresponds to a collection of measurement operators for the referee where $R(a|x)$ is the measurement that corresponds to question $x \in \Sigma$ and answer $a \in \Gamma$. The function $R$ must satisfy
\begin{align}
	\sum_{a \in \Gamma} R(a|x) = \I_{\R}
\end{align} 
for every $x \in \Sigma$.
 
A monogamy-of-entanglement game closely follows the way in which an extended nonlocal game is played. First, Alice and Bob prepare a state $\sigma \in \Density(\U \otimes \R \otimes \V)$ and share it with the referee. The referee then selects a single question $x \in \Sigma$ according to the probability distribution $\pi$, and sends $x$ to both Alice and Bob. Alice and Bob then produce and send respective responses $a$ and $b$ to the referee. When the referee receives $a$ and $b$, it performs a measurement $\{R(c|x) : c \in \Gamma\}$ on its portion of the shared state, yielding some measurement outcome. The game is won if and only if the measurement outcomes $a$ and $b$ produced by Alice and Bob agree with the outcome of the referee's measurement. A monogamy-of-entanglement game is depicted in Figure~\ref{fig:monogamy-game}. 
\begin{figure}[!htpb] 
	\begin{center}
		\includegraphics[scale=0.9]{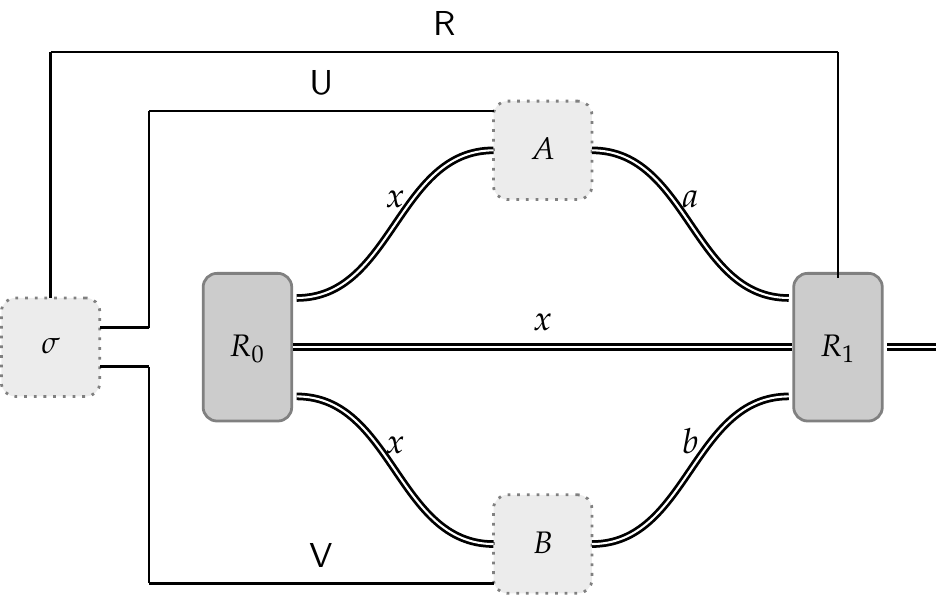}
	\end{center}
		\caption[A monogamy-of-entanglement game.]{A monogamy-of-entanglement game. The state $\sigma \in \Density(\U \otimes \R \otimes \V)$ contained in registers $(\reg{U},\reg{R},\reg{V})$ is prepared by Alice and Bob, where $\reg{R}$ is sent to the referee and $\reg{U}$ belongs to Alice and $\reg{V}$ belongs to Bob. The referee selects question $x$ according to the $\pi$ distribution, and sends $x$ to both Alice and Bob. Alice and Bob then generate and send answers $a$ and $b$ to the referee. Alice and Bob win if and only if all measurement outcomes agree. }
		\label{fig:monogamy-game}
\end{figure}

\subsection{Strategies and values of monogamy-of-entanglement games}
Since monogamy-of-entanglement games are a type of extended nonlocal game, one may also define unentangled strategies, standard quantum strategies, commuting measurement strategies, and non-signaling strategies as considered in Chapter~\ref{chap:extended_nonlocal_games} in a similar manner. We shall define some of these strategies and their corresponding values for the case of monogamy-of-entanglement games explicitly. 

For instance, a \index{standard quantum strategy (monogamy-of-entanglement game)}{\emph{standard quantum strategy}} for a monogamy-of-entanglement game consists of finite-dimensional complex Euclidean spaces $\U$ for Alice and $\V$ for Bob, a quantum state $\sigma \in \Density(\U \otimes \R \otimes \V)$, and two collections of measurements 
\begin{align}
	\{ A_a^x : a \in \Gamma \} \subset \Pos(\U) \quad \textnormal{and} \quad \{ B_a^x : a \in \Gamma \} \subset \Pos(\V),
\end{align}
for each $x \in \Sigma$. The measurement operators satisfy the constraint that 
\begin{align}
	\sum_{a \in \Gamma} A_a^x = \I_{\U} \quad \textnormal{and} \quad \sum_{a \in \Gamma} B_a^x = \I_{\V},
\end{align}
for each $x \in \Sigma$. 
For a monogamy-of-entanglement game, the winning probability for Alice and Bob when they use a standard quantum strategy is given by 
\begin{align}
	\sum_{x \in \Sigma} \pi(x) \sum_{a \in \Gamma} \biggip{A_a^x \otimes R(a|x) \otimes B_a^x}{\sigma}.
\end{align}
In fact, we may simplify the above expression slightly. Recall from Section~\ref{sec:standard-quantum-strategies-extended-nonlocal-games} that we may assume that $\sigma \in \Density(\U \otimes \R \otimes \V)$ is pure for any extended nonlocal game, and the measurements of the referee are positive semidefinite, it follows by convexity that we may write the winning probability for Alice and Bob when they use a standard quantum strategy as 
\begin{align} \label{eq:monogamy-ent-val}
	\biggnorm{\sum_{x \in \Sigma} \pi(x) \sum_{a \in \Gamma} A_a^x \otimes R(a|x) \otimes B_a^x }.
\end{align}
For a given monogamy-of-entanglement game $G = (\pi,R)$, we write $\omega^*(G)$ to denote the standard quantum value of $G$, which is the supremum winning value of Alice and Bob's winning probability over all standard quantum strategies for $G$. 

An \index{unentangled strategy (monogamy-of-entanglement game)}{\emph{unentangled strategy}} for a monogamy-of-entanglement game is simply a standard quantum strategy for which the state $\sigma \in \Density(\U \otimes \R \otimes \V)$ initially prepared by Alice and Bob is fully separable. The unentangled value of a monogamy-of-entanglement game, $G$, can be directly derived from the unentangled value of an extended nonlocal game from equation~\eqref{eq:enlg-unentangled-value} as 
\begin{align} \label{eq:monogamy-unentangled-value}
	\omega(G) = \max_{f : \Sigma \rightarrow \Gamma} \biggnorm{\sum_{x \in \Sigma} \pi(x) R(f(x)|x) },
\end{align}
noting again that Alice and Bob only win in a monogamy-of-entanglement game when their measurement outcomes agree with the measurement outcome of the referee. 

A \index{non-signaling strategy (monogamy-of-entanglement game)}{\emph{non-signaling strategy}} for a monogamy-of-entanglement game consists of a non-signaling assemblage $K : \Gamma \times \Sigma \rightarrow \Pos(\R)$ such that 
\begin{align}
	\sum_{a \in \Gamma} K(a,b|x,y) = \xi_b^y \quad \textnormal{and} \quad \sum_{b \in \Gamma} K(a,b|x,y) = \rho_a^x,
\end{align}
for all $x \in \Sigma$ and $y \in \Sigma$ where $\{\xi_b^y : y \in \Sigma, \ b \in \Gamma \}$ and $\{ \rho_a^x : x \in \Sigma, \ a \in \Gamma \}$ are collections of operators satisfying 
\begin{align}
	\sum_{a \in \Gamma} \rho_a^x = \tau = \sum_{b \in \Gamma} \xi_b^y,
\end{align}
for all $x \in \Sigma$ and $y \in \Sigma$ and where $\tau \in \Density(\R)$ is a density operator. For any monogamy-of-entanglement game the winning probability when Alice and Bob use a non-signaling strategy is given by 
\begin{align}
	\sum_{x \in \Sigma} \pi(x) \sum_{a \in \Gamma} \biggip{R(a|x)}{K(a,a|x,x)}.
\end{align}
For a monogamy-of-entanglement game, $G$, the non-signaling value, $\omega_{\ns}(G)$ is the supremum value of the winning probability of $G$ taken over all non-signaling strategies for Alice and Bob.

\subsection{The BB84 monogamy-of-entanglement game} \label{sec:the-bb84-monogamy-of-entanglement-game}

In the following example, we shall consider one type of monogamy-of-entanglement game referred to as the \index{BB84 monogamy-of-entanglement game}{\emph{BB84 monogamy-of-entanglement game}}, denoted as $G_{\BB84}$ for short. As we shall see, the name of the game comes from the sets of measurements that the referee uses, which are defined from the BB84 measurement operators~\cite{Bennett1984}. This game was initially introduced and studied in~\cite{Tomamichel2013}. Note that we also already previously considered this game in Chapter~\ref{chap:extended_npa_hierarchy} when we looked at the examples found in Sections~\ref{sec:examples-upper-bounds-extended-npa} and~\ref{sec:examples-lower-bounds}.
\begin{example}[BB84 monogamy-of-entanglement game~\cite{Tomamichel2013}]\label{ex:bb84-monogamy-game}
	Let $\Sigma = \Gamma = \{0,1\}$, define 
		\begin{equation} \label{eq:bb84-meas-ops}
	\begin{aligned}
		R(0|0) &= E_{0,0} = \begin{pmatrix} 1 & 0 \\ 0 & 0 \end{pmatrix}, \\
		R(1|0) &= E_{1,1} = \begin{pmatrix} 0 & 0 \\ 0 & 1 \end{pmatrix}, \\
		R(0|1) &= \frac{1}{2}\left( E_{0,0} + E_{0,1} + E_{1,0} + E_{1,1} \right) = \begin{pmatrix} \frac{1}{2} & \frac{1}{2} \\ \frac{1}{2} & \frac{1}{2} \end{pmatrix} , \\
		R(1|1) &= \frac{1}{2}\left( E_{0,0} - E_{0,1} - E_{1,0} + E_{1,1} \right) = \begin{pmatrix} \frac{1}{2} & -\frac{1}{2} \\ -\frac{1}{2} & \frac{1}{2} \end{pmatrix},
	\end{aligned}
	\end{equation}
	and define $\pi(0) = \pi(1) = 1/2$. Then the BB84 monogamy-of-entanglement game, denoted as $G_{\BB84}$, is specified by $G_{\BB84} = (\pi,R)$. 
\end{example}
In~\cite{Tomamichel2013}, the authors also showed that even if Alice and Bob adopt a standard quantum strategy for $G_{\BB84}$, they will perform \emph{no better} than had they simply used an unentangled strategy, 
\begin{align}
	\omega(G_{\BB84}) = \omega^*(G_{\BB84}) = \cos^2(\pi/8) \approx 0.8536.
\end{align}  
That is to say, Alice and Bob gain no advantage in sharing entanglement with the referee. Recall that in Sections~\ref{sec:examples-upper-bounds-extended-npa} and~\ref{sec:examples-lower-bounds}, we computed the lower and upper bound on the standard quantum value of $G_{\BB84}$ and found that both values agree and are equal to $\cos^2(\pi/8)$. 

\subsection{Comparing standard quantum and unentangled strategies for monogamy-of-entanglement games} \label{sec:monog-classical-quantum}

Recall from Section~\ref{sec:the-bb84-monogamy-of-entanglement-game} that $\omega(G_{\BB84}) = \omega^*(G_{\BB84})$, meaning that it makes no difference whether Alice and Bob adopt a standard quantum or unentangled strategy for $G_{\BB84}$, as they will win with the same probability either way. A natural question then is whether this behavior persists in general for the class of monogamy-of-entanglement games. Specifically, is it the case that for any monogamy-of-entanglement game, $G$, that 
\begin{align}
	\omega(G) = \omega^*(G)?
\end{align}

In this section, we shall show that for any monogamy-of-entanglement game where the size of the question set is two and the size of the answer set is arbitrary, the standard quantum and unentangled values are indeed equal. However, in Section~\ref{sec:MUB-4-3}, we shall show that this behavior is \emph{not true} for the entire class of monogamy-of-entanglement games and present an explicit example of a monogamy-of-entanglement game that yields a strictly higher standard quantum value than unentangled value.  

\begin{theorem}
	Let $G$ be any monogamy-of-entanglement game for which the question set $\Sigma$ satisfies $\abs{\Sigma}=2$. It holds that 
	\begin{align}
		\omega(G) = \omega^*(G).
	\end{align}
\end{theorem}

\begin{proof}
	It is evident that $\omega(G) \leq \omega^*(G)$, as this is the case for every extended nonlocal game (and therefore every monogamy-of-entanglement game), so it remains to prove the reverse inequality. 
	Assume without loss of generality that $\Sigma = \{0,1\}$, and that $G = (\pi,R)$ for $\pi(0) = \lambda$ and $\pi(1) = 1 - \lambda$. Consider any choice of projective measurements 
	\begin{align}
		\left \{ A_a^0 : a \in \Gamma \right \} \quad \textnormal{and} \quad \left \{ A_a^1 : a \in \Gamma  \right \}
	\end{align}
on $\U$ for Alice and 
	\begin{align}
		\left \{ B_a^0 : a \in \Gamma \right \} \quad \textnormal{and} \quad \left \{ B_a^1 : a \in \Gamma \right \}
	\end{align}
on $\V$ for Bob. First, note that for an optimal choice of the initial state, we can write the standard quantum value of $G$ in terms of the following equation 
	\begin{align} \label{eq:classical-quantum-monogamy}
		\omega^*(G) = \biggnorm{\lambda \sum_{a \in \Gamma} A_a^0 \otimes R(a|0) \otimes B_a^0 + (1-\lambda) \sum_{b \in \Gamma}  A_b^1 \otimes R(b|1) \otimes B_b^1 }.
	\end{align}
Note that the operator inside the norm of equation~\eqref{eq:classical-quantum-monogamy} is positive semidefinite since all of the measurement operators are also positive semidefinite. 

Recall that for positive semidefinite operators $P \in \Pos(\U)$ and $Q \in \Pos(\V)$ that if $P \leq Q$ then it implies that $\norm{P} \leq \norm{Q}$. To observe this fact, note that for a positive semidefinite operator, $X$, the spectral norm yields the largest eigenvalue of that operator. An equivalent way to state that $X$ is positive semidefinite is to say that $X$ is Hermitian with nonnegative eigenvalues. 


We can, therefore, upper bound $\omega^*(G)$ in the following way 
\begin{align} \label{eq:classical-quantum-monogamy-identity}
	\omega^*(G) \leq \biggnorm{\lambda \sum_{a \in \Gamma} A_a^0 \otimes  R(a|0) \otimes \I_{\V} + (1 - \lambda) \sum_{b \in \Gamma} \I_{\U} \otimes R(b|1) \otimes B_b^1}.
\end{align}
Since we enforce that the operators $A_a^x$ and $B_b^y$ are valid measurement operators, it holds that 
\begin{align}
	\sum_{a \in \Gamma} A_a^x = \I_{\U} \quad \textnormal{and} \quad \sum_{b \in \Gamma} B_b^y = \I_{\V}
\end{align}
for all $x \in \Sigma$ and $y \in \Sigma$. Let us now replace the identity operators from equation~\eqref{eq:classical-quantum-monogamy-identity} with these sums to obtain
\begin{align}
	\omega^*(G) \leq \biggnorm{\lambda \sum_{a,b \in \Gamma} A_a^0 \otimes R(a|0) \otimes B_b^1 + (1 - \lambda) \sum_{a,b \in \Gamma} A_a^0 \otimes R(b|1) \otimes B_b^1}.
\end{align}
Since $\{A_a^0 \otimes B_b^1 : a,b \in \Gamma\}$ are pairwise orthogonal projections, i.e. that $\ip{A_a^0 \otimes B_b^1}{A_{a^{\prime}}^0 \otimes B_{b^{\prime}}^1} = 0$ for $a \not= a^{\prime}$ and $b \not= b^{\prime}$ and also that it holds that 
\begin{align}
	\biggnorm{\sum_{k} A_k \otimes \Pi_k} = \max_k \norm{A_k}
\end{align}
for a projective measurement $\{\Pi_k\}$, we have that 
\begin{align}
	\biggnorm{\sum_{(a,b) \in \Gamma} A_a^0 \otimes \left( \lambda R(a|0) + (1-\lambda) R(b|1) \right) \otimes B_b^1 } \leq \max_{a,b \in \Gamma} \biggnorm{ \lambda R(a|0) + (1-\lambda)R(b|1) }.
\end{align}
It follows from equation~\eqref{eq:monogamy-unentangled-value} that 
\begin{align}
	\omega(G) = \max_{a,b \in \Gamma} \biggnorm{ \lambda R(a|0) + (1-\lambda) R(b|1) }.
\end{align}
Therefore $\omega^*(G) \leq \omega(G)$.

\end{proof}

\section{Parallel repetition of monogamy-of-entanglement games} \label{sec:parallel-rep-moe-games}

For an integer $r \geq 1$ and some monogamy-of-entanglement game, $G$, the \index{parallel repetition}{\emph{$r$-fold parallel repetition}} of a monogamy-of-entanglement game is when Alice and Bob play $r$ copies of $G$, denoted as $G^r$, wherein the referee gives the players $r$ independent and identically distributed pairs of questions simultaneously and expects a response from Alice and Bob for each instance. The referee accepts if and only if all of the $r$ responses satisfy the criteria for the initial game, and rejects otherwise. The parallel repetition of a monogamy-of-entanglement game is depicted in Figure~\ref{fig:parallel-rep-monogamy-game}. 

Define the complex Euclidean spaces $\R_1, \ldots, \R_r$ and define alphabets 
\begin{align}
	\Sigma = \Sigma_1 \times \cdots \times \Sigma_r \quad \textnormal{and} \quad \Gamma = \Gamma_1 \times \cdots \times \Gamma_r
\end{align}
such that $x_1 \in \Sigma_1, \ldots, x_r \in \Sigma_r$ are selected from $\Sigma$ according to 
\begin{align}
	\pi^k : \Sigma_1 \times \cdots \times \Sigma_k \rightarrow [0,1]
\end{align}
where $\pi^k(x_1, \ldots, x_r) = \pi(x_1) \cdots \pi(x_r)$. Then the $r$-fold parallel repetition of $G$ starts off with the referee accepting $r$ registers $\reg{R}_1, \ldots, \reg{R}_r$ from Alice and Bob where the contents of the registers correspond to the state 
\begin{align}
	\sigma \in \Density(\U \otimes \R_1 \otimes \cdots \otimes \R_r \otimes \V),
\end{align}
and selecting $r$ questions $x_1 \in \Sigma, \ldots, x_r \in \Sigma$ according to $\pi$. The referee then sends $x_1, \ldots, x_r$ to Alice and Bob. The players return $r$ answers $a_1 \in \Gamma_1, \ldots, a_r \in \Gamma_r$ for each question. The referee then performs a measurement from the set 
\begin{align} \label{eq:monog-parallel-rep-ref-ops}
	\left \{ R(a_1, \ldots, a_r|x_1, \ldots, x_r) = R(a_1|x_1) \otimes \cdots \otimes R(a_r|x_r) : a_i \in \Gamma, \ x_i \in \Sigma \right \}.
\end{align}
Alice and Bob win the parallel repetition of $G$ if and only if their answers win each of the $r$ instances of $G^r$. That is, for a monogamy-of-entanglement game, Alice and Bob win if and only if the outcomes of their measurements in every $r$ instance of the game matches with the referee's measurement outcome for every $r$ games. 

\begin{figure}[!htpb] 
	\begin{center}
		\includegraphics[scale=0.9]{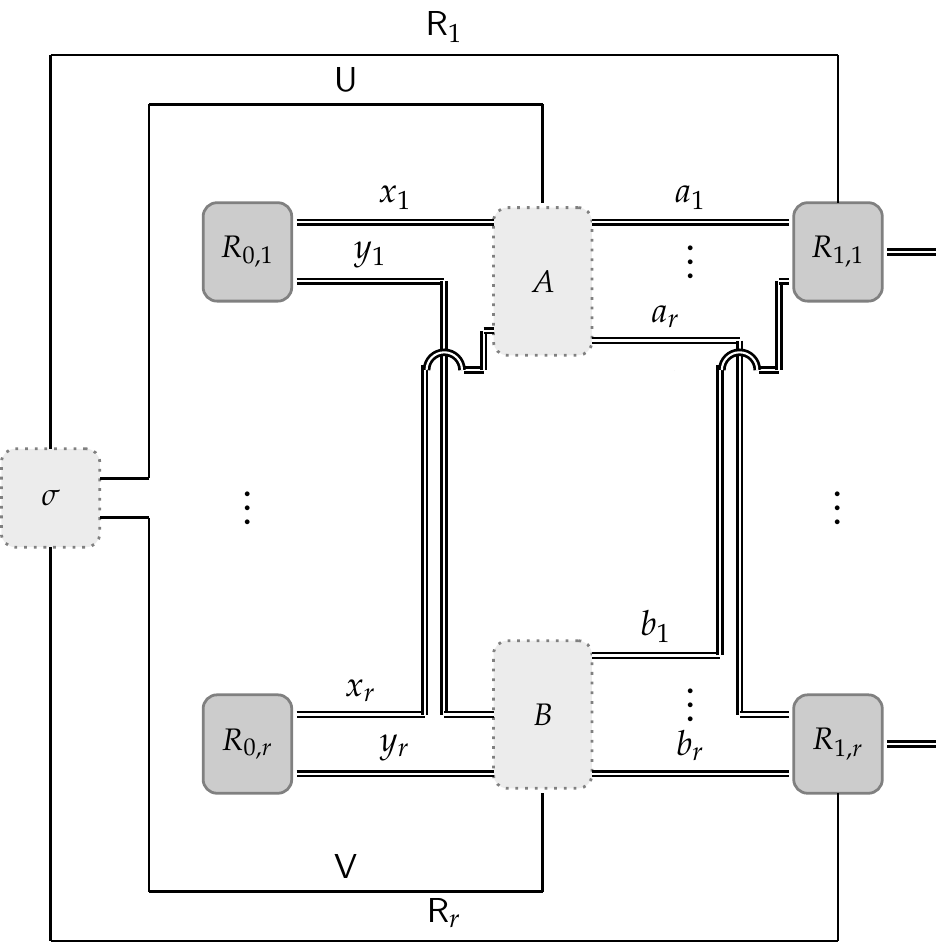}
	\end{center}
		\caption[Parallel repetition of a monogamy-of-entanglement game.]{Parallel repetition of a monogamy-of-entanglement game. Alice and Bob prepare registers $(\reg{R}_1, \ldots, \reg{R}_r)$ and send to the referee. The referee then asks questions $x_1, \ldots, x_r$ to Alice and $y_1, \ldots, y_r$ to Bob. Alice and Bob respond to each question with answers $a_1, \ldots, a_r$ and $b_1, \ldots, b_r$. Once the referee receives all the answers, it performs a measurement. Alice and Bob win the parallel repetition of the monogamy-of-entanglement game if and only if all their answers match the referee's measurement outcome for every $r$ instance.}
		\label{fig:parallel-rep-monogamy-game}
\end{figure}

One may ask how $\omega(G^r)$ depends on $\omega(G)$ and $r$. It is evident that $\omega(G^r) \geq \omega(G)^r$ since Alice and Bob can simply perform the same strategy in each instance. For any game that has the property $\omega(G) = 1$, it holds that $\omega(G^r) = 1$ for any $r$. One may also wish to ask the question of how $\omega(G^r)$ scales in the event that $\omega(G) < 1$. First note that $\omega(G^r) \leq \omega(G)$. This can be seen since, in order for the players to win all instances of the game, they must win the original game, $G$. Note also that $\omega(G)^r \leq \omega(G^r)$ and $\omega(G)^r \leq \omega(G)$. This holds since the players can simply play each game independently with the optimal strategy for the original game. We, therefore, have the following inequality relationship for the parallel repetition of $G$\
\begin{align}
	\omega(G)^r \leq \omega(G^r) \leq \omega(G). 
\end{align} 
It may be tempting to conclude that $\omega(G^r) = \omega(G)^r$ for all games, however this was surprisingly disproven~\cite{Fortnow1990, Feige1991, Verbitsky1996, Feige2002}. Specifically in~\cite{Fortnow1990}, Fortnow introduced a game $G$ for which $\omega(G^2) > \omega(G)^2$. This result was later improved by Feige~\cite{Feige1991}, by exhibiting an example of a game where $\omega(G^2) = \omega(G)$ with $\omega(G) < 1$.  

We say that a game $G$ exhibits the property of \index{strong parallel repetition}{\emph{strong parallel repetition}} if the value of the game raised to the $r$ power is equal to the value of running the game $r$ times. For instance, a monogamy-of-entanglement game, $G$, where the players use a standard quantum strategy satisfies strong parallel repetition if and only if
\begin{align}
	\omega^*(G^r) = \omega^*(G)^r. 
\end{align}
Strong parallel repetition has been also referred to as perfect parallel repetition elsewhere in the literature as in~\cite{Cleve2008}. 

It is a fact proved in~\cite{Tomamichel2013} that the BB84 game, $G_{\BB84}$, exhibits the property of strong parallel repetition,
\begin{align}
	\omega^*(G_{\BB84}^r) = \omega^*(G_{\BB84})^r = \left( \cos^2(\pi/8) \right)^r,
\end{align}
where $r$ is the number of rounds of repetition performed. A natural question for the general class of monogamy-of-entanglement games might be whether this behavior holds for any monogamy-of-entanglement game. In Section~\ref{sec:monog-strong-parallel-rep}, we prove that for any monogamy-of-entanglement game $G = (\pi,R)$ where the set of measurements belonging to the referee, $R$, are projective, the distribution $\pi$ is uniform, the size of the question set is $\abs{\Sigma} = 2$, and the size of the answer set is $\abs{\Gamma} = k$ for some integer $k \geq 1$, then strong parallel repetition holds. 

This result of strong parallel repetition holds for the case when Alice and Bob use either an unentangled or a standard quantum strategy since we know from Section~\ref{sec:monog-classical-quantum} that 
\begin{align}
	\omega(G) = \omega^*(G)
\end{align}
for any monogamy-of-entanglement game $G$, with $\abs{\Sigma} = 2$ and $\abs{\Gamma} = k$ for some integer $k \geq 1$. Specifically, the result that is shown in Section~\ref{sec:monog-strong-parallel-rep} is that 
\begin{align}
	\omega(G^r) = \omega(G)^r \quad \textnormal{and} \quad \omega^*(G^r) = \omega^*(G)^r,
\end{align}
where $r$ is the number of repetitions and $G = (\pi,R)$ is a monogamy-of-entanglement where $\abs{\Sigma} = 2$, $\abs{\Gamma} = k$, $\pi$ is uniform, and $R$ is a collection of projective measurement operators. 

While strong parallel repetition holds for a specific class of monogamy-of-entanglement games when Alice and Bob use either an unentangled or standard quantum strategy, we can ask a similar question when Alice and Bob use a non-signaling strategy instead. As we shall see in Section~\ref{sec:monog-non-signaling}, it turns out that strong parallel repetition does \emph{not} hold in the non-signaling scenario. We shall illustrate this by showing a counter-example to strong parallel repetition by using a non-signaling version of the BB84 monogamy-of-entanglement game and showing that
\begin{align}
	\omega_{\ns}(G_{\BB84}^r) \not= \omega_{\ns}(G_{\BB84})^r.
\end{align}
for $r = 2$. 

\subsection{Strong parallel repetition for certain monogamy-of-entanglement games with two questions} \label{sec:monog-strong-parallel-rep}

We begin this section by recalling a theorem from~\cite{Tomamichel2013}.

\begin{theorem}[Tomamichel, Fehr, Kaniewski, and Wehner (Theorem 4 of~\cite{Tomamichel2013})] \label{thm:parallel-rep-monogamy-bound}
	Let $G = (\pi,R)$ be a monogamy-of-entanglement game for which $\pi$ is uniform over $\Sigma$, define 
	\begin{align} \label{eq:monog-c-g}
		c(G) = \max_{\substack{x,y \in \Sigma \\ x \not= y}} \max_{a,b \in \Gamma} \biggnorm{\sqrt{R(a|x)} \sqrt{R(b|y)}}^2, 
	\end{align}
	and let $G^r$ denote the game played $r$ times in parallel. It holds that 
	\begin{align} \label{eq:monog-tfkw-bound}
		\omega^*(G^r) \leq \left( \frac{1}{\abs{\Sigma}} + \frac{\abs{\Sigma} - 1}{\abs{\Sigma}} \sqrt{c(G)} \right)^r.
	\end{align}
\end{theorem}

Equation~\eqref{eq:monog-c-g} may be referred to as the maximal overlap of the referee's measurements. As was observed in~\cite{Tomamichel2013}, this quantity satisfies
\begin{align}
	\frac{1}{\abs{\Gamma}} \leq c(G) \leq 1 \quad \textnormal{and} \quad c(G^r) = c(G)^r. 
\end{align}
For any monogamy-of-entanglement game, $G$, Theorem~\ref{thm:parallel-rep-monogamy-bound} provides an upper bound on the standard quantum value achieved when running $G$ for $r$ times in parallel as given by equation~\eqref{eq:monog-tfkw-bound}. In this section, we shall show that for $\abs{\Sigma} = 2$ that the bound from equation~\eqref{eq:monog-tfkw-bound} is indeed tight. 

\begin{theorem} \label{thm:monog-parallel-rep-2}
	Let $G = (\pi,R)$ be a monogamy-of-entanglement game for which $\pi$ is uniform over $\Sigma$ with $\abs{\Sigma} = 2$. It holds that 
	\begin{align}
	\omega^*(G^r) = \left(\frac{1}{2} + \frac{1}{2} \sqrt{c(G)} \right)^r.
\end{align}
\end{theorem}

In order to prove theorem~\ref{thm:monog-parallel-rep-2}, we first require the following proposition. 
\begin{prop} \label{prop:monogamy-unentangled}
	Let $G = (\pi,R)$ be a monogamy-of-entanglement game for which $\Sigma = \{0,1\}$, $\pi$ is uniform over $\Sigma$, and $R(a|x)$ is a projection operator for each $x \in \Sigma$ and $a \in \Gamma$. It holds that 
	\begin{align}
		\omega(G) = \frac{1}{2} + \frac{1}{2} \max_{a,b \in \Gamma} \biggnorm{R(a|0)R(b|1)}.
	\end{align}
\end{prop}
Proving this proposition requires the use of the following lemma.
\begin{lemma} \label{lem:monogamy-projections}
	Let $\Pi_0$ and $\Pi_1$ be nonzero projection operators on $\complex^r$. It holds that 
	\begin{align}
		\bignorm{\Pi_0 + \Pi_1} = 1 + \bignorm{\Pi_0 \Pi_1}.
	\end{align}
\end{lemma}

\begin{proof}
	For every choice of unit vectors $u_0, u_1 \in \complex^r$, one has the formula
	\begin{align}
		\bignorm{ u_0 u_0^* + u_1 u_1^* } = 1 + \abs{\ip{u_0}{u_1}},
	\end{align}
	which follows from the observation that the Hermitian operator $u_0 u_0^* + u_1 u_1^*$ has (at most) two nonzero eigenvalues $1 \pm \abs{\ip{u_0}{u_1}}$. Letting $\S$, $\S_0$, and $\S_1$ denote the unit spheres in the spaces $\complex^r$, $\im(\Pi_0)$, and $\im(\Pi_1)$, respectively, one has 
	\begin{equation}
		\begin{aligned}
			\bignorm{\Pi_0 + \Pi_1} &= \max \left \{ v^* \left( \Pi_0 + \Pi_1 \right) v : v \in \S \right \} \\
			&= \max \left \{ \bignorm{\Pi_0 v}^2 + \bignorm{\Pi_1 v}^2 : v \in \S \right \} \\
			&= \max \left \{ \abs{\ip{u_0}{v}}^2 + \abs{\ip{u_1}{v}}^2 : v \in \S, \ u_0 \in \S_0, \ u_1 \in \S_1 \right \} \\
			&= \max \left \{ v^* \left( u_0 u_0^* + u_1 u_1^* \right)v : v \in \S, \ u_0 \in \S_0, \ u_1 \in \S_1 \right \} \\
			&= \max \left \{ \bignorm{u_0 u_0^* + u_1 u_1^*} : u_0 \in \S_0, \ u_1 \in \S_1 \right \} \\
			&= \max \left \{ 1 + \abs{\ip{u_0}{u_1}} : u_0 \in \S_0, \ u_1 \in \S_1 \right \} \\
			&= 1 + \bignorm{\Pi_0 \Pi_1},
		\end{aligned}
	\end{equation}
which proves the lemma. 
\end{proof}

\begin{proof}[Proof of Proposition~\ref{prop:monogamy-unentangled}]
	From equation~\eqref{eq:enlg-unentangled-value}, we have that the unentangled value of the game $G$ is given by 
	\begin{align} \label{eq:monog-parallel-rep-proj}
		\omega(G) = \max_{a,b \in \Gamma} \biggnorm{\frac{1}{2} R(a|0) + \frac{1}{2} R(b|1) } = \frac{1}{2} \max_{a,b \in \Gamma} \biggnorm{R(a|0) + R(b|1)}.
	\end{align}
	It follows from Lemma~\ref{lem:monogamy-projections} that $\norm{R(a|0) + R(b|1)} = 1 + \norm{R(a|0) R(b|1)}$ which allows us to write equation~\eqref{eq:monog-parallel-rep-proj} as 
	\begin{align}
		\omega(G) = \frac{1}{2} + \frac{1}{2} \max_{a,b \in \Gamma} \biggnorm{R(a|0)R(b|1)},
	\end{align}
	proving the proposition. 
\end{proof}

\begin{proof}[Proof of Theorem~\ref{thm:monog-parallel-rep-2}]
	The upper bound
		\begin{align}
			\omega^*(G^r) \leq \left(\frac{1}{2} + \frac{1}{2} \sqrt{c(G)} \right)^r,
	\end{align}
	follows from Theorem~\ref{thm:parallel-rep-monogamy-bound}, initially shown in~\cite{Tomamichel2013}. 
	Showing the other direction
		\begin{align}
	\omega^*(G^r) \geq \left(\frac{1}{2} + \frac{1}{2} \sqrt{c(G)} \right)^r,
\end{align}
 follows from Proposition~\ref{prop:monogamy-unentangled} and Lemma~\ref{lem:monogamy-projections}. The reason that this direction holds for any number of repetitions $r$ is that Alice and Bob can simply play an optimal strategy for each $r$ games, $r$ times in parallel. This implies that 
 \begin{align}
 	\omega^*(G^r) \geq \omega(G^r) \geq \left( \frac{1}{2} + \frac{1}{2} \max_{a,b \in \Gamma} \biggnorm{R(a|0)R(b|1)} \right)^r = \left( \frac{1}{2} + \frac{1}{2} \sqrt{c(G)} \right)^r,
 \end{align}
 which matches the upper bound from Theorem~\ref{thm:parallel-rep-monogamy-bound}.
\end{proof}

\subsection{No strong parallel repetition for monogamy-of-entanglement games with non-signaling provers} \label{sec:monog-non-signaling}

\begin{claim}[No strong parallel repetition for non-signaling provers] \label{thm:no-spr-non-signaling}
	There exists a monogamy-of-entanglement game, $G$, such that
	\begin{align} \label{eq:no-spr-non-signaling}
		\omega_{\ns}(G^2) \not= \omega_{\ns}(G)^2.
	\end{align}		
\end{claim}

\begin{proof}[Proof of Claim~\ref{thm:no-spr-non-signaling}]
	We shall verify equation~\eqref{eq:no-spr-non-signaling} numerically using the convex optimization software CVX~\cite{Grant2008a} in addition to the software listing~\ref{code:ns-counter-example-monogamy-game} in Appendix~\ref{chap:AppendixA}. The explicit monogamy-of-entanglement game that we shall use to verify this claim is $G_{\BB84}$, the BB84 game as mentioned in Section~\ref{sec:the-bb84-monogamy-of-entanglement-game}. It may be checked by running the software listing~\ref{code:ns-counter-example-monogamy-game} that 
	\begin{align}
		\omega_{\ns}(G_{\BB84}^2) \approx 0.73826. 
	\end{align}
However, we may also verify that a single repetition of $G_{\BB84}$ is $\cos^2(\pi/8)$, that is
	\begin{align}
		\omega_{\ns}(G_{\BB84}) = \cos^2(\pi/8).
	\end{align}
From this, it is clear that 
\begin{align}
	\omega_{\ns}(G_{\BB84}^2) \not= \omega_{\ns}(G_{\BB84})^2 = \cos^4(\pi/8),
\end{align}
which concludes the proof. 
\end{proof}

\section[Upper and lower bounds on monogamy-of-entanglement games]{Upper and lower bounds on \\ monogamy-of-entanglement games} \label{sec:upper-and-lower-bounds-moe-games}

In this section, we apply the upper and lower bound techniques for extended nonlocal games from Chapter~\ref{chap:extended_npa_hierarchy}, and apply them to monogamy-of-entanglement games. In doing so, we are able to verify the existence of a monogamy-of-entanglement game where Alice and Bob perform better if they use a standard quantum strategy as opposed to an unentangled one.

\subsection[A monogamy-of-entanglement game with quantum advantage]{A monogamy-of-entanglement game with \\ quantum advantage} \label{sec:MUB-4-3}

\begin{example}[A monogamy-of-entanglement game with quantum advantage]\label{ex:mub-4-3-monogamy-game}
  Let $\zeta = \exp(\frac{2 \pi i}{3})$ and consider the following four mutually unbiased bases:
  \begin{equation}\label{eq:MUB43}
    \begin{aligned}
      \B_0 &= \left\{ e_0,\: e_1,\: e_2 \right\}, \\
      \B_1 &= \left\{ \frac{e_0 + e_1 + e_2}{\sqrt{3}},\:
      \frac{e_0 + \zeta^2 e_1 + \zeta e_2}{\sqrt{3}},\:
      \frac{e_0 + \zeta e_1 + \zeta^2 e_2}{\sqrt{3}} \right\}, \\
      \B_2 &= \left\{ \frac{e_0 + e_1 + \zeta e_2}{\sqrt{3}},\:
      \frac{e_0 + \zeta^2 e_1 + \zeta^2 e_2}{\sqrt{3}},\:
      \frac{e_0 + \zeta e_1 + e_2}{\sqrt{3}} \right\}, \\
      \B_3 &= \left\{ \frac{e_0 + e_1 + \zeta^2 e_2}{\sqrt{3}},\:
      \frac{e_0 + \zeta^2 e_1 + e_2}{\sqrt{3}},\:
      \frac{e_0 + \zeta e_1 + \zeta e_2}{\sqrt{3}} \right\}.
    \end{aligned}
  \end{equation} 
 	Define a monogamy-of-entanglement game $G = (\pi,R)$ so that  
 	\begin{align}
 		\pi(0) = \pi(1) = \pi(2) = \pi(3) = \frac{1}{4}
 	\end{align}
 	and $R$ is such that 
 	\begin{align}
 		\{ R(0|x), R(1|x), R(2|x) \}
 	\end{align} 
 	represents a measurement with respect to the basis $\B_x$, for each $x \in \{0,1,2,3\}$. In order to observe that $\omega(G) < \omega^*(G)$, first consider the following unentangled strategy. Alice and Bob prepare the state 
 	\begin{align}
 		u = \left( 1 - i \sqrt{\frac{3}{4}} \right) e_0 + \left( 1 + i \sqrt{\frac{3}{5}} \right) e_1 + \left( 1 + \frac{3}{\sqrt{5}} \right) e_2,
 	\end{align}
 	and sends it to the referee. In the event that $x = 0$ or $x = 1$, Alice and Bob respond with $a = b = 2$. If instead $x = 2$ or $x = 3$, Alice and Bob respond with $a = b = 0$. The value of the game in this case is given by 
 	\begin{align} \label{eq:MUB43-unentangled}
 		\omega(G) = \frac{1}{4}\biggip{R(2|0) + R(2|1) + R(0|2) + R(0|3)}{\rho} = \frac{3 + \sqrt{5}}{8} \approx 0.6545,
 	\end{align}
 	where $\rho = u u^* \in \Density(\R \otimes \A \otimes \B)$. An exhaustive search over all unentangled strategies reveals that equation~\eqref{eq:MUB43-unentangled} is optimal. In contrast, a computer search over quantum strategies using the lower bound techniques from Sectionf~\ref{sec:lower-bound-extended-nonlocal-games} has revealed that 
 	\begin{align} \label{eq:lb-MUB-4-3}
 		\omega^*(G) \geq 0.660986,
 	\end{align}
 	which is strictly larger than the unentangled value of this game. This strategy is available for download from the software repository~\cite{Johnston2015b} and is also provided as a software example in Appendix~\ref{chap:AppendixA}. 
	It is uncertain what the optimal standard quantum strategy is for this game, but the value of such a strategy is bounded as follows  
	\begin{align}
		2/3 \geq \omega^*(G) \geq 0.660986. 
	\end{align}

\end{example}

\subsection{Synopsis of monogamy-of-entanglement games}

The following table gives an overview of what is currently known about the class of monogamy-of-entanglement games and summarizes some of the main contributions of this chapter. 

\begin{table}[!htpb]
\centering
\def\arraystretch{1.5}
\begin{tabular}{|c|c|c|c|c|}
  \hline
    Inputs ($\abs{\Sigma}$) & Outputs ($\abs{\Gamma}$) & $\omega^*(G) = \omega(G)$ & $\omega^*(G^r) = \omega^*(G)^r$ & $\omega_{\ns}(G^r) = \omega_{\ns}(G)^r$  \\
  \hline \hline
  $2$ & $\geq 1$ & yes & yes\footnotemark & no \\
  \hline
  $3$ & $\geq 1$ & ? & ? & no \\
  \hline
  $4$ & $3$ & no & ? & no \\
  \hline
\end{tabular}
\caption[Results on monogamy-of-entanglement games.]{A table of known results for monogamy-of-entanglement games. The first and second column refer to the number of inputs and outputs for such a game. The third column states whether or not the unentangled and standard quantum values are equal, and the last two columns state whether or not strong parallel repetition holds with either quantum players or non-signaling players respectively.}
\label{table:results-monogamy-of-entanglement-games}
\end{table}

\footnotetext{So long as the measurements used by the referee are projective and the probability distribution, $\pi$, from which the questions are asked is uniform.}

\chapter{Conclusions and open problems}
\label{chap:conclusions}

In this thesis, we have laid the foundation for the extended nonlocal game model, a superset of the nonlocal game model where the referee also holds a quantum system. 

In Chapter~\ref{chap:extended_nonlocal_games} we defined and analyzed the analogous types of strategies and corresponding game values (standard quantum, unentangled, commuting measurement, and non-signaling) that the players, Alice and Bob, can make use of in such a game. 

In Chapter~\ref{chap:infinite_entanglement}, we took a deeper look at the extended nonlocal game model and showed that there exists an example of an extended nonlocal game where if the dimension of Alice and Bob's shared quantum system is finite, then the standard quantum value will be strictly less than $1$. However, taking the limit as the dimension tends to infinity, the standard quantum value approaches $1$. We saw how this result implies something non-trivial about tripartite steering inequalities, specifically that there exists a tripartite steering inequality for which an infinite-dimensional state is required in order to maximally violate the inequality. 

In Chapter~\ref{chap:extended_npa_hierarchy}, we provided a technique to place upper bounds on the standard quantum value of an extended nonlocal game that generalizes the QC hierarchy, which we referred to as the extended QC hierarchy. We have shown that the hierarchy enjoys many of the same useful properties that the original QC hierarchy does, specifically, convergence to the set of commuting measurement assemblages. We also adapted the techniques of Liang and Doherty~\cite{Liang2007} to place lower bounds on the standard quantum value of extended nonlocal games. Furthermore, we have also presented software that calculates lower and upper bounds using these techniques of certain special classes of extended nonlocal games.

In Chapter~\ref{chap:monogamy_games}, we took these tools and analyzed the class of monogamy-of-entanglement games, a class of games that were initially studied in the context of position-based cryptography~\cite{Tomamichel2013}. We proved a number of properties that these games have including how they behave under parallel repetition, how entanglement may help in Alice and Bob's strategies, etc. 

A number of questions regarding the class of monogamy-of-entanglement games remain open. Specifically, 

\begin{question}
	{\bf Other examples of monogamy-of-entanglement games where $\omega(G) < \omega^*(G)$}.
\end{question}

	The complete landscape of how the quantum and classical values compare for different instances of monogamy-of-entanglement games is something to be explored. We only know of a small number of isolated examples where $\omega(G) < \omega^*(G)$ for a monogamy-of-entanglement game, $G$. 

	In Section~\ref{sec:MUB-4-3} a set of $\abs{\Sigma} = 4$ mutually unbiased bases in $\abs{\Gamma} = 3$ dimensions allow Alice and Bob to perform better if they adopt a standard quantum strategy instead of an unentangled strategy. This is the smallest example of a monogamy-of-entanglement game that was found having this property. Is there an example having fewer questions or fewer answers? This example would have to have at least three questions, since we know that for $\abs{\Sigma} = 2$, that the unentangled and standard quantum values agree for any number of outputs as shown in Section~\ref{sec:monog-classical-quantum}. Numerical results indicate that the monogamy-of-entanglement game consisting of $\abs{\Sigma} = 3$ where the referee's measurements are defined in terms of mutually unbiased bases gives no such separation. Is it possible that another monogamy-of-entanglement game with $\abs{\Sigma} = 3$ questions exists where such a separation between unentangled and standard quantum values exists? 

	One brute force method that can be used to check if there exists a monogamy-of-entanglement game for $\abs{\Sigma} = 3$ where a standard quantum strategy will outperform an unentangled strategy is to run a computer search over randomly generated instances of such monogamy-of-entanglement games. The software provided in the Appendix of this thesis~\ref{chap:AppendixA} as well as hosted on the software repository~\cite{Russo2015a} provides a suite of tools that give upper and lower bounds on the quantum value (as described in Chapter~\ref{chap:extended_npa_hierarchy}) as well as tools for calculating the unentangled value of any monogamy-of-entanglement game. One approach would be to randomly generate monogamy-of-entanglement games where $\abs{\Sigma} = 3$ and $\abs{\Gamma} \geq 2$, and see if any example of such games yield $\omega(G) < \omega^*(G)$. This approach does not seem particularly promising, as if such a game were to exist with this property, it most likely has a very specific structure that would be difficult to capture by random generation.
	
	On a related note, under what conditions does a monogamy-of-entanglement game based on mutually unbiased bases admit a standard quantum over unentangled strategy advantage? Numerically, it may be checked that a monogamy-of-entanglement game consisting $\abs{\Sigma} = 5$ and $\abs{\Gamma} = 4$ also yields a standard quantum advantage over any unentangled strategy. Does this behavior persist for any monogamy-of-entanglement game defined by mutually unbiased bases as long as the number of inputs is at least $\abs{\Sigma} = 4$, and the number of outputs is at least $\abs{\Gamma} = 3$? Furthermore, do there exist other monogamy-of-entanglement games where $\abs{\Sigma} \geq 4$ and $\abs{\Gamma} \geq 3$ such that $\omega(G) < \omega^*(G)$? Just as a computer search can be constructed where $\abs{\Sigma} = 3$ and $\abs{\Gamma} \geq 2$, one may also formulate a search that checks for larger instances as well. 

\begin{question}
	Parallel repetition for monogamy-of-entanglement games?
\end{question} 

It was shown in Section~\ref{sec:monog-strong-parallel-rep} (Theorem~\ref{thm:parallel-rep-monogamy-bound}) that for any monogamy-of-entanglement game defined in terms of projective measurements for the referee where $\abs{\Sigma} = 2$ and $\abs{\Gamma} \geq k$ for some integer $k \geq 1$ that strong parallel repetition holds. Would it be possible to extend from projective measurements to non-projective measurements, such as POVMs? After simulating approximately $10^8$ random instances of monogamy-of-entanglement games with $\abs{\Sigma} = 2$ defined in terms of POVMs, all games were found to obey strong parallel repetition for $r = 2$ rounds of repetition.   

Furthermore, the claim that strong parallel repetition holds for monogamy-of-entanglement games where the measurements of the referee are projective and $\abs{\Sigma} = 2$ assumes that the questions that the referee asks are selected uniformly at random. Is it possible that the strong parallel repetition property will continue to hold despite the distribution of questions? If indeed it does hold under nonuniform distributions, the bound from equation~\eqref{eq:monog-tfkw-bound} from Theorem~\ref{thm:parallel-rep-monogamy-bound} will most likely be in a more complicated form. Ultimately, the overall goal for parallel repetition of monogamy-of-entanglement games would be to either prove or disprove strong parallel repetition for the entire class of such games. 

There also exists other questions and directions for further research. 

\begin{question}
	Other examples of using extended nonlocal games to study tripartite steering. 
\end{question}

As mentioned in Chapter~\ref{chap:infinite_entanglement}, we were able to prove a non-trivial statement about a certain type of tripartite steering using the extended nonlocal game model. Given the connection between extended nonlocal games and tripartite steering, are there other possible questions we can answer that become more apparent using the extended nonlocal game model? 

\begin{question}
	Does there exist a nonlocal game $G$ such that $\omega^*(G) = 1$ and that $\omega_N^*(G) < 1$ for every positive integer $N$?
\end{question}

As mentioned in Chapter~\ref{chap:infinite_entanglement}, it is known that nonlocal games with quantum questions and quantum answers do satisfy the above property~\cite{Leung2013}. However, it is unknown for nonlocal games with classical questions and classical answers. This question is most likely difficult to solve.


\bibliographystyle{alphaurl}
\cleardoublepage 
\phantomsection  
\renewcommand*{\bibname}{References}

\addcontentsline{toc}{chapter}{\textbf{References}}

\bibliography{refs_jab}


\appendix

\chapter*{APPENDICES}
\addcontentsline{toc}{chapter}{APPENDICES}
\chapter{Software}
\label{chap:AppendixA}

\section*{Setup}
\subsection*{Requirements}

\begin{itemize}
    \item MATLAB,
    \item CVX  $\geq$ 2.1 \cite{Grant2008a},
    \item QETLAB $\geq$ 0.8 \cite{Johnston2015}.
\end{itemize}

\subsection*{List of functions}

\begin{itemize}
    
    \item \texttt{MonogamyGameValueUB} (by N. Johnston) ---
        Given a monogamy-of-entanglement game, $G$, the function calculates an upper bound on the quantum value of $G$;
        
    \item \texttt{MonogamyGameValueLB} ---
        Given a monogamy-of-entanglement game, $G$, the function calculates an lower bound on the quantum value of $G$;
        
    \item \texttt{MUB} (by N. Johnston) --- generates a set of mutually unbiased bases for a given dimension;
\end{itemize}

\sloppy
\definecolor{lightgray}{gray}{0.5}
\setlength{\parindent}{0pt}

\section{Software Listings}

All of the following software listings in this Appendix are hosted on the Github repository found here~\cite{Russo2016c}. 

\subsection{The first level of the extended QC hierarchy for the BB84 extended nonlocal game} \label{code:first-level-qc-hierarchy-bb84}

\begin{verbatim}
e0 = [1;0]; e1 = [0;1]; ep = [1;1]/sqrt(2); em = [1;-1]/sqrt(2);
psi0_dm = e0*e0'; psi0_dmc = e1*e1';
psi1_dm = ep*ep'; psi1_dmc = em*em';

R00 = psi0_dm/2;
R01 = psi0_dmc/2; 
R10 = psi1_dm/2;
R11 = psi1_dmc/2;  

dim = 9;
A00_B00 = zeros(dim); A01_B01 = zeros(dim); 
A10_B10 = zeros(dim); A11_B11 = zeros(dim); 

% These are the relative positions of these 
% entries as indexed by strings in the matrix. 
A00_B00(2,6) = 1; A00_B00(6,2) = 1;
A01_B01(3,7) = 1; A01_B01(7,3) = 1;
A10_B10(4,8) = 1; A10_B10(8,4) = 1;
A11_B11(5,9) = 1; A11_B11(9,5) = 1;

A = 1/2*( kron(R00, A00_B00) + kron(R01, A01_B01) ) + ...
	1/2*( kron(R10, A10_B10) + kron(R11,A11_B11) );

cvx_begin sdp
	cvx_precision best
    %#ok<*VUNUS>    % suppress MATLAB warnings for equality checks in CVX
    %#ok<*EQEFF>    % suppress MATLAB warnings for inequality checks in CVX 
    
    % Admissible matrix
    variable M(2*dim,2*dim) hermitian
    
    % Sub-block matrices found in the admissible matrix
    variable M11(dim,dim)
    variable M12(dim,dim)
    
    variable M21(dim,dim)
    variable M22(dim,dim)
            
    M == [ M11 M12;
    	   M21 M22 ];
       
    maximize trace( A*M )      
    
    subject to 
 
    % Normalization condition:
    M11(1,1) + M22(1,1) == 1;
        
    for i = 1:dim
        for j = 1:dim
            % Ensure commutation relation holds
            %(i.e. [A,B] = 0)            
            M11(i,j) == M11(j,i);
            M12(i,j) == M12(j,i);
            M21(i,j) == M21(j,i);
            M22(i,j) == M22(j,i);
            
            % Enforce operators as projective measurements
            % (i.e. the square of the same operator is found in the top
            % column / row of the diagonal entry). 
            M11(i,i) == M11(1,i);
            M11(i,i) == M11(i,1);
                        
            M12(i,i) == M12(1,i);
            M12(i,i) == M12(i,1);

            M21(i,i) == M21(1,i);
            M21(i,i) == M21(i,1);

            M22(i,i) == M22(1,i);
            M22(i,i) == M22(i,1);
        end
    end  

	% Enforce that projective measurements sum to 1:
    for i = 1:dim
    	for j = 1:dim
    		if mod(i,2) == 0    			
    		M11(i,j) + M11(i+1,j) == M11(1,j);
    		M12(i,j) + M12(i+1,j) == M12(1,j);
			M21(i,j) + M21(i+1,j) == M21(1,j);
			M22(i,j) + M22(i+1,j) == M22(1,j);
    		end
    		if mod(j,2) == 0
    		M11(i,j) + M11(i,j+1) == M11(i,1);
    		M12(i,j) + M12(i,j+1) == M12(i,1);
    		M21(i,j) + M21(i,j+1) == M21(i,1);
    		M22(i,j) + M22(i,j+1) == M22(i,1);   		    		
	    	end
    	end
    end

    % Ensure that the matrix is PSD.
    M >= 0;
cvx_end
\end{verbatim}
\color{black}
\color{lightgray} 
\begin{verbatim}     
cvx_optval =

    0.8536
\end{verbatim}
\color{black}

\subsection{The first level of the extended QC hierarchy for the CHSH extended nonlocal game} \label{code:first-level-qc-hierarchy-chsh}

\begin{verbatim}
e0 = [1;0]; e1 = [0;1]; ep = [1;1]/sqrt(2); em = [1;-1]/sqrt(2);
psi0_dm = e0*e0'; psi0_dmc = e1*e1';
psi1_dm = ep*ep'; psi1_dmc = em*em';

R00 = psi0_dm/2;
R01 = psi0_dmc/2; 
R10 = psi1_dm/2;
R11 = psi1_dmc/2;  

dim = 9;
A00_B00 = zeros(dim); A01_B01 = zeros(dim); 
A10_B10 = zeros(dim); A11_B11 = zeros(dim); 

% These are the relative positions of these entries as
% indexed by strings in the matrix. 
A00_B00(2,6) = 1; A00_B00(6,2) = 1;
A01_B01(3,7) = 1; A01_B01(7,3) = 1;
A10_B10(4,8) = 1; A10_B10(8,4) = 1;
A11_B11(5,9) = 1; A11_B11(9,5) = 1;

A00_B00 = zeros(dim); A01_B01 = zeros(dim); 

A00_B10 = zeros(dim); A01_B11 = zeros(dim);

A10_B00 = zeros(dim); A11_B01 = zeros(dim); 

A10_B11 = zeros(dim); A11_B10 = zeros(dim); 

A00_B00(2,6) = 1; A00_B00(6,2) = 1;
A00_B10(2,8) = 1; A00_B10(8,2) = 1;

A01_B01(3,7) = 1; A01_B01(7,3) = 1;
A01_B11(3,9) = 1; A01_B11(9,3) = 1;

A10_B00(4,6) = 1; A10_B00(6,4) = 1;
A10_B11(4,9) = 1; A10_B11(9,4) = 1;

A11_B01(5,7) = 1; A11_B01(7,5) = 1;
A11_B10(5,8) = 1; A11_B10(8,5) = 1;

% CHSH ENLG
A = 1/4*(kron(R00, A00_B00) + kron(R01, A01_B01)) + ...
    1/4*(kron(R00, A00_B10) + kron(R01, A01_B11)) + ...
    1/4*(kron(R00, A10_B00) + kron(R01, A11_B01)) + ...
    1/4*(kron(R10, A10_B11) + kron(R11, A11_B10));

cvx_begin sdp
	cvx_precision best
    %#ok<*VUNUS>    % suppress MATLAB warnings for equality checks in CVX
    %#ok<*EQEFF>    % suppress MATLAB warnings for inequality checks in CVX 
    
    % Admissible matrix
    variable M(2*dim,2*dim) hermitian
    
    % Sub-block matrices found in the admissible matrix
    variable M11(dim,dim)
    variable M12(dim,dim)
    
    variable M21(dim,dim)
    variable M22(dim,dim)
            
    M == [ M11 M12;
    	   M21 M22 ];
       
    maximize trace( A*M )      
    
    subject to 
 
    % Normalization condition:
    M11(1,1) + M22(1,1) == 1;
        
    for i = 1:dim
        for j = 1:dim
            % Ensure commutation relation holds
            %(i.e. [A,B] = 0)            
            M11(i,j) == M11(j,i);
            M12(i,j) == M12(j,i);
            M21(i,j) == M21(j,i);
            M22(i,j) == M22(j,i);
            
            % Enforce operators as projective measurements
            % (i.e. the square of the same operator is found in the top
            % column / row of the diagonal entry). 
            M11(i,i) == M11(1,i);
            M11(i,i) == M11(i,1);
                        
            M12(i,i) == M12(1,i);
            M12(i,i) == M12(i,1);

            M21(i,i) == M21(1,i);
            M21(i,i) == M21(i,1);

            M22(i,i) == M22(1,i);
            M22(i,i) == M22(i,1);
        end
    end  

	% Enforce that projective measurements sum to 1:
    for i = 1:dim
    	for j = 1:dim
    		if mod(i,2) == 0    			
    		M11(i,j) + M11(i+1,j) == M11(1,j);
    		M12(i,j) + M12(i+1,j) == M12(1,j);
			M21(i,j) + M21(i+1,j) == M21(1,j);
			M22(i,j) + M22(i+1,j) == M22(1,j);
    		end
    		if mod(j,2) == 0
    		M11(i,j) + M11(i,j+1) == M11(i,1);
    		M12(i,j) + M12(i,j+1) == M12(i,1);
    		M21(i,j) + M21(i,j+1) == M21(i,1);
    		M22(i,j) + M22(i,j+1) == M22(i,1);   		    		
	    	end
    	end
    end

    % Ensure that the matrix is PSD.
    M >= 0;
cvx_end
\end{verbatim}
\color{black}
\color{lightgray} 
\begin{verbatim}     
cvx_optval =

    0.75783
\end{verbatim}
\color{black}

\subsection{The non-signaling value for the CHSH extended nonlocal game} \label{code:ns-val-chsh-enlg}

\begin{verbatim}
n = 1;
dim = 2^n;

e0 = [1;0];         e1 = [0;1];
ep = [1;1]/sqrt(2); em = [1;-1]/sqrt(2);
eip = (e0 + 1j*e1)/sqrt(2); eim = (e0 - 1j*e1)/sqrt(2);

psi0_dm = e0*e0'; psi0_dmc = e1*e1';
psi1_dm = ep*ep'; psi1_dmc = em*em';
psi2_dm = eip*eip'; psi2_dmc = eim*eim';

P = zeros(2,2,2,2);
%P(:,:,1,1) = (psi0_dm)/2; P(:,:,1,2) = (psi0_dmc)/2;
%P(:,:,2,1) = (psi1_dm)/2; P(:,:,2,2) = (psi1_dmc)/2;

P = zeros(2,2,2,2,2,2);
P(:,:,1,1,1,1) = psi0_dm/2;
P(:,:,1,1,2,2) = psi0_dmc/2;

P(:,:,1,2,1,1) = psi0_dm/2;
P(:,:,1,2,2,2) = psi0_dmc/2;

P(:,:,2,1,1,1) = psi0_dm/2;
P(:,:,2,1,2,2) = psi0_dmc/2;

P(:,:,2,2,1,2) = psi1_dm/2;
P(:,:,2,2,2,1) = psi1_dmc/2;

cvx_begin sdp 
    %#ok<*VUNUS>    % suppress MATLAB warnings for equality checks in CVX
    %#ok<*EQEFF>    % suppress MATLAB warnings for inequality checks in CVX
     
    variable rho(dim,dim,dim,dim,dim,dim) semidefinite
    variable sig(dim,dim,dim,dim) hermitian
    variable xi(dim,dim,dim,dim) hermitian
    variable tau(dim,dim) hermitian 
        
    % construct objective function
    obj_fun = 0;
    for x = 1:dim
        for y = 1:dim
            for a = 1:dim
                for b = 1:dim
                   obj_fun = obj_fun + ip( P(:,:,x,y,a,b), rho(:,:,x,y,a,b) );
               end
           end
       end
    end
    
    maximize  obj_fun
    
    subject to 
    
    rho_b_sum = sum(rho,6);
    for x = 1:dim
        for y = 1:dim
            for a = 1:dim
                rho_b_sum(:,:,x,y,a) == sig(:,:,x,a);
            end
        end
    end
    
    rho_a_sum = sum(rho,5);
    for x = 1:dim
        for y = 1:dim
            for b = 1:dim
                rho_a_sum(:,:,x,y,b) == xi(:,:,y,b);
            end
        end
    end
    
    sig_a_sum = sum(sig,4);
    xi_b_sum = sum(xi,4);
    for x = 1:dim
        sig_a_sum(:,:,x) == tau;
    end
    for y = 1:dim
        xi_b_sum(:,:,y) == tau;
    end
    
    trace(tau) == 1; 
    tau >= 0;          
                        
cvx_end
cvx_optval
\end{verbatim}
\color{lightgray} 
\begin{verbatim}     
cvx_optval =

    0.75
\end{verbatim}
\color{black}

\subsection{Implementation of the see-saw method for computing lower bounds on the BB84 extended nonlocal game} \label{code:bb84-enlg-lower-bound}

\begin{verbatim}
e0 = [1;0];         e1 = [0;1];
ep = [1;1]/sqrt(2); em = [1;-1]/sqrt(2);

psi0_dm = e0*e0'; psi0_dmc = e1*e1';
psi1_dm = ep*ep'; psi1_dmc = em*em';

lvl = 1;
reps = 1;    
j_max = 4;

xdim = 2;
ydim = 2;

num_inputs = 2;
num_outputs = 2;

I = eye(xdim,ydim);

R = zeros(2,2,2,2,2,2);
R(:,:,1,1,1,1) = psi0_dm/2;
R(:,:,1,1,2,2) = psi0_dmc/2;
R(:,:,2,2,1,1) = psi1_dm/2;
R(:,:,2,2,2,2) = psi1_dmc/2; 

best = 0;
    
for k = 1:j_max
    k
    
    % Generate random bases from the orthogonal colums of randomly 
    % generated unitary matrices.     
    B = zeros(xdim,ydim,num_inputs,num_outputs);
    for y = 1:num_inputs
        U = RandomUnitary(num_outputs);
        for b = 1:num_outputs
            B(:,:,y,b) = U(:,b)*U(:,b)';
        end
    end  


    % Run the actual alternating projection algorithm between
    % the two SDPs. 
    it_diff = 1;
    prev_win = -1;
    while it_diff > 10^-6
        % Optimize over Alice's measurement operators while
        % fixing Bob's. If this is the first iteration, then the 
        % previously randomly generated operators in the outer loop are
        % Bob's. Otherwise, Bob's operators come from running the next
        % SDP.
        cvx_begin sdp quiet                
            variable rho(xdim^(2*reps),ydim^(2*reps),...
            num_inputs,num_outputs) hermitian
            variable tau(xdim^(2*reps),ydim^(2*reps)) hermitian

            win = 0;
            for x = 1:num_inputs
                for y = 1:num_inputs
                    for a = 1:num_outputs
                        for b = 1:num_outputs                               
                            win = win + ...                            
                           	trace( (kron(R(:,:,x,y,a,b), ...
                           	B(:,:,y,b)))' * rho(:,:,x,a) );                                
                        end
                    end
                end
            end
            
            maximize real(win)

            subject to 
                
                % Sum over "a" for all "x". 
                rho_a_sum = sum(rho,4);
                for x = 1:num_inputs
                    rho_a_sum(:,:,x) == tau;
                end
                                                                           
                % Enforce that tau is a density operator.
                trace(tau) == 1;
                tau >= 0;
                                   
                rho >= 0; 

        cvx_end            
        win = real(win);
                   
        % Now, optimize over Bob's measurement operators and fix 
        % Alice's operators as those coming from the previous SDP.
        cvx_begin sdp quiet
                        
            variable B(xdim,ydim,num_inputs,num_outputs) hermitian 
                                                         
            win = 0;
            for x = 1:num_inputs
                for y = 1:num_inputs
                    for a = 1:num_outputs
                        for b = 1:num_outputs                                                               
                            win = win + ...
                             trace( (kron(R(:,:,x,y,a,b), ...
                             B(:,:,y,b)))' * rho(:,:,x,a) );
                        end
                    end
                end
            end     
            
            maximize real(win)
            
            subject to 
                          
                % Bob's measurements operators must be PSD and sum to I
                B_b_sum = sum(B,4);
                for y = 1:num_inputs
                    B_b_sum(:,:,y) == I;
                end
                B >= 0;                                                           
                             
        cvx_end           
        win = real(win);
        
        it_diff = win - prev_win;
        prev_win = win;
     end
    
    % As the SDPs keep alternating, check if the winning probability
    % becomes any higher. If so, replace with new best.
    if best < win
        
        best = win;

        % take purification of tau
        pur = PartialTrace(tau,2);            

        A = zeros(xdim,ydim,num_inputs,num_outputs); 
        for x = 1:num_inputs
            for a = 1:num_outputs
               A(:,:,x,a) = pur^(-1/2) * PartialTrace(rho(:,:,x,a),2) * pur^(-1/2); 
            end
        end
         
        opt_strat_A = A;
        opt_strat_B = B;
    end

end;
 
best
\end{verbatim}
\color{lightgray} 
\begin{verbatim}     
best =

    0.8536
\end{verbatim}
\color{black}

\subsection{The BB84 monogamy game (Example \ref{ex:bb84-monogamy-game})} \label{code:bb84-game}

\begin{verbatim}
% Create the BB84 basis.
e0 = [1;0]; e1 = [0;1];
ep = [1;1]/sqrt(2); em = [1;-1]/sqrt(2); 

psi0 = e0*e0'; psi1 = e1*e1';
psip = ep*ep'; psim = em*em'; 

% Referee's first basis: {|0><0|, |1><1|}
R{1} = {psi0,psi1};

% Referee's second basis: {|+><+|, |-><-|}
R{2} = {psip,psim};

% BB84 game for a single repetition.
reps = 1; 

% Level of the extended QC hierarchy 
lvl = 1;

% Calculate the lower and upper bounds on the BB84 game:
%   cos^2(pi/8) \approx 0.8536
lb = MonogamyGameValueLB(R,reps,lvl)
ub = MonogamyGameValueUB(R,reps,lvl)
\end{verbatim}\color{lightgray} 
\begin{verbatim}     
lb =

    0.8535
ub =

    0.8535
\end{verbatim}
\color{black}

\subsection{A monogamy-of-entanglement game defined by mutually unbiased bases (Example \ref{ex:mub-4-3-monogamy-game})} \label{code:mub-4-3}

\begin{verbatim}
% Number of inputs and outputs
nin = 4;
nout = 3;

% Create the mutually unbiased bases consisting of 4-inputs and 3-outputs.
m = MUB(nout); 
R = {};
for i = 1:nin
    for j = 1:nout
        R{i}{j} = m{i}(:,j) * m{i}(:,j)';
    end
end

% Number of repetitions of the game. 
reps = 1; 

% Level of the extended QC hierarchy.
lvl = 1;

% Calculate the lower and upper bounds on the quantum value of
% the mutually unbiased basis game:
lb = MonogamyGameValueLB(R,reps,lvl)
ub = MonogamyGameValueUB(R,reps,lvl)
\end{verbatim}
\color{lightgray} 
\begin{verbatim}     
lb =

    0.6610
ub =

    0.6667
\end{verbatim}
\color{black}

\subsection{A counter-example to strong parallel repetition for monogamy-of-entanglement games with non-signaling provers (Proof of Theorem \ref{thm:no-spr-non-signaling})} \label{code:ns-counter-example-monogamy-game}

\begin{verbatim}
% Create the BB84 basis.
e0 = [1;0]; e1 = [0;1];
ep = [1;1]/sqrt(2); em = [1;-1]/sqrt(2); 

psi0 = e0*e0'; psi1 = e1*e1';
psip = ep*ep'; psim = em*em'; 

% Referee's first basis: {|0><0|, |1><1|}
R{1} = {psi0,psi1};

% Referee's second basis: {|+><+|, |-><-|}
R{2} = {psip,psim};

% BB84 game for a single repetition.
reps = 1; 

% Level of the extended QC hierarchy corresponds to non-signaling 
lvl = 0;

% Calculate the lower and upper bounds on the BB84 game:
rep_1_val = MonogamyGameValue(R,reps,lvl)

% BB84 game for a single repetition.
reps = 2; 

% Calculate the lower and upper bounds on the BB84 game:
rep_2_val = MonogamyGameValue(R,reps,lvl)
\end{verbatim}
\color{lightgray} 
\begin{verbatim}     
rep_1_val =

    0.8536

rep_2_val =

    0.7383
\end{verbatim}
\color{black}

\printindex


\end{document}